\documentclass[11pt]{article}

\usepackage{amsfonts, amssymb, latexsym, amsmath, amsthm, mathrsfs,
            verbatim, geometry}
\usepackage{graphics}
\usepackage{slashed}
\usepackage{url,color}
\usepackage{enumerate}

\usepackage{hyperref}
\hypersetup{
    colorlinks,
    citecolor=red,
    filecolor=green,
    linkcolor=blue,
    urlcolor=black
}

\begin{document}

\arraycolsep = 0.3\arraycolsep
\def\R{\mathbb R}
\def\C{\mathbb C}
\def\N{\mathbb N}
\def\Z{\mathbb Z}
\newcommand\J{\mathscr J}
\def\A{\mathscr{A}}
\def\p{\left\langle v\right\rangle}
\def\tp{\left\langle \tilde v\right\rangle}
\def\u{\left\langle u\right\rangle}
\def\T{\mathbb{T}}
\def\energy{\mathcal{E}}
\def\D{\mathcal{D}}
\def\be{\begin{equation}}
\def\ee{\end{equation}}
\def\bea{\begin{eqnarray}}
\def\eea{\end{eqnarray}}
\def\beas{\begin{eqnarray*}}
\def\eeas{\end{eqnarray*}}
\def\supp{\mathrm{supp}\,}
\def\sign{\mathrm{sign}\,}

\def\g{\partial}
\def\t{\bar{\partial}}
\def\a{{\bf a}}
\def\K{\mathbb{S}}
\def\l{\lambda}
\def\pa{\partial }
\def\energynergy{{\overline{\mathcal E}^\varphi}}
\def\w{w_{\delta}}
\def\S{\mathcal S}
\def\der{\pa_r^a\slashed\nabla^\beta}
\def\car{\pa^\nu}
\def\d{\text{div}}
\def\sn{\slashed\nabla}
\def\M{\mathscr M}
\def\N{\mathscr N}
\def\drho{\delta\rho}
\def\dl{\delta\lambda}
\def\bcr{\begin{color}{red}}
\def\bcb{\begin{color}{blue}}
\def\ec{\end{color}}

\def\norm{\mathcal S^N}
\def\vortnorm{\mathcal B^N}
\def\energy{\mathcal E^N}
\def\dissipation{\mathcal D^N}
\def\pa{\partial}
\newcommand{\prfe}{\hspace*{\fill} $\Box$

\smallskip \noindent}

\sloppy
\newtheorem{maintheorem}{Theorem}
\newtheorem{theorem}{Theorem}[section]
\newtheorem{definition}[theorem]{Definition}
\newtheorem{proposition}[theorem]{Proposition}
\newtheorem{example}[theorem]{Example}
\newtheorem{corollary}[theorem]{Corollary}
\newtheorem{lemma}[theorem]{Lemma}
\newtheorem{remark}[theorem]{Remark}

\renewcommand{\theequation}{\arabic{section}.\arabic{equation}}

\title{Stability and instability of self-gravitating relativistic
       matter distributions}
\author{Mahir Had\v zi\'c\thanks{Department of Mathematics,
        University College London, London, UK} ,
        Zhiwu Lin\thanks{Department of Mathematics,
        Georgia Institute of Technology, Atlanta, USA} ,
        Gerhard Rein\thanks{Department of Mathematics,
        University of Bayreuth, Germany}}

\maketitle

\begin{abstract}
  We consider steady state solutions of the massive, asymptotically
  flat, spherically symmetric Einstein-Vlasov system, i.e.,
  relativistic models of galaxies or globular clusters,
  and steady state solutions of the Einstein-Euler system, i.e.,
  relativistic models of stars. Such steady states are embedded into
  one-parameter families parameterized by their central redshift $\kappa>0$.
  We prove their linear instability when $\kappa$
  is sufficiently large, i.e., when they are strongly relativistic,
  and that the instability
  is driven by a growing mode.
  Our work confirms the scenario of dynamic instability proposed in the
  1960s by Zel'dovich \& Podurets (for the Einstein-Vlasov system) and by
  Harrison, Thorne, Wakano, \& Wheeler (for the Einstein-Euler system).
  Our results are in sharp contrast to the
  corresponding non-relativistic, Newtonian setting.

  We carry out a careful analysis of the linearized dynamics around
  the above steady states and prove an exponential trichotomy result and
  the corresponding index theorems for the stable/unstable invariant spaces.
  Finally, in the case of the Einstein-Euler system we prove a rigorous version
  of the turning point principle which relates the stability of steady states
  along the one-parameter family to the
  winding points of the so-called mass-radius curve.
\end{abstract}

\tableofcontents

\section{Introduction and main results}
\setcounter{equation}{0}

We consider a smooth spacetime manifold $M$
equipped with a  Lorentzian metric $g_{\alpha \beta}$
with signature $(-{}+{}+{}+)$.
The Einstein equations read
\begin{equation} \label{feqgen}
G_{\alpha \beta} = 8 \pi T_{\alpha \beta},
\end{equation}
where $G_{\alpha \beta}$ is the Einstein tensor induced by the metric, and
$T_{\alpha \beta}$ is the energy-momentum tensor given by the
matter content of the spacetime; Greek indices run from $0$ to $3$,
and we choose units in which the speed
of light and the gravitational constant are equal to $1$.
To obtain a closed system,
the field equations \eqref{feqgen} have to be supplemented by
evolution equations for the matter and by
the definition of $T_{\alpha \beta}$ in terms of the matter and the metric.
We consider two matter models, namely a collisionless gas as described
by the collisionless Boltzmann or Vlasov equation and an ideal fluid
as described by the Euler equations. This results in the
Einstein-Vlasov and Einstein-Euler system respectively.
We study these systems under the assumption that the
spacetime is spherically symmetric and asymptotically flat,
but we first formulate them in general,
together with our main results.

\subsection{The Einstein-Vlasov system}

In the case of a collisionless gas matter is described
by the number density $f$ of the particles on
phase space.  The world line of a test particle
on $M$ obeys the geodesic equation
\[
\dot x^\alpha = p^\alpha,\
\dot p^\alpha = - \Gamma^\alpha_{\beta \gamma} p^\beta p^\gamma,
\]
where $x^\alpha$ denote general coordinates on $M$,
$p^\alpha$ are the corresponding canonical momenta,
$\Gamma^\alpha_{\beta \gamma}$ are the Christoffel symbols
induced by the metric $g_{\alpha \beta}$,
the dot indicates differentiation with respect to proper time along the
world line of the particle, and the Einstein summation convention is applied.
We assume that all the particles have the same
rest mass $m_0\geq 0$ and move forward in time, i.e.,
their number density $f$ is a non-negative function
supported on the mass shell
\[
PM := \left\{ g_{\alpha \beta} p^\alpha p^\beta = -m_0^2,\ p^\alpha \
\mbox{future pointing} \right\},
\]
a submanifold of the tangent bundle $TM$ of the spacetime manifold $M$
which is invariant under the geodesic flow.
Since we are interested in the massive case $m_0>0$, we normalize $m_0=1$.
But at a crucial point in our analysis the massless case $m_0=0$ will
play an important role.
Letting Latin indices range from $1$ to $3$
we use coordinates $(t,x^a)$ with zero shift which implies that
$g_{0a}=0$. On the mass shell $PM$ the variable $p^0$ then becomes
a function of the remaining variables $(t,x^a,p^b)$:
\[
p^0 = \sqrt{-g^{00}} \sqrt{1+g_{ab}p^a p^b}.
\]
Since the particles move like test particles in the given metric,
their number density
$f=f(t,x^a,p^b)$ is constant along the geodesics and hence satisfies
the Vlasov equation
\begin{equation} \label{vlgen}
\partial_t f + \frac{p^a}{p^0}\,\partial_{x^a} f
-\frac{1}{p^0}\,\Gamma^a_{\beta \gamma} p^\beta p^\gamma\,\partial_{p^a} f = 0.
\end{equation}
The energy-momentum tensor is given by
\begin{equation} \label{emtvlgen}
T_{\alpha \beta}
=\int p_\alpha p_\beta f \,|g|^{1/2} \,\frac{dp^1 dp^2 dp^3}{-p_0},
\end{equation}
where $|g|$ denotes the modulus of the determinant of the metric,
and indices are raised and lowered using the metric, i.e.,
$p_\alpha = g_{\alpha \beta}p^\beta$.
The system \eqref{feqgen}, \eqref{vlgen}, \eqref{emtvlgen}
is the Einstein-Vlasov system in general coordinates. We
want to describe isolated systems, and therefore we require that the
spacetime is asymptotically flat.

In Section~\ref{sssection} we will see that steady states of
this system can be obtained via an ansatz
\be \label{miceqstate}
f = \phi(E)
\ee
where $E= -p_0$ is the local or particle energy and
$\phi$ is a prescribed ansatz function;
we refer to \eqref{miceqstate}
as the microscopic equation of state.
We can now state our main result for the Einstein-Vlasov system in an informal
way; the precise meaning of the parameter
$\kappa>0$ is explained in Section~\ref{sssection},
in particular \eqref{yiv} and \eqref{redshiftdef}.

\begin{maintheorem} \label{mainev}
  Let $(f_\kappa)_{\kappa>0}$ be a one-parameter
  family of spherically symmetric steady states to the
  Einstein-Vlasov system with a fixed, decreasing microscopic equation of state
  $\phi$ and $\kappa$ the central redshift of the corresponding steady state.
  \begin{itemize}
  \item[(a)]
    For $\kappa$ sufficiently large, the associated steady state
    is dynamically unstable in the sense that the linearized Einstein-Vlasov
    system possesses an exponentially growing mode.
  \item[(b)]
    For any $\kappa>0$ the phase space of the linearized system
    splits into three invariant subspaces:
    the stable, unstable, and center space.
    The dimension of the stable/unstable space is equal to the
    negative Morse index of
    a suitably defined, macroscopic Schr\"odinger-type operator.
  \end{itemize}
  \end{maintheorem}

The Einstein-Vlasov system models large stellar systems such as
galaxies or globular clusters,
and the central redshift is a measure of how relativistic the corresponding
state is. Part~(a) of the theorem says
that strongly relativistic galaxies or globular clusters are unstable.
This is in sharp contrast to the corresponding Newtonian situation,
i.e., to the Vlasov-Poisson system. For this system any steady state
induced by a strictly decreasing microscopic equation of state is non-linearly
stable, cf.\ \cite{GuRe2001,GuRe2007,LeMeRa3,Mou}.
The instability result in Theorem~\ref{mainev} is therefore a genuinely
relativistic phenomenon.

A proof of instability which does not rely on the refined spectral
information about the existence of growing modes can be found in
Theorem~\ref{lininstabEV}, while the existence of a growing mode
is given by Theorem~\ref{thm:main VE2}. A rigorous statement
of exponential trichotomy is provided in Theorem~\ref{thm:main VE}.

\subsection{The Einstein-Euler system}

In this case the matter is described by its mass-energy density $\rho$,
its four velocity $u^\alpha$ which is a future-pointing, time-like unit
vector field, and its pressure $p$. These quantities are defined on the
spacetime manifold $M$ and induce the energy-momentum tensor
\be
T_{\alpha \beta} = (\rho+p) u_\alpha u_\beta + p\, g_{\alpha \beta}.
\label{emtegen}
\ee
The Bianchi identity applied to the field equations yields the evolution
equations for the fluid
\[
\nabla^\alpha T_{\alpha \beta} = 0,
\]
where $\nabla^\alpha$ denotes the covariant derivative associated with the
metric. More explicitly,
\be \label{euler1}
u^\alpha \nabla_\alpha \rho + (\rho + p)  \nabla^\alpha u_\alpha =0,
\ee
\be \label{euler2}
(\rho + p) u^\alpha \nabla_\alpha u_\beta  +
(g_{\alpha\beta} + u_\alpha u_\beta) \nabla^\alpha p =0.
\ee
In order to close the system we need to prescribe a (macroscopic) equation
of state which relates $p$ and $\rho$,
\be \label{eqstate}
p = P(\rho),
\ee
with a prescribed function $P$. Notice that $p$ will always denote the pressure
as a function on spacetime, while $P$ will denote its functional relation to
the mass-energy density $\rho$. We state our
main result for the Einstein-Euler system,
which consists of
\eqref{feqgen}, \eqref{emtegen}, \eqref{euler1}, \eqref{euler2},
\eqref{eqstate},
in an informal way.

\begin{maintheorem} \label{mainee}
  Let $(\rho_\kappa, p_\kappa, u^\alpha_\kappa\equiv 0)_{\kappa>0}$
  be a one-parameter family of spherically symmetric steady states to the
  Einstein-Euler system with a fixed, strictly increasing
  equation of state (satisfying suitable assumptions) and $\kappa$ the
  central redshift of the corresponding steady state.
  \begin{itemize}
  \item[(a)]
    For $\kappa$ sufficiently large,
    the associated steady state is dynamically unstable in the sense that the
    linearized Einstein-Euler system possesses  an exponentially growing mode.
  \item[(b)]
    For any $\kappa>0$, the phase space of the linearized system
    splits into three invariant subspaces:
    the stable, unstable, and center space.
    The dimension of the stable/unstable space is equal to the
    negative Morse index of
    a suitably defined Schr\"odinger-type operator $\Sigma_\kappa$.
  \item[(c)]
    A version of the turning point principle holds, i.e., as $\kappa\in ]0,\infty[$ varies,
    the stability of the steady state can be inferred from the so-called
    mass-radius diagram and the knowledge of the negative Morse index
    of $\Sigma_\kappa$.
  \end{itemize}
\end{maintheorem}

The Einstein-Euler system models stars, and hence part (a) of the theorem says
that strongly relativistic steady stars---those with very
large central redshift---are unstable. The analogous comment as in the Vlasov
case applies concerning the Newtonian situation.
The rigorous statement and proof of the turning point principle is given in
Theorem~\ref{T:TPP}, the exponential trichotomy is shown in
Theorem~\ref{thm:main EE}, and the instability for large $\kappa$
can be found in Corollary~\ref{C:UNSTABLESTARS}.

\subsection{Motivation and history of the problem}

The Newtonian analogue of the Einstein-Vlasov system is the
gravitational Vlasov-Poisson system
\be \label{E:VLASOV}
\pa_t f+ v\cdot \nabla_x f - \nabla_x U\cdot\nabla_v f = 0,
\ee
\be \label{E:POISSON}
\Delta U = 4\pi \rho,\quad \lim_{|x|\to\infty} U(t,x) = 0,
\ee
\be \label{E:SPATIALDENSITY}
\rho(t,x) = \int f(t,x,v)\,dv.
\ee
Here $f=f(t,x,v)\geq 0$,
a function of time $t$, position $x\in {\mathbb R^3}$, and
velocity $v\in {\mathbb R^3}$, is the phase-space density
of the stars in a galaxy, and $U =U(t,x)$
is the gravitational potential induced by the macroscopic, spatial mass
density $\rho=\rho(t,x)$; integrals without explicitly specified domain
always extend over $\R^3$.
A convenient way of constructing steady states of this system is
to make an ansatz
\[
f(x,v)=\Phi(1 -\frac{E(x,v)}{E^0})
\]
where
\[
E(x,v)=\frac12 |v|^2+U(x)
\]
is the local particle energy.
Modulo some technical assumptions the profile
$\Phi\geq 0$, which we refer to as the microscopic equation of state, is
arbitrary, but must vanish on $]-\infty,0[$. Hence
$E^0<0$ represents the maximal particle energy allowed in the system.
Any such choice of $f$ satisfies \eqref{E:VLASOV} with the given
potential $U$, and the problem of finding a steady state is reduced
to solving~\eqref{E:POISSON}, where the right hand side
becomes a function of $U$ obtained by substituting the ansatz into
\eqref{E:SPATIALDENSITY}. Since such steady states necessarily
are spherically symmetric, it seems convenient to parameterize them
by $U(0)$. But in view of the boundary condition in \eqref{E:POISSON}
and the need of a cut-off energy $E^0$ in the ansatz this seems one
free parameter too many. Instead, one can use $y=E^0-U$ as the basic unknown
function and use $\kappa = y(0)>0$ as the free parameter.
Under suitable assumptions on $\Phi$ the function $y$ has a unique zero
at some radius $R>0$ which marks the boundary of the support of the
induced steady state which then also has finite mass,
cf.~\cite{RaRe} and the references there.
With $E^0=\lim_{r\to\infty}y(r)<0$
the corresponding potential $U$ vanishes at infinity,
and the parameter $\kappa=y(0) = U(R)-U(0)$ is the potential energy difference
between the center and the boundary of the compactly supported steady state.
Examples for ansatz functions $\Phi$ which yield such one-parameter families
of steady states are the polytropes where $\Phi(\eta)=\eta^k$ for $\eta>0$
with $-1<k<7/2$,
and the King model where $\Phi(\eta)=e^\eta$,
both of which are used in astrophysics, cf.~\cite{BiTr} .

A central challenge in the qualitative description of the Vlasov-Poisson
dynamics is to understand the behavior of solutions close to the above
steady states.
In the astrophysics literature formal arguments towards linearized
stability were given for example by Antonov~\cite{An1961},
Doremus, Feix, Baumann~\cite{DoFeBa}, and Kandrup, Sygnet~\cite{KS}.
The stability criterion proposed by these authors is that
$\Phi$ is strictly increasing on $[0,\infty[$.
This is physically reasonable, as it states that the number of stars
in the galaxy is a decreasing function of their energy.
Much work has been invested in a mathematically rigorous, nonlinear proof of
this stability criterion, cf.~\cite{GuRe2001,GuRe2007,GuoLin, LeMeRa2,LeMeRa3}
and the references there.
Remarkably, the size of $\kappa$ is irrelevant for these stability results.

Theorem~\ref{mainev} shows that the above stability criterion for Vlasov-Poisson
is {\em false} for the Einstein-Vlasov system: When the central
redshift $\kappa$ is sufficiently large, there exists a growing
mode despite the monotonicity assumption on $\Phi$.
This fundamental difference between the Vlasov-Poisson and
the Einstein-Vlasov system is driven by relativistic effects
which become dominant at sufficiently large values of the
central redshift $\kappa$. Another difference between these two systems
seems to be that in the non-relativistic case non-linearly stable steady
states can be obtained as minimizers of suitable energy-Casimir
functionals, cf.~\cite{Rein07} and the references there, but this
does not work so well in the relativistic case, as the problem is supercritical;
in \cite{Wo} an attempt in this direction has been made,
but as shown in \cite{AndKun} that paper is not correct. We also mention an interesting recent work
of Andersson and Korzy\'nski~\cite{AnKo} on the variational derivation of the Einstein-Vlasov system.

The Newtonian analogue of the Einstein-Euler system is the Euler-Poisson system,
where the Euler equations
\[
\pa_t \rho+ \text{div}(\rho\,{\bf u}) = 0,
\]
\[
\rho\left(\pa_t+{\bf u}\cdot \nabla\right){\bf u}
+ \nabla p + \rho\nabla U = 0
\]
are coupled to the Poisson equation \eqref{E:POISSON}.
Here $\rho$ is the macroscopic fluid density, and ${\bf u}$ is the
fluid $3$-velocity.
To close this system, one must prescribe an equation of state
$p = P (\rho)$.
A well-investigated choice is again the polytropic one,
i.e., $P(\rho)=\rho^\gamma$, $\gamma>1$.
When $\gamma>\frac65$ there exist spherically symmetric stationary
distributions of compact support, cf.~\cite{BiTr,RaRe}.
A classical, linear stability analysis shows that polytropic steady
stars are stable if $\gamma>\frac43$. By analogy to the Vlasov-Poisson case
this result holds independently of the size of the central density $\rho(0)$
of the steady state. Theorem~\ref{mainee} shows that this
changes drastically in the relativistic context:
Strongly relativistic stars, i.e.,
stars with large central density or equivalently large central redshift,
are unstable.

The subject of stability of isolated self-gravitating solutions of the
Einstein-Euler system has a long history. The first linear stability
result and a variational characterization of the stability question in
spherical symmetry goes back to the pioneering work of
Chandrasekhar~\cite{Chandrasekhar1964} in 1964. The same stability criterion
was then derived by Harrison, Thorne, Wakano, and Wheeler
in 1965~\cite{HaThWaWh}, who were the first to show that Chandrasekhar's
result is equivalent to the positive-definiteness of the second variation
of the ADM mass, subject to the constant baryon number constraint, see also~\cite{BaThMe1966}.
Moreover, a turning point principle is formulated in~\cite{HaThWaWh} and
investigated numerically. It proposes that along the one-parameter family
stability changes to instability or
vice versa when the steady state passes through an extremal point
of the so-called mass-radius curve. Recently, a proof of such a turning point principle was given in~\cite{HaLin}. For a detailed survey of these and
related results from the 1960s see~\cite{Th1966}. A variational approach to stability and the turning point
principle are elegantly formulated in~\cite{Ca1970a,Ca1970b,Ca1971}. 
In Theorem~\ref{T:TPP}
we formulate and prove a rigorous version of the
turning point principle for the Einstein-Euler system,
i.e., of part~(c) of Theorem~\ref{mainee}. A general abstract discussion of
turning point principle is astrophysics is given in~\cite{So1981}.
For the different notion of {\em thermodynamic} instabilities and its
relation to turning point principles,
cf.\ the discussions in~\cite{KaHoKl1975,ScWa2014}.
We refer the reader to~\cite{Gl2000, Straumann} for a comprehensive
overview of this topic from the physics point of view.
By contrast to what is known in the vicinity of massive steady states,
the vacuum solution of the Einstein-Vlasov system is asymptotically
non-linearly stable against small
perturbations~\cite{FaJoSm,LiTa,RR92a}.

A burst of interest in the stability of highly relativistic,
self-gravitating bodies occurred in the mid 1960s with the discovery of
quasars~\cite{BKZ,ZeNo}. Since these objects are very powerful
and very concentrated sources of energy, it was initially unclear whether
their high redshifts were due to them being far away or due to them being very
relativistic. Zel'dovich and Podurets gave an early contribution to the
subject by showing numerically that certain spherically symmetric steady
states turn unstable when their central redshift passes a threshold
value~\cite{ZePo}. This lead Ipser, Thorne, Fackerell and others to
initiate a systematic study of the question of stability of steady
relativistic
galaxies~\cite{FaIpTh,Fa1970,FaSu1976a,FaSu1976b,IT68,IP1969,IP1969b,IP1980}.
In particular, a turning point principle for the Einstein-Vlasov system was
formulated.
It states that the transition from a stable to an unstable configuration
occurs at a critical point of the binding energy, plotted as a function
of the central redshift. For some numerical results in this direction
see~\cite{RaShTe1989,ShTe1985}.
On the other hand, Bisnovatyi-Kogan and Zel'dovich
pursued the question whether there exist {\em stable} self-gravitating
configurations with arbitrarily high central redshifts \cite{BKZ,BKThorne1970}.
The steady states studied in~\cite{BKZ} are singular
and have infinite mass and extent. However, in Section~\ref{sssection}
we shall see that certain explicit steady states of the massless
Einstein-Vlasov system have the same macroscopic properties
as the solutions in \cite{BKZ}. These massless solutions, which we
refer to as the BKZ solutions, capture the behavior of massive steady states
close to the center of the galaxy when the central redshift is sufficiently
large. An analogous assertion holds for the Einstein-Euler system,
see Proposition~\ref{P:EVTOBKZ}.
This observation appears to be new and of independent interest,
but it also plays a fundamental role in our proof of instability in part (a) of
Theorems~\ref{mainev} and~\ref{mainee}.

In the context of relativistic stars,
Meltzer and Thorne~\cite{MeltzerThorne1966}
discuss the stability properties of very dense relativistic stars. They
predict that perturbations localized to the core of a very dense  star
will result in gravitational collapse, which is consistent with our
methodology, see Sections~\ref{SS:NEGATIVE} and~\ref{SS:NEGATIVEEE}.

\subsection{Methodology}

Under natural assumptions on a given equation of state $f = \phi (E)$ or
$p=P(\rho)$, both the Einstein-Vlasov and the Einstein-Euler systems
posses a corresponding
one-parameter family of compactly
supported steady states with finite mass, parameterized by the value of
the central density $\rho(0)$, or equivalently,
by the central redshift $\kappa>0$, cf.~Section~\ref{sssection},
in particular \eqref{yiv} and \eqref{redshiftdef}.
Even though the equations of equilibrium are classical~\cite{OpVo1939},
a rigorous proof of the existence and finiteness of the total mass and
support of a steady star/galaxy is nontrivial and depends crucially on
the assumptions on the micro- and macroscopic equations of state,
see~\cite{RaRe} and references therein.
In what follows $M_\kappa$ denotes the ADM mass and $R_\kappa$ the
radius of the support of the corresponding steady state.

In both cases, the formal linearization around the steady state leads to
a Hamiltonian partial differential equation, which comes with a
rich geometric structure. In the case of the Einstein-Vlasov  system
the steady state can be interpreted as a critical point of the
so-called energy-Casimir functional~\cite{HaRe2013,KM}.
To prove the instability for large values of $\kappa>0$ it is
therefore natural to investigate the sign of the second variation
of this energy-Casimir.
In Theorem~\ref{T:NEGATIVEA} we
construct an explicit test function that turns the second variation
negative for $\kappa>0$ sufficiently large.
The key to the construction is a precise understanding
of the limiting behavior of the sequence of steady states
$(f_\kappa)_{\kappa>0}$ in the singular limit $\kappa\to\infty$.
We show that in a suitably rescaled annulus around the center $r=0$
the behavior of the steady states is at the leading order described
by the BKZ solutions mentioned above, see
Sections~\ref{SS:BKZ} and \ref{SS:LARGEKAPPA}.
In the physics literature such a limit is referred to as
the ultra-relativistic limit.  Since the BKZ solutions are known explicitly,
we obtain sharp a priori
control over the second variation in the large $\kappa$ regime,
and we can make a judicious choice of a test function whose support
is spatially localized to the aforementioned annulus,
close to the center of the galaxy. In Theorem~\ref{T:NEGATIVEA} we carry out
careful estimates showing that such a test function makes
the second variation functional negative.
A similar procedure applies to the Einstein-Euler system,
see Theorem~\ref{T:NEGATIVEA2}.

The existence of a negative energy direction is sufficient for a proof of
linear exponential instability as shown in
Theorem~\ref{lininstabEV} for the Einstein-Vlasov system; a similar proof
can be given for the Einstein-Euler system. To this end we adapt
a strategy developed in the study of plasma instabilities by
Laval, Mercier, and Pellat~\cite{LaMePe}.
However, with nonlinear applications in mind it is important
to show that the instability is driven by a growing mode,
a statement not afforded by Theorem~\ref{lininstabEV}. In the context of
the Vlasov theory, this is a highly nontrivial question, as the presence of
the continuous spectrum is generally unavoidable, and a more refined analysis
is necessary. Thus to get a more precise spectral information, we must
carefully analyze the linearized operator. For the Euler-Einstein system
Makino showed that the spectrum of the linearized operator is in fact
discrete~\cite{Ma1998,Ma2016} by formulating it as a singular
Sturm-Liouville type operator on a bounded domain.

Both the linearized Einstein-Euler and the linearized Einstein-Vlasov system
can be written in the second order form
\[
\pa_{tt} f  + L_\kappa f = 0, \ \ f\in X,
\]
where $L_\kappa \colon X \supset D(L_\kappa) \to X$ a  self-adjoint  linear operator,
and $X$ a Hilbert space. In both cases, it is shown that the second
variation of the energy corresponds to to the quadratic form $(L_\kappa f,f)_X$.
In the case of the Einstein-Vlasov system, a naive attempt to minimize the
functional
\[
f\mapsto (L_\kappa f,f)_X
\]
leads to difficulties due to the possible loss of compactness along minimizing
sequences. A related loss of compactness has been well-known in the stability
theory for various Vlasov-matter systems, most notably the Vlasov-Poisson and
the Vlasov-Maxwell system~\cite{LinStrauss2008}. We adopt a different approach.

{\em The reduced operator.}
The key step in our analysis is to construct a suitable reduced operator
which by definition is a self-dual, macroscopic, Schr\"odinger-type operator
$S_\kappa\colon  \dot H^1(\mathbb R^3)\to \dot H^1(\mathbb R^3)' $ such that for any
$f\in D(L_\kappa)$ there exists an element $\psi\in  \dot H^1(\mathbb R^3)$ such that
\be\label{E:REDUCEDINTRO}
(L_\kappa f,f )_X \ge  \langle S_\kappa \psi,\psi\rangle,
\ee
and for any $\psi \in \dot H^1(\mathbb R^3)$ there exists an $f\in D(L_\kappa)$ such that
\be\label{E:REDUCEDINTRO2}
\langle S_\kappa \psi,\psi\rangle \ge (L_\kappa f,f )_X.
\ee
Here $S_\kappa = -\Delta_\kappa + V_\kappa$, where $\Delta_\kappa$ is an
explicit, nondegenerate, linear, elliptic operator with variable coefficients,
$V_\kappa$ a compactly supported potential, and $\dot H^1(\mathbb R^3)$ is a homogeneous
Sobolev space.
Since $S_\kappa$ has at most finitely many negative eigenvalues, we can use
the bounds~\eqref{E:REDUCEDINTRO}, \eqref{E:REDUCEDINTRO2} to conclude
that the negative Morse index of $L_\kappa$ is finite.
The derivation of $S_\kappa$ in both the kinetic and the fluid case is
{\em new}.

An attempt to construct a reduced operator was made earlier by
Ipser~\cite{IP1969}, but in that work the bound~\eqref{E:REDUCEDINTRO}
with a different choice of $S_\kappa$ is satisfied only under an additional
hypothesis on the steady states, which appears to be hard to verify
rigorously. We point out that a related construction of the reduced operators
for the Newtonian analogues is nontrivial, yet  considerably simpler;
for the Vlasov-Poisson system see~\cite{GuoLin}
and for the Euler-Poisson system~\cite{LinZeng2017b}.

While the reduced operator can be constructed for any $\kappa>0$,
see Theorems~\ref{T:VLASOVREDUCED} and~\ref{T:EULERREDUCED}, only
in the regime of large $\kappa$ we know that there exists at
least one negative direction. By proving that $L_\kappa$ is self-adjoint
and not merely symmetric, we then infer that the unstable space consists
of eigenfunctions of finite multiplicity. The proof of self-adjointness
of the operator $L_\kappa$ in both the Einstein-Vlasov and the Einstein-Euler
case is not obvious and is presented in Lemmas~\ref{L:SECONDORDEREV}
and~\ref{prop-self-adjoint-composition} respectively.

{\em Exponential trichotomy.}
To obtain further information about the linearized dynamics, we are forced
to work with the first order Hamiltonian formulation of the problem.
This part of our analysis applies to all $\kappa>0$.
Abstractly, both problems can be recast in the form
\be\label{E:FIRSTORDERINTRO}
\pa_t {\bf u} = \mathscr J_\kappa \mathscr L_\kappa {\bf u},
\ee
where $\mathscr L_\kappa \colon \mathcal X \to \mathcal X'$ is a self-dual
operator and $\mathscr J_\kappa \colon \mathcal X'\to \mathcal X$ is
anti-self-dual; $\mathcal X$ is an appropriate Hilbert space.
The operators $\mathscr J_\kappa, \mathscr L_\kappa$ for the
Einstein-Vlasov system are given in Lemma~\ref{L:EVLIN1}, for the
Einstein-Euler system in Lemma~\ref{L:LINEAREE}.
Using crucially the existence of the reduced operator
and~\eqref{E:FIRSTORDERINTRO} we are able to apply the general
framework developed by Lin and Zeng~\cite{LinZeng2017} to obtain a
quantitative exponential trichotomy result, see Theorems~\ref{thm:main VE}
and~\ref{thm:main EE}. Roughly speaking, we show that the phase space
around the steady state naturally splits into a stable, unstable, and
center invariant subspace. Under a non-degeneracy assumption on
the reduced operator $S_\kappa$, we can go a step further and prove a
quantitative Lyapunov stability statement on the center space in the natural
energy topology, see parts (v) of
Theorems~\ref{thm:main VE} and~\ref{thm:main EE}.

The exponential trichotomy statement shall provide a foundation for the
understanding of the nonlinear dynamics in the vicinity of steady states.
For instance, one expects that in the presence of growing modes the
phase space will split into perturbations leading to gravitational
collapse, and those leading to global existence via dispersion,
separated by an invariant (co-dimension 1 center stable) manifold. This is consistent
with the dynamical picture proposed in the study of
criticality phenomena~\cite{AnRe1,GaGu}.

{\em Turning point principle.}
In the context of self-gravitating relativistic bodies, it is desirable to
have a simple criterion for stability
that depends on certain bulk properties of the system under study. Precisely
such an idea was put forward by Zel'dovich~\cite{Ze1963} and
Harrison, Thorne, Wakano, and Wheeler in~\cite{HaThWaWh}, wherein the
so-called turning point principle was formulated. If one plots the mass
$M_\kappa$ and the steady star radius $R_\kappa$ along a curve parameterized
by the central redshift $\kappa$, it is proposed in~\cite{HaThWaWh}
that the stability can be lost to instability and vice versa only at the
extremal points of the mass plotted as a function of $\kappa$ ,
see also~\cite{ScWa2014,Straumann}.

The steady states of the Einstein-Euler system can be interpreted as
critical points of the ADM mass among all densities of the same total
baryon mass. This observation goes back to~\cite{HaThWaWh}. By computing
the second variation of the ADM mass we can derive the well-known
Chandrasekhar stability criterion~\cite{Chandrasekhar1964},
which states that the static star is (spectrally)
stable if the linearized operator is non-negative on a codimension-1
subspace of a certain Hilbert space, see Proposition~\ref{P:CHANDRASEKHAR}.
Using this characterization, we show in Theorem~\ref{T:TPP} that when $\frac{d}{d\kappa} M_\kappa \frac{d}{d\kappa}\left(\frac{M_\kappa}{R_\kappa}\right)\ne 0$
the number of growing modes associated with the linearized operator
$\mathscr J_\kappa \mathscr L_\kappa$ equals $n^-(\Sigma_\kappa)-i_\kappa$,
where $\Sigma_\kappa$ is the reduced operator associated with the
Einstein-Euler system, $n^-(\Sigma_\kappa)$ is its negative Morse index, and
the index $i_\kappa\in\{0,1\}$ depends  on the
mass-radius curve  through the formula
\[
i_\kappa =
  \begin{cases}
    1&\text{if }
    \frac{d}{d\kappa} M_\kappa \frac{d}{d\kappa}\left(\frac{M_\kappa}{R_\kappa}\right) >0, \\
    0&\text{if }
    \frac{d}{d\kappa} M_\kappa\frac{d}{d\kappa}\left(\frac{M_\kappa}{R_\kappa}\right) <0.
  \end{cases}
\]
Finally, in line with Antonov's First Law for Newtonian
self-gravitating systems, we prove a relativistic ``micro-macro"
stability principle in Theorem~\ref{T:MICROMACRO}, which in a precise way
states that a steady galaxy with a certain microscopic equation of state
is spectrally stable if an individual star
with the corresponding macroscopic equation of state
is spectrally stable.

{\em Some open questions.}
A natural question for further study is the {\em nonlinear instability}
of steady states when $\kappa\gg1$.
For the Einstein-Vlasov case a local well-posedness result was established
in \cite{RR92a} which can be used (or if necessary adapted) for initial data
close to a steady state. A non-trivial problem which needs to be
understood before attacking the non-linear regime is the regularity
of eigenfunctions corresponding to the growing modes of the linearized system.
For the Einstein-Euler case we are naturally led
to the vacuum free boundary problem wherein the location and the regularity
of the boundary separating the star from the vacuum is an
unknown. This question comes with a number of mathematical difficulties, and
even the local-in-time well-posedness theory is an open problem in this context.

When $\kappa\ll1$ it is known~\cite{HaRe2013, HaRe2014} that the steady states
of the Einstein-Vlasov system are linearly stable. Since the Einstein
equations are energy supercritical, this type of a priori control coming
from the linearized problem is far from sufficient for proving any kind of
nonlinear stability; note that the situation is different for the
Vlasov-Poisson system. It is an open problem to show that nonlinear orbital
stability is true for $\kappa\ll1$.
A linear stability result for so-called hard stars is given
in~\cite{FoSc}.

A further open question is to show, as conjectured in the physics
literature~\cite{HaThWaWh, Straumann}, that for the Einstein-Euler system
the number of unstable directions grows to infinity as $\kappa\to\infty$.
As pointed out in Remark~\ref{R:GROWINGMODES}
there is strong evidence for this to be true.

\section{The Einstein-Vlasov and Einstein-Euler system in spherical symmetry}
\setcounter{equation}{0}

For the above systems, questions like the stability or instability of
steady states are at present out of reach
of a rigorous mathematical treatment, unless simplifying symmetry
assumptions are made. We assume spherical symmetry, use
Schwarzschild coordinates $(t,r,\theta,\varphi)$, and write
the metric in the form
\be \label{metric}
ds^2=-e^{2\mu(t,r)}dt^2+e^{2\lambda(t,r)}dr^2+
r^2(d\theta^2+\sin^2\theta\,d\varphi^2).
\ee
Here $t\in \R$ is a time coordinate,
and the polar angles $\theta\in [0,\pi]$ and $\varphi\in [0,2\pi]$
coordinatize the surfaces of constant $t$ and $r>0$.
The latter are the orbits
of $\mathrm {SO}(3)$, which acts isometrically on this spacetime, and
$4 \pi r^2$ is the area of these surfaces. The boundary condition
\begin{equation}\label{boundcinf}
\lim_{r\to\infty}\lambda(t, r)=\lim_{r\to\infty}\mu(t, r)=0.
\end{equation}
guarantees asymptotic flatness, and in order to guarantee a regular center
we impose the boundary condition
\begin{equation}\label{boundc0}
\lambda(t, 0)=0.
\end{equation}
It is convenient to introduce the corresponding Cartesian coordinates
\[
x = (x^1,x^2,x^3) =
r (\sin \theta \cos \varphi,\sin \theta \sin \varphi,
\cos \theta).
\]
Instead of  the corresponding canonical momenta $p=(p^1,p^2,p^3)$
we use the non-canonical momentum variables
\[
v^a = p^a + (e^\lambda -1)\frac{x\cdot p}{r} \, \frac{x^a}{r},\ a=1,2,3.
\]
In these variables,
\be \label{pdef}
p_0 = - e^\mu \sqrt{1+|v|^2} =: - e^\mu \p,
\ee
and $f$ is spherically symmetric iff
\[
f(t,x,v) = f(t,Ax,Av),\ x, v \in \R^3,\ A \in \mathrm{SO}\,(3).
\]
The spherically symmetric, asymptotically flat Einstein-Vlasov system
takes the following form:
\be \label{vlasov}
\partial_t f + e^{\mu - \lambda}\frac{v}{\p}\cdot \partial_x f -
\left( \dot \lambda \frac{x\cdot v}{r} + e^{\mu - \lambda} \mu'
\p \right) \frac{x}{r} \cdot \partial_v f =0,
\ee
\be
e^{-2\lambda} (2 r \lambda' -1) +1
=
8\pi r^2 \rho , \label{ein1}
\ee
\be
e^{-2\lambda} (2 r \mu' +1) -1
=
8\pi r^2 p, \label{ein2}
\ee
\be
\dot\lambda =
- 4 \pi r e^{\lambda + \mu} j, \label{ein3}
\ee
\be
e^{- 2 \lambda} \left(\mu'' + (\mu' - \lambda')
(\mu' + \frac{1}{r})\right)
- e^{-2\mu}\left(\ddot\lambda +
\dot\lambda \, (\dot\lambda - \dot\mu)\right)
=
8 \pi p_T, \label{ein4}
\ee
where
\bea
\rho(t,r)
&=&
\rho(t,x) = \int \p f(t,x,v)\,dv ,\label{r}\\
p(t,r)
&=&
p(t,x) = \int \left(\frac{x\cdot v}{r}\right)^2
 f(t,x,v)\frac{dv}{\p}, \label{p}\\
j(t,r)
&=&
j(t,x) = \int \frac{x\cdot v}{r} f(t,x,v) dv, \label{j}\\
p_T(t,r)
&=&
p_T(t,x) = \frac{1}{2} \int \left|{\frac{x\times v}{r}}\right|^2
f(t,x,v) \frac{dv}{\p}. \label{p_T}
\eea
Here $\dot{}$ and ${}'$ denote the derivatives with respect to
$t$ and $r$ respectively.
For a detailed derivation of these
equations we refer to \cite{Rein95}.
It should be noted that in this formulation no raising and
lowering of indices using the metric appears anywhere.
It is a completely explicit system of partial differential equations
where $x,v \in \R^3$,
$x\cdot v$ denotes the Euclidean scalar product,
$|v|^2 = v\cdot v$, and $\p$ is defined in \eqref{pdef};
integrals without explicitly specified domain
extend over $\R^3$.

We note that $p$ and $p_T$ are the pressure in the radial and
tangential directions respectively, which for the Vlasov matter model
in general are not equal. They are equal for the isotropic steady states
which we consider in the next section, and of course also for the
Euler matter model.

We now formulate the spherically symmetric Einstein-Euler system
in Schwarzschild coordinates, i.e., for a metric of the form \eqref{metric},
where we keep the boundary conditions
\eqref{boundcinf} and \eqref{boundc0}.
The spherically symmetric matter quantities
$\rho=\rho(t,r)$, $p=p(t,r)$, $u=u(t,r)$ are scalar functions,
and the four velocity is $u^\alpha = (u^0,u,0,0)$ where
\[
u^0 = e^{-\mu}\sqrt{1+e^{2\lambda} u^2} =: e^{-\mu} \u;
\]
here $u^2$ is a power and not to be confused with a component of
the four velocity; no relativity index notation appears from this point on.
The field equations become
\be
e^{-2\lambda} (2 r \lambda' -1) +1 =
8\pi r^2 \left(\rho + e^{2\lambda}(\rho + p) u^2\right),\label{eelambda}
\ee
\be
e^{-2\lambda} (2 r \mu' +1) -1 =
8\pi r^2 \left(p + e^{2\lambda}(\rho + p) u^2\right),\label{eemu}
\ee
\be
\dot \lambda =
- 4 \pi  r e^{\mu + 2\lambda}  \u \, u\, (\rho + p)
\label{eelambdad},
\ee
\be
e^{- 2 \lambda} \left(\mu'' + (\mu' - \lambda')(\mu' + \frac{1}{r})\right)
- e^{-2\mu}\left(\ddot \lambda + \dot\lambda (\dot \lambda - \dot \mu)\right)
= 8 \pi p. \label{ee2ndo}
\ee
The Euler equations become
\bea
\dot \rho + e^{\mu} \frac{u}{\u} \rho'
+ \left(\rho + p\right)
\biggl[\dot \lambda + e^\mu \frac{u}{\u}
\left(\lambda' + \mu'+ \frac{2}{r}\right)
+ e^\mu \frac{u'}{\u}
+ e^{2\lambda} \frac{u}{\u}
\frac{\dot \lambda u + \dot u}{\u}\biggr]
&=&
0,\qquad \label{rhod}\\
(\rho +p) \left[e^{2\lambda} \left( \dot u + 2 \dot \lambda u\right)
  + e^{\mu} \u \mu'
+ e^{\mu+2\lambda} \frac{u}{\u}
\left(u' + \lambda' u\right)\right]
+ e^{\mu}\u p' + e^{2\lambda} u \,\dot p
&=&
0.\qquad \label{ud}
\eea
One should keep in mind that \eqref{eqstate} holds,
which should be used to express $p$ and its derivatives above.

\section{Steady states} \label{sssection}
\setcounter{equation}{0}
\subsection{The basic set-up; one-parameter families}\label{setup}
\noindent
We first consider the Vlasov case.
The characteristic system of the stationary, spherically symmetric
Vlasov equation reads
\[
\dot x = \frac{v}{\p},\
\dot v = - \p \mu'(r) \frac{x}{r},
\]
and the particle energy
\be \label{parten}
E=E(x,v)=e^{\mu(x)}\p=e^{\mu(x)}\sqrt{1+|v|^2}
\ee
is constant along characteristics. Hence the static
Vlasov equation is satisfied if $f$ is taken to be a
function of the particle energy, and the following
form of this ansatz is convenient:
\be \label{trueansatz}
f(x,v) = \phi(E) = \Phi\left(1-\frac{E}{E^0}\right).
\ee
Here $E^0>0$ is a prescribed cut-off energy---notice that
the particle energy (\ref{parten}) is always positive.
Because of spherical symmetry, the quantity
\[
L=|x\times v|^2
\]
is conserved along characteristics as well so that we could include a dependence
on $L$ in the ansatz~\eqref{trueansatz}.
This leads to anisotropic steady states which cannot be treated
in parallel with the Euler case and are not pursued in this paper.
We require that $\Phi$ has the following properties;
for the stability analysis in Section~\ref{sec_stab_EV}
this will be strengthened by ($\Phi$2):

\smallskip

\noindent
{\bf Assumption ($\Phi$1).}
$\Phi \colon \R \to [0,\infty[$, with
$\Phi(\eta)=0$ for $\eta\leq 0$,
$\Phi\in L^\infty_\mathrm{loc}(]0,\infty[)$,
and there exists $-1/2<k<3/2$ and constants $c_1, c_2 >0$ such that for all
sufficiently small $\eta >0$,
\be\label{Phiasymp}
c_1 \eta^k \leq \Phi(\eta) \leq c_2 \eta^k.
\ee

\smallskip

These assumptions are sufficient for the results of the present section,
in particular for the analysis of the steady states in the
limit of large central redshift, which we believe has some interest in itself.
A typical guiding example which satisfies
$(\Phi 1)$ is the function
\[
\Phi(\eta) = c\, \eta^k \ \mbox{for}\ \eta \geq 0,\
\Phi(\eta)=0\ \mbox{for}\ \eta < 0
\]
with constants $c>0$ and $-1/2<k<3/2$. These are the so-called polytropes,
named so by analogy to the well-known polytropic equations of state in
compressible fluid dynamics.

Since only the metric
quantity $\mu$ enters into the definition of the particle energy
$E$ in \eqref{parten}, the field equations can be reduced to a
single equation for $\mu$. It is tempting to prescribe
$\mu(0)$, but since the ansatz \eqref{trueansatz} contains the cut-off
energy $E^0$ as another, in principle free parameter and since
$\mu$ must vanish at infinity due to \eqref{boundcinf} this approach is not
feasible. Instead we define $y:= \log E^0 - \mu$ so that $e^\mu = E^0 e^{-y}$.
For the ansatz \eqref{trueansatz}
the spatial mass density and pressure become functions of $y$, i.e.,
\be \label{rhoyrelv}
\rho(r) = g(y(r)), \quad
p(r) = h(y(r)) = p_T(r),
\ee
where
\be \label{gdefv}
g(y) :=
4 \pi e^{4 y} \int_0^{1-e^{-y}} \Phi(\eta)\, (1-\eta)^2\,
\left((1-\eta)^2-e^{-2y}\right)^{1/2} d\eta
\ee
and
\be \label{hdefv}
h(y) :=
\frac{4 \pi}{3} e^{4 y} \int_0^{1-e^{-y}} \Phi(\eta)\,
\left((1-\eta)^2-e^{-2y}\right)^{3/2} d\eta.
\ee
The functions $g$ and $h$ are continuously differentiable
on $\R$, cf.\ \cite[Lemma 2.2]{RR00}, and they vanish for $y<0$.
The metric coefficient $\lambda$ can be eliminated from the system,
because the field equation \eqref{ein1} together with the boundary condition
\eqref{boundc0} at zero imply that
\begin{equation}\label{lambdadef}
e^{-2\lambda(r)} = 1 - \frac{2  m (r)}{r},
\end{equation}
where the mass function $m$ is defined by
\be \label{mdef_ssv}
m(r) = m(r,y) = 4\pi \int_0^r \sigma^2 \rho(\sigma)\, d\sigma.
\ee
Hence the static Einstein-Vlasov
system is reduced to the equation
\be \label{yeq}
y'(r)= - \frac{1}{1-2 m(r)/r} \left(\frac{m(r)}{r^2} + 4 \pi r p(r)\right) ,
\ee
which is equivalent to \eqref{ein2};
here $m$, $\rho$, and $p$ are given in terms of $y$ by \eqref{rhoyrelv}
and \eqref{mdef_ssv}.

In \cite{RaRe} it is shown that for every central value
\be\label{yiv}
y(0)=\kappa>0
\ee
there exists a unique smooth solution $y=y_\kappa$
to \eqref{yeq}, which is defined on $[0,\infty[$
and which has a unique zero at some radius $R>0$. In view of
\eqref{rhoyrelv}--\eqref{hdefv} this implies that the induced quantities
$\rho$ and $p$ are supported on the interval $[0,R]$,
and a non-trivial steady state
of the Einstein-Vlasov system with compact support and finite mass is obtained.
We observe that
the limit $y(\infty):=\lim_{r\to \infty} y(r) < 0$ exists, since $r\mapsto y(r)$ is a decreasing function and has a unique zero as mentioned above.
The metric quantity $\mu$ is given by $e^{\mu(r)}= E^0 e^{-y(r)}$, and in order
that $\mu$ has the correct boundary value at infinity we must define
$E^0:=e^{y(\infty)}$. Since $y(R)=0$ we also see that  $E^0=e^{\mu(R)}$.
We want to relate the parameter $\kappa$ to the redshift factor $z$
of a photon which is emitted at the center $r=0$ and received at the
boundary $R$ of the steady state; this is not the standard definition
of the central redshift where the photon is received at infinity,
but it is a more suitable parameter here:
\be \label{redshiftdef}
z = \frac{e^{\mu(R)}}{e^{\mu(0)}} - 1 = \frac{e^{y(0)}}{e^{y(R)}} - 1 = e^\kappa -1.
\ee
Hence $\kappa$ is in one-to-one correspondence with the central redshift
factor $z$ with $\kappa\to \infty$ iff  $z\to \infty$, and although this
is not the standard terminology
we refer to $\kappa$ as the central redshift.
For a fixed ansatz function $\Phi$ we therefore obtain
a family $(y_\kappa)_{\kappa>0}$ of solutions to \eqref{yeq},
which induce steady states to the Einstein-Vlasov system
parameterized by the central redshift $\kappa$,
and each member of this family represents a galaxy in equilibrium,
which has finite mass and compact support.

\begin{remark}
  The central redshift $\kappa$ and the central density $\rho_c=\rho(0)$
  are related via $\rho_c= g(\kappa)$, see~\eqref{gdefv}.
  It can also be seen from~\eqref{gdefv} that $g$ is monotonically
  increasing so that $\kappa$ and $\rho_c$ are in a 1-1 relationship.
  In the literature both $\kappa$ and $\rho_c$ are used to parameterize
  the steady state solutions, and the two parametrizations are equivalent.
\end{remark}

We now consider the Euler case.
For stationary, spherically symmetric solutions of the Einstein-Euler system
the velocity field necessarily vanishes, $u=0$, cf.~\eqref{eelambdad}.
The remaining field equations \eqref{eelambda}, \eqref{eemu}, \eqref{ee2ndo}
then become identical to the Vlasov case, the Euler equation \eqref{rhod}
is satisfied identically, and \eqref{ud} reduces to
\be \label{tov0}
(\rho +p) \mu' + p'= 0.
\ee
We define
\[
Q(\rho) := \int_0^\rho \frac{P'(s)}{s+P(s)} ds,\ \rho\geq 0.
\]
Then \eqref{tov0} can be written as
\[
0= Q'(\rho)\rho' + \mu' = \frac{d}{dr}\left(Q(\rho) + \mu\right)
\]
which means that on the support of the matter,
\[
Q(\rho(r)) + \mu(r) = const.
\]
In analogy to the Vlasov case we introduce $y=const -\mu$ and find that $\rho$
is given in terms of $y$,
\be \label{gdefe}
\rho = g(y) := \left\{
\begin{array}{ccl}
  Q^{-1}(y)&,&y>0,\\
  0&,&y\leq 0.
\end{array}
\right.
\ee
Taking into account the equation of state \eqref{eqstate} it follows that
\be \label{hdefe}
p = h(y) := P(g(y)).
\ee
Hence the stationary system takes exactly the same form as in the Vlasov case
and can be reduced to the equation \eqref{yeq}, the only difference being
the definitions of the functions $g$ and $h$ in \eqref{gdefv}, \eqref{hdefv}
or \eqref{gdefe}, \eqref{hdefe} respectively.
The interpretation of $\kappa=y(0)$ remains as explained
for the Vlasov case above.

We now specify the assumptions on the function $P$
which defines the equation of state.

\smallskip

\noindent
{\bf Assumption (P1).}
$P \in C^1([0,\infty[)$ with $P'(\rho) >0$ for $\rho>0$,
$P(0)=0$, there exist constants $c_1, c_2>0$ and $0<n<3$ such that
\be \label{Pass1}
P'(\rho) \leq c_1 \rho^{1/n}\ \mbox{for all sufficiently small}\ \rho>0
\ee
and the inverse of $P$ exists on $[0,\infty[$ and
satisfies the estimate
\be \label{Pass2}
|P^{-1}(p) - 3 p| \leq c_2 p^{1/2}\ \mbox{for all}\ p>0.
\ee

\smallskip

Before we discuss examples for such equations of state in the next subsection, we briefly mention
that the condition~\eqref{Pass1} forces the equation of state to take the approximate form $P(\rho)\sim_{\rho\to0}\rho^{1+\frac1n}$ in the vicinity
of the vacuum boundary, while~\eqref{Pass2} makes the equation of state linear at the leading order for large values of $\rho$, i.e.
$P(\rho)\sim_{\rho\to\infty}\frac13\rho$. The first assumption is the well-known polytropic law from the classical gas dynamics, while the second assumption is necessary to ascertain that the speed of sound remains smaller than the speed of light at high densities.

Assumption \eqref{Pass1} together with the required regularity guarantees
that the function $g$ and $h$ defined in \eqref{gdefe} and \eqref{hdefe}
are $C^1$ on $\R$, and for each $\kappa= y(0)>0$ the equation \eqref{yeq}
has a unique solution with the same properties as stated for the Vlasov
case, cf.\ \cite{RaRe}.
For easier reference we collect the findings of this subsection.

\begin{proposition} \label{ssfamilies}
  \begin{itemize}
  \item[(a)]
    Let $\Phi$ satisfy $(\Phi 1)$.
     Then there exists a one-parameter family of steady states
     $(f_{\kappa},\lambda_{\kappa},\mu_{\kappa})_{\kappa >0}$
     of the spherically symmetric, asymptotically flat
     Einstein-Vlasov system.
  \item[(b)]
    Let $P$ satisfy $({\mathrm P}1)$.
     Then there exists a one-parameter family
     of steady states $(\rho_{\kappa},\lambda_{\kappa},\mu_{\kappa})_{\kappa >0}$
     of the spherically symmetric, asymptotically flat
     Einstein-Euler system.
  \end{itemize}
  In both cases these steady states are compactly supported on some interval
  $[0,R_\kappa]$,
  $\rho_\kappa, p_\kappa \in C^1([0,\infty[)$,
  $y_\kappa, \mu_\kappa, \lambda_\kappa  \in C^2([0,\infty[)$, and
  $\rho'_\kappa(0)=p'_\kappa(0)= y'_\kappa(0)= \mu'_\kappa(0)= \lambda'_\kappa(0)=0$.
  The following two identities hold on  $[0,\infty[$, the second being
  known as the Tolman-Oppenheimer-Volkov equation:
  \be \label{laplusmu}
  \lambda_\kappa' + \mu_\kappa'
  = 4 \pi r e^{2\lambda_\kappa} \left(\rho_\kappa +p_\kappa\right),
  \ee
  \be \label{tov}
   p_\kappa' = - \left(\rho_\kappa +p_\kappa\right) \,\mu_\kappa'.
  \ee
\end{proposition}
\begin{proof}
  Eqn.~\eqref{laplusmu}, which holds also in the time-dependent case,
  follows by adding the field equations
  \eqref{ein1}, \eqref{ein2} or \eqref{eelambda}, \eqref{eemu}
  respectively. For the Einstein-Euler case \eqref{tov} was already
  stated above as \eqref{tov0}. That this equation also holds in
  the Einstein-Vlasov case is due to the fact that a steady state
  of this system is macroscopically also one of the Einstein-Euler system,
  as will be seen in the next section.
\end{proof}

\subsection{Microscopic and macroscopic equations of state}\label{micromacro}

Let us consider a microscopic equation of state $\Phi$ which satisfies
the assumption ($\Phi 1$) (see~\eqref{Phiasymp}), and a corresponding
steady state of the Einstein-Vlasov system.
Given the relations \eqref{rhoyrelv} it is tempting to write
\[
p = h(g^{-1}(\rho))
\]
and interpret this steady state as a solution to the Einstein-Euler
system with the macroscopic equation of state given by
\be \label{pphi}
P = P_\Phi := h\circ g^{-1};
\ee
we fix some $\kappa >0$ and drop the corresponding dependence
since it plays no role here. Indeed, this maneuver is perfectly
rigorous and provides a class of examples $P=P_\Phi$ which satisfy
$({\mathrm P} 1)$ (see~\eqref{Pass1}--\eqref{Pass2}). To see this, we first observe that
\be \label{ghj}
g = 3 h + 3 j
\ee
with
\be \label{jdefv}
j(y) :=
\frac{4 \pi}{3} e^{2y} \int_0^{1-e^{-y}} \Phi(\eta)\,
\left((1-\eta)^2-e^{-2y}\right)^{1/2} d\eta.
\ee
The functions $h$ and $j$ and hence also $g$ are
continuously differentiable with
\be \label{hprime}
h' = 4 h + 3 j
\ee
and
\be \label{jprime}
j'(y) = 2 j(y) + \frac{4 \pi}{3} \int_0^{1-e^{-y}} \Phi(\eta)\,
\left((1-\eta)^2-e^{-2y}\right)^{-1/2} d\eta,
\ee
cf.\ \cite[Lemma 2.2]{RR00}.
In particular, the functions $g$ and $h$ are strictly increasing
on $[0,\infty[$
with $\lim_{y\to\infty} g(y) = \lim_{y\to\infty} h(y) = \infty$
so that these functions are one-to-one and onto on $[0,\infty[$.

\begin{proposition} \label{veqstate}
  Assume that $\Phi$ satisfies $(\Phi 1)$.
  Then $P_\Phi$ defined in \eqref{pphi} satisfies
  $({\mathrm P} 1)$.
  If in addition $\Phi \in C^1([0,\infty[)$,
  then $P_\Phi \in C^2([0,\infty[)$.
\end{proposition}

\noindent
\begin{proof}
Let $P$ be given by \eqref{pphi}. Then for $\rho>0$ and with
$y=g^{-1}(\rho)>0$,
\be \label{Pprime}
P'(\rho) = \frac{h'(g^{-1}(\rho))}{g'(g^{-1}(\rho))}
= \frac{1}{3} \frac{g'(y) - j'(y)}{g'(y)}
\ee
which implies that $P\in C^1([0,\infty[)$ with $P(0)=0$ and
\be \label{pprimecond}
0< P'(\rho) < \frac{1}{3},\ \rho>0.
\ee
We claim that
\be \label{pcond}
\left|P(\rho) -\frac{1}{3} \rho\right| \leq C \rho^{1/2},\ \rho\geq 0
\ee
for some constant $C>0$; such constants may depend only on $\Phi$
and may change their value from line to line.
For $0\leq \rho \leq 1$ this is obvious, since by
\eqref{pprimecond}, $P(\rho) \leq \rho/3$. With $y=g^{-1}(\rho)$ the estimate
\eqref{pcond} is equivalent to
\be \label{pcondy}
\left|h(y) -\frac{1}{3} g(y)\right| \leq C g(y)^{1/2}
\ee
which needs to be shown for $y\geq g^{-1}(1)=:y_0$. Now
\[
h(y) -\frac{1}{3} g(y) = j(y) \leq
\frac{4 \pi}{3} e^{2y} \int_0^{1} \Phi(\eta)\,
(1-\eta)\, d\eta = C e^{2y},
\]
while
\[
g(y)^{1/2} \geq h(y)^{1/2} \geq
\left(\int_0^{1-e^{-y_0}} \Phi(\eta)\,
\left((1-\eta)^2-e^{-2y_0}\right)^{3/2} d\eta\right)^{1/2} e^{2 y} = C e^{2y},
\ y\geq y_0.
\]
Combining both estimates yields \eqref{pcondy} for $y\geq y_0$.
We rewrite \eqref{pcond}
in terms of $P^{-1}$ to find that
\be \label{pinvcond}
\left| P^{-1}(p) - 3 p\right| \leq C (P^{-1}(p))^{1/2},\ p\geq 0.
\ee
By \eqref{pcond}, $P(\rho) \geq \rho/6$ for $\rho$ large which means that
$P^{-1}(p) \leq 6 p$ for $p$ large so that \eqref{pinvcond} implies
\eqref{Pass2} in that case. We prove that
there exist constants $c_1, c_2 >0$ such that
with $\gamma = 1/(k+3/2)$,
\be \label{Pestsmallrho}
c_1 \rho^{\gamma} \leq P'(\rho) \leq c_2 \rho^{\gamma},\quad
c_1 \rho^{\gamma+1} \leq P(\rho) \leq c_2 \rho^{\gamma+1} \
\mbox{for}\ \rho>0\ \mbox{small}.
\ee
This implies \eqref{Pass1} with $n=k+3/2 \in ]1,3[$, and
$P^{-1}(p) \leq C p^{1/(\gamma+1)}$ for $p$ small.
Thus
\[
\left| P^{-1}(p) - 3 p\right| \leq C( p^{1/(\gamma+1)} + p) \leq C p^{1/2}
\]
for $p$ small which completes the proof of \eqref{Pass2} so that (P1) holds;
notice that $\gamma < 1$ since $k>-1/2$.
It therefore remains to prove \eqref{Pestsmallrho}, and it suffices
to prove the estimate for $P'$ which implies the one for $P$.
In view of \eqref{Pprime} and with $\rho=g(y)$ the estimate for $P'$ is
equivalent to
\be \label{Pprimeestsmally}
c_1 g(y)^\gamma g'(y) \leq h'(y) \leq c_2 g(y)^\gamma g'(y)
\ \mbox{for}\ y>0\ \mbox{small}.
\ee
Now
\[
g=3 h + 3 j \geq 3 j,\ g' = 3 h' + 3 j' \geq 3 j',\
h'= 4 h + 3 j \geq 3 j
\]
and $j'$ is given by \eqref{jprime}.
Hence all the relevant terms are of the form
\[
e^{a y} \int_0^{1-e^{-y}} \Phi(\eta)\,
\left((1-\eta)^2-e^{-2y}\right)^{m} d\eta
\]
with $a\in\{0,2,4\}$ and $m\in\{-1/2,1/2,3/2\}$, and these terms
need to be estimated from above and from below for $y$ small.
We  can drop the factor $e^{a y}$. For $\Phi(\eta)$ we use the assumption
\eqref{Phiasymp}, and we observe that
\[
c (1-\eta -e^{-y}) \leq (1-\eta  + e^{-y}) (1-\eta - e^{-y})
= (1-\eta)^2 - e^{-2y} \leq 1-\eta - e^{-y}.
\]
Hence the relevant terms can be estimated from above and below by
\[
\int_0^{1-e^{-y}} \eta^k\,
\left(1-\eta-e^{-y}\right)^{m} d\eta = C (1-e^{-y})^{k+m+1} .
\]
This implies that $h'(y)$ can be estimated from above and below by
$(1-e^{-y})^{k+3/2}$, and   $g(y)^\gamma g'(y)$
can be estimated from above and below by
$(1-e^{-y})^{\gamma (k+3/2)+(k+1/2)}$. But by the choice of $\gamma$ the two exponents
are equal and the estimate \eqref{Pprimeestsmally} holds.

It remains to show the regularity assertion. So let in addition
$\Phi \in C^1([0,\infty[)$. We observe that
\beas
g(y)
&=&
\int \Phi(1-e^{-y}\p)\p dv,\\
h(y)
&=&
\int \Phi(1-e^{-y}\p) \left(\frac{x\cdot v}{r}\right)^2 \frac{dv}{\p}.
\eeas
We differentiate these expressions under the integral and change variables
to find that
\beas
g'(y)
&=&
4 \pi e^{4 y} \int_0^{1-e^{-y}} \Phi'(\eta)\,(1-\eta)^3
\left((1-\eta)^2 - e^{-2 y}\right)^{1/2} d\eta,\\
h'(y)
&=&
\frac{4 \pi}{3} e^{4 y} \int_0^{1-e^{-y}} \Phi'(\eta)\,(1-\eta)\,
\left((1-\eta)^2 - e^{-2 y}\right)^{3/2} d\eta.
\eeas
Arguing as in \cite[Lemma 2.2]{RR00} it follows that $g,h \in C^2(\R)$,
and hence by \eqref{Pprime}, $P\in C^2([0,\infty[)$.
\end{proof}

\begin{remark}
  \begin{itemize}
    \item[(a)]
      The above proof shows that any equation of state derived from
      a microscopic ansatz satisfying $(\Phi 1)$ obeys the asymptotic
      relation
      \be \label{radiation}
      P(\rho) \approx \frac{1}{3} \rho\ \mbox{for}\ \rho\ \mbox{large},
      \ee
      which for technical reasons
      is expressed in terms of $P^{-1}$ in $({\mathrm P} 1)$.
      The limiting equation of state $P(\rho)=\rho/3$
      is known in astrophysics and cosmology as the equation of state for
      radiation. It is remarkable that this physically reasonable behavior
      is taken care of automatically if the equation of state derives from a
      microscopic one.
    \item[(b)]
      For the fluid model this asymptotic behavior must be put in
      by hand. In particular, it excludes equations of state of the form
      $P(\rho)=c \rho^\gamma$ with $\gamma >1$.
      Such equations of state, which also violate the requirement that the
      speed of sound $\sqrt{P'}$ must be less than the speed of light,
      are common for the Euler-Poisson system.
    \item[(c)]
      It would be a minor technical modification to replace the factor $1/3$
      in \eqref{radiation} and $({\mathrm P} 1)$ by some factor
      $\alpha \in ]0,1[$.
  \end{itemize}
\end{remark}

The above discussion indicates
that macroscopic quantities induced by an
isotropic steady state of the Einstein-Vlasov system
represent a steady state of the Einstein-Euler system with
an equation of state given by \eqref{pphi}.
To complete this argument, it remains to check
the stationary Euler equation \eqref{tov}.
But using the relations \eqref{ghj} and \eqref{hprime},
\[
p' = h'(y)\, y' = -(4 h + 3 j)(y) \mu' = -(g(y) + h(y)) \mu'
= - (\rho + p)\, \mu'.
\]

As a genuinely fluid dynamical equation of state we consider the one
used for neutron stars in \cite{HsuLinMakino}. Here for $y\geq 0$,
\[
\tilde g(y) :=
3 \int_0^y s^2 \sqrt{1+s^2} ds,\
\tilde h(y) :=
\int_0^y \frac{s^4}{\sqrt{1+s^2}} ds,
\]
and the relation between $p$ and $\rho$ is given by
\[
\rho = \tilde g(y),\ p= \tilde h(y),
\]
or equivalently, by the equation of state
\be \label{eqstatens}
p = P_{\mathrm{NS}}(\rho) := \tilde h(\tilde g^{-1}(\rho)).
\ee

\begin{proposition} \label{nsok}
  The equation of state \eqref{eqstatens} for a neutron star is well
  defined, satisfies $({\mathrm P}1)$, and in addition
  $P_\mathrm{NS} \in C^2([0,\infty[)$ with $P_\mathrm{NS}'' >0$ on $]0,\infty[$.
\end{proposition}

\begin{proof}
Clearly, $\tilde g, \tilde h \in C^\infty([0,\infty[)$ with
\[
\tilde g'(y) = 3 y^2 \sqrt{1+y^2},\ \tilde h'(y) = \frac{y^4}{\sqrt{1+y^2}}.
\]
Also, both functions are one-to-one and onto on the interval $[0,\infty[$.
Hence $P_\mathrm{NS} \in C^\infty(]0,\infty[)\cap C([0,\infty[)$ with
$P_\mathrm{NS}(0)=0$. Moreover,
\be \label{pnsprime}
P_\mathrm{NS}'(\rho) =
\frac{1}{3} \frac{(\tilde g^{-1}(\rho))^2}{1+(\tilde g^{-1}(\rho))^2}.
\ee
This implies that for $\rho >0$,
\be \label{pnsprimeest}
0 < P_\mathrm{NS}'(\rho) < \frac{1}{3}
\ee
and
\[
\left|P_\mathrm{NS}'(\rho)-\frac{1}{3}\right|
= \frac{1}{3} \frac{1}{1+(\tilde g^{-1}(\rho))^2}.
\]
For $y>0$ large it holds that $\tilde g(y)\leq C y^4$ and hence
$\tilde g^{-1}(\rho)\geq C \rho^{1/4}$ and
\[
\left|P_\mathrm{NS}'(\rho)-\frac{1}{3}\right| \leq C \rho^{-1/2}
\]
for $\rho>0$ large, but due to \eqref{pnsprimeest} this estimate
also holds for $\rho>0$ small. Hence
\be \label{pnsest13}
\left|P_\mathrm{NS}(\rho)-\frac{1}{3}\rho\right| \leq C \rho^{1/2},\ \rho\geq 0.
\ee
Next we observe that for $y>0$ small it holds that
$y^3 \leq \tilde g(y) \leq C y^3$ and hence for $\rho>0$ small,
$c \rho^{1/3} \leq \tilde g^{-1}(\rho) \leq \rho^{1/3}$
which together with \eqref{pnsprime} implies that
\be \label{pnsest}
c \rho^{2/3} \leq P_\mathrm{NS}'(\rho) \leq \rho^{2/3},\quad
c \rho^{5/3} \leq P_\mathrm{NS}(\rho) \leq C \rho^{5/3}
\ee
for $\rho\geq 0$ small; constants like $0<c<C$ may change their value
from line to line. In particular, \eqref{Pass1} in $({\mathrm P} 1)$ holds.
The estimate \eqref{pnsest13} implies that
\[
\left| P_\mathrm{NS}^{-1}(p) - 3 p\right| \leq C  (P_\mathrm{NS}^{-1}(p))^{1/2}.
\]
Since $P_\mathrm{NS}(\rho)\geq \rho/6$ for $\rho >0$ large and hence
$P_\mathrm{NS}^{-1}(p) \leq 6 p$ for $p>0$ large, this implies
\eqref{Pass2} for such $p$. But if $\rho$ and $p$ are small, then \eqref{pnsest}
implies that $P_\mathrm{NS}^{-1}(p) \leq C p^{3/5}$ and hence
\[
\left| P_\mathrm{NS}^{-1}(p) - 3 p\right| \leq C  (p^{3/5} + p) \leq C p^{1/2}
\]
which completes the proof of \eqref{Pass2}.
From \eqref{pnsprime} it follows that
\[
P_\mathrm{NS}''(\rho) =
\frac{2}{3} \frac{\tilde g^{-1}(\rho)(\tilde g^{-1})'(\rho)}
     {(1+(\tilde g^{-1}(\rho))^2)^2} > 0
\]
for $\rho>0$, and the proof is complete.
\end{proof}


\subsection{The Bisnovatyi-Kogan--Zel'dovich (BKZ) solution} \label{SS:BKZ}


A key ingredient in our instability analysis is rather precise information
on the form of the steady states for large central redshift $\kappa$.
In order to obtain this information it is useful to understand that
the steady states approach steady states of a certain version of the
Einstein-Vlasov or Einstein-Euler system as $\kappa\to\infty$.
In this subsection we introduce this limiting system
and discuss a special, explicit solution to it.

We first consider the Vlasov case.
For $\kappa\to\infty$ we expect that close to the center the
corresponding solution of \eqref{yeq}, \eqref{yiv} is large, so that formally,
\[
g(y) = 4 \pi e^{4 y} \int_0^{1-e^{-y}} \Phi(\eta)\, (1-\eta)^2\,
\left((1-\eta)^2 - e^{-2y}\right)^{1/2} d\eta
\approx g^\ast (y)
\]
and
\[
h(y) =
\frac{4 \pi}{3} e^{4y} \int_0^{1-e^{-y}} \Phi(\eta)\,
\left((1-\eta)^2-e^{-2y}\right)^{3/2} d\eta
\approx h^\ast (y)
\]
where
\be \label{ghdef_as}
g^\ast(y) :=
e^{4y},\
h^\ast(y) :=\frac{1}{3} e^{4y},\  y\in \R,
\ee
and for the sake of notational simplicity we normalized
\[
4 \pi \int_0^{1} \Phi(\eta)\, (1-\eta)^3\, d\eta = 1.
\]
Since $g^\ast$ and $h^\ast$ are strictly positive,
a corresponding steady state is never compactly supported.
Let us for the moment assume that we have a solution $y$
of \eqref{yeq}, \eqref{yiv} where $g$ and $h$ are replaced by
$g^\ast$ and $h^\ast$, and let
$\mu, \lambda, \rho, p$ be induced by $y$.
Then these quantities satisfy
the stationary Einstein equations together with the radiative equation of state
\be \label{radeqstate}
p = P^\ast(\rho) = \frac{1}{3} \rho.
\ee
We want to understand what happens with the Vlasov equation in the limit
$\kappa \to \infty$. First, we note that the equation
of state \eqref{radeqstate} can not come from an isotropic
steady state particle distribution of the form \eqref{trueansatz}
(and of course not from an anisotropic one either):
\[
p(r) =
\int  f(x,v)\, \left(\frac{x\cdot v}{r}\right)^2 \frac{dv}{\p}
=
\frac{1}{3}\int  f(x,v)\, |v|^2 \frac{dv}{\p}
<
\frac{1}{3}\int  f(x,v)\, (1+|v|^2) \frac{dv}{\p}
= \frac{1}{3} \rho(r).
\]
The physical meaning of this
is simply that massive particles are not radiation.
However, this observation carries its own cure: In the $\kappa \to \infty$
limit the stationary Vlasov equation \eqref{vlasov}
must be replaced by its ultrarelativistic, or massless version
\be \label{urvlasov}
e^{\mu - \lambda}\frac{v}{|v|}\cdot \partial_x f -
e^{\mu - \lambda} |v| \, \mu'(r) \, \frac{x}{r}\cdot \partial_v f =0.
\ee
An isotropic solution of the latter equation is obtained by the ansatz
\[
f(x,v) = \phi(e^{\mu(r)}|v|);
\]
it is straightforward that this satisfies \eqref{urvlasov}.
Moreover, for such a particle distribution $f$,
\[
p(r) = \int f(x,v)\, \left(\frac{x\cdot v}{r}\right)^2
\frac{dv}{|v|}
=
\frac{1}{3}\int  f(x,v)\, |v|^2 \frac{dv}{|v|} \\
= \frac{1}{3} \rho(r)
\]
as required by \eqref{radeqstate}.
Finally,
\[
\rho(r) = \int \phi(e^{\mu(r)}|v|)\, |v|\, dv
= 4\pi \int_0^\infty \phi(e^{\mu(r)} u)\, u^3 du
= 4\pi \int_0^\infty \phi(\eta)\, \eta^3 d\eta \, e^{-4\mu(r)}
\]
which is as expected from \eqref{ghdef_as} so that
the above $f$ is a consistent solution to the massless
Vlasov equation \eqref{urvlasov}, and a stationary solution of the
massless Einstein-Vlasov system, i.e., of the system
\eqref{feqgen}, \eqref{vlgen}, \eqref{emtvlgen} for particles with
rest mass $m_0=0$.


\begin{lemma}
For any $\kappa >0$ there exists a unique
solution $y=y_\kappa\in C^1([0,\infty[)$ to the problem
\be \label{grzeq}
y'(r)= - \frac{1}{1-2 m^\ast(r)/r} \left(\frac{m^\ast(r)}{r^2}
+ 4 \pi r p^\ast(r)\right) ,
\ee
\be\label{ziv}
y(0)=\kappa>0,
\ee
where $\rho^\ast=g^\ast(y)$, $p^\ast=h^\ast(y)$ with \eqref{ghdef_as},
and
\be \label{minfdef}
m^\ast(r) = m^\ast(r,y) = 4\pi \int_0^r \sigma^2 \rho^\ast(\sigma)\, d\sigma.
\ee
\end{lemma}

\begin{proof}
The proof follows using the arguments in \cite[Theorem~3.4]{Rein94}.
\end{proof}


We want to find a special, explicit solution of  \eqref{grzeq}.
To this end, we make the ansatz
\[
y(r) = \gamma - \beta \ln r,\ r>0 .
\]
Then
\[
\rho(r) = g^\ast(y(r)) = e^{4\gamma} r^{-4\beta},\
p(r) = h^\ast(y(r)) = \frac{1}{3} e^{4\gamma} r^{-4\beta},
\]
and
\[
m(r) = \frac{4 \pi e^{4\gamma}}{3-4\beta} r^{3-4\beta}.
\]
Substituting this into \eqref{grzeq} it turns out that the latter equation
holds iff
\[
\beta = \frac{1}{2} \ \mbox{and}\ \frac{56}{3}\pi e^{4\gamma} = 1.
\]
Thus
\be \label{bkzsolution}
\rho(r) = \frac{3}{56 \pi} r^{-2},\ p(r) = \frac{1}{56 \pi} r^{-2},\
m(r) = \frac{3}{14} r,\ \frac{2m(r)}{r} =  \frac{3}{7},\
e^{2\lambda} = \frac{7}{4},\ \mu'(r) = \frac{1}{2 r}.
\ee
These macroscopic data of the solution are
the same as for the massive solutions found
by {\sc  Bisnovatyi-Kogan} and {\sc Zel'dovich} in \cite{BKZ},
and we refer to it as the BKZ solution. We note that the corresponding
$y$ is singular at the origin, and the solution
violates the condition \eqref{boundc0} for a regular center. It
has infinite mass and also violates the boundary condition \eqref{boundcinf}
at infinity, i.e., it does not represent an isolated system.

\begin{remark}[Geometric properties of the BKZ solution]
To understand the BKZ solution better we compute
its Ricci curvature $R(r)$
and its Kretschmann scalar $K(r)$. With the help of Maple
it turns out that
\[
R(r) = 0,\ K(r) := R_{\alpha\beta\gamma\delta} R^{\alpha\beta\gamma\delta}(r)
= \frac{72}{49} r^{-4}
\]
so that this solution has a spacetime singularity at $r=0$.
A radially outgoing null geodesic satisfies the relation
\[
0 = - c_1 r dt^2 + \frac{7}{4} dr^2,
\]
hence
\[
\frac{dr}{dt} = \sqrt{\frac{4 c_1}{7} r}  =: c_2  \sqrt{r}
\]
and
\[
r(t) = \left(\sqrt{r(0)} + \frac{1}{2} c_2 t\right)^2,\
t\geq - \frac{2 \sqrt{r(0)}}{c_2}
\]
describe the radially outgoing null geodesics. They
start at the singularity and escape to $r=\infty$, i.e.,
the singularity is visible for observers away from the singularity,
and hence it violates the strong cosmic censorship conjecture.
The concept of weak cosmic censorship is not applicable to this solution,
since it is not asymptotically flat. According to the
cosmic censorship hypothesis ``naked'' singularities, i.e., those
which are visible
to distant observers, should be ``non-generic'' and/or ``unstable''.
Our analysis turns out to be nicely consistent with this hypothesis:
regular steady states, by being close to the BKZ solution
for large central redshift, seem to inherit this instability and are
indeed unstable themselves, as we will show below.
\end{remark}

\begin{remark}
  The massless Einstein-Vlasov system has been considered in the literature,
  for example in \cite{andfajthal,dafermos,taylor}.
  The results of the present section, more precisely of the next subsection,
  show
  that the massless system captures the behavior of solutions of the massive
  system in a strongly relativistic situation, at least in the context of
  spherically symmetric steady states.
  The massless system plays a similar role in \cite{Andr2007}.

  The size of the central redshift is not the only possible measure of
  the strength of relativstic effects in a given steady state.
  Another possible choice is
  the compactness ratio $\frac{2m(r)}{r}$, which by the Buchdahl
  inequality~\cite{Buch,Andr2008}
  is always less than $\frac89$. We see from~\eqref{bkzsolution}
  that the compactness ration for the BKZ solution is
  $\frac37$ and it is therefore ``far" from saturating the
  Buchdahl bound $\frac89$.
  \end{remark}

So far, the discussion of the present subsection was restricted to the
Einstein-Vlasov system, but in view of the relation between the corresponding
steady states explained in Subsection~\ref{micromacro} it applies equally well
to the Einstein-Euler system with the equation of state \eqref{radiation},
which we would refer to as the massless or radiative Einstein-Euler system.

\subsection{The ultrarelativistic limit $\kappa\to\infty$}
\label{SS:LARGEKAPPA}

In this subsection we prove that in a specified region
close to the center the steady states provided by Proposition~\ref{ssfamilies}
are approximated by the BKZ solution when $\kappa$ is sufficiently large.
This is done in two steps. We first prove that the former steady states
are approximated by solutions of the massless problem
\eqref{grzeq} with the same central value \eqref{ziv}. In a second step
we show that the behavior of the latter solutions is captured by the BKZ one.

For the first step we start with the Vlasov case.
It is fairly simple to make the asymptotic relation between
the functions $g$ and $h$ defined by \eqref{gdefv} and \eqref{hdefv}
and their massless counterparts
$g^\ast$ and $h^\ast$ precise:
\[
|g(y) - g^\ast(y)| + |h(y) - h^\ast(y)| \leq C e^{2y},\ y\in \R,
\]
It should be noted that this is indeed a good approximation
for large $y$, since all the terms on the left are then of order $e^{4y}$.
This approximation would suffice to estimate the difference between
solutions of \eqref{yeq}, \eqref{yiv} and solutions of
\eqref{grzeq}, \eqref{ziv}.
However, for the functions $g$ and $h$ defined by
\eqref{gdefe} and \eqref{hdefe}
in the Euler case the above estimate seems hard or impossible to obtain.
For the latter case a different set-up must be used to obtain the desired
asymptotics, but that set-up can be chosen such that it also works
for the Vlasov case.

So from now on we treat both cases simultaneously and consider a
steady state as provided by Proposition~\ref{ssfamilies}.
For the moment we drop the subscript
$\kappa$ and note that by \eqref{tov}, which as we showed holds for
the Vlasov case as well,
\[
p' = -(\rho + p)\, \mu'
= -\frac{\rho + p}{1-\frac{2 m}{r}}\left(\frac{m}{r^2} + 4\pi r p\right) .
\]
We want to read this as a differential equation for $p$, and so we observe that
the equation of state $p=P(\rho)$ can be inverted to define
$\rho$ in terms of $p$:
\be \label{eqstateinv}
\rho = P^{-1}(p) =: S (p);
\ee
we notice that under our general assumptions
the function $P$ is one-to-one and onto on the
interval $[0,\infty[$ both in the Vlasov and in the Euler case
and both for the massive and the massless equations.
We also notice that $p$ determines all the other components of
the corresponding steady state. Hence we consider the equation
\be \label{peq}
p' =
-\frac{S(p) + p}{1-\frac{8 \pi}{r}\int_0^r s^2 S(p)\,ds}
\left(\frac{4 \pi}{r^2}\int_0^r s^2 S(p)\,ds + 4\pi r p\right),
\ee
and its massless counterpart where according to \eqref{radeqstate}
we set $S(p)=3 p$:
\be \label{peqml}
p' =
-\frac{4 p}{1-\frac{8 \pi}{r}\int_0^r s^2 3 p\,ds}
\left(\frac{4 \pi}{r^2}\int_0^r s^2 3 p\,ds + 4\pi r p\right).
\ee
In what follows, $p_\kappa$ and $p^\ast_\kappa$ denote the solutions
to \eqref{peq} and \eqref{peqml} respectively, satisfying the
boundary condition
\be\label{piv}
p_\kappa (0) = \frac{1}{3} e^{4 \kappa} = p^\ast_\kappa(0).
\ee
\begin{remark}\label{repara}
  Strictly speaking we reparametrize our steady state
  family here and use the central pressure as the new parameter.
  However, the quantities $y$ and $p$ are in a strictly increasing,
  one-to-one correspondence in such a way that $y\to \infty$
  iff $p\to \infty$, i.e., $p(0) \to \infty$ is equivalent to
  $y(0) \to \infty$. In view of \eqref{ghdef_as} the above way of writing
  $p(0)$ guarantees that at least for $\kappa$ large this quantity
  asymptotically coincides with its original definition, which may justify
  the misuse of notation.
\end{remark}

We now show that close to the origin and for large $\kappa$ the behavior
of $p_\kappa$ is captured by $p^\ast_\kappa$.
\begin{lemma} \label{p-pml}
  There exists a constant $C>0$ which depends
  only on $\Phi$ or $P$ respectively
  such that for all $\kappa > 0$ and $r\geq 0$,
  \[
  |p_\kappa(r) - p^\ast_\kappa (r)|
  \leq
  C e^{6\kappa} \left(r^2 + e^{4\kappa} r^4 \right)
  \exp\left(C\left(e^{4\kappa} r^2 + e^{8\kappa} r^4\right)\right) .
  \]
\end{lemma}

\begin{proof}
In order to keep the notation simple we drop the subscript $\kappa$ and
write $p$ and $p^\ast$ for the two solutions which we want to compare.
We introduce rescaled variables as follows:
\be \label{rescale}
p(r) = \alpha^{-2} \sigma(\tau),\
p^\ast(r) = \alpha^{-2} \sigma^\ast(\tau),\ r = \alpha \tau,
\ee
with $\alpha >0$. A straightforward computation shows that
\be \label{sigmaeq}
\sigma'(\tau)=
- \frac{\sigma(\tau) + \alpha^2 S(\alpha^{-2}\sigma(\tau))}
{1-\frac{8\pi}{\tau}\int_0^\tau s^2 \alpha^2 S(\alpha^{-2}\sigma(s))\, ds}
\left(\frac{4\pi}{\tau^2}\int_0^\tau s^2 \alpha^2 S(\alpha^{-2}\sigma(s))\, ds
+ 4 \pi \tau \sigma(\tau)\right) ,
\ee
\be \label{zetaeq}
{\sigma^\ast}'(\tau)=
- \frac{4 \sigma^\ast(\tau)}
{1-\frac{8\pi}{\tau}\int_0^\tau s^2 3 \sigma^\ast(s))\, ds}
\left(\frac{4\pi}{\tau^2}\int_0^\tau s^2 3 \sigma^\ast(s))\, ds
+ 4 \pi \tau \sigma^\ast(\tau)\right) ,
\ee
and
\be\label{etazetaic}
\sigma(0)=\frac{1}{3}=\sigma^\ast(0);
\ee
notice that \eqref{zetaeq} is actually the same equation as \eqref{peqml}.
We choose
\[
\alpha = e^{-2\kappa}.
\]
In order to estimate the difference $\sigma' - {\sigma^\ast}'$
and apply a Gronwall argument we first collect a number of estimates.
In these estimates
$C$ denotes a positive constant which depends only on the ansatz function
$\Phi$ or $P$ and which may change its value from line to line. In particular,
such constants are independent of $\tau$ and $\kappa$.
First we note that $\sigma$ and $\sigma^\ast$ are decreasing, and hence for
$\tau \geq 0$,
\[
0\leq \sigma(\tau), \sigma^\ast(\tau) \leq \frac{1}{3}.
\]
By \eqref{Pass2},
\[
\left|\alpha^2 S\left(\frac{\sigma(\tau)}{\alpha^2}\right)
- 3 \sigma(\tau)\right|
= \alpha^2 \left|S\left(\frac{\sigma(\tau)}{\alpha^2}\right)
- 3 \frac{\sigma(\tau)}{\alpha^2}\right|
\leq C \alpha^2 \left(\frac{\sigma(\tau)}{\alpha^2}\right)^{1/2}
\leq C \alpha = C e^{-2\kappa}.
\]
This implies that
\[
\alpha^2 S(\sigma(\tau)/\alpha^2)\leq C \alpha + 3 \sigma(\tau) \leq C
\]
and
\be \label{S-3}
\left|\alpha^2 S(\sigma(\tau)/\alpha^2) - 3 \sigma^\ast(\tau)\right|
\leq C e^{-2\kappa} + 3 \left|\sigma(\tau) - \sigma^\ast(\tau)\right|.
\ee
A non-trivial issue is to get uniform control on the denominators
in \eqref{sigmaeq} and \eqref{zetaeq}.
But the spherically symmetric steady states
of the Einstein-Vlasov or Einstein-Euler system which we consider here
satisfy the estimate
\begin{equation}\label{buchdahl}
\frac{2 m(r)}{r} < \frac{8}{9},\ r>0,
\end{equation}
which is known as Buchdahl's inequality,
cf.~\cite{Buch}; a proof for very general
matter models which covers the present
situation and also anisotropic steady states is given in
\cite{Andr2008}.
Hence,
\[
\frac{1}{1-\frac{8\pi}{\tau}\int_0^\tau s^2 \alpha^2 S(\alpha^{-2}\sigma(s))\,
ds},
\quad
\frac{1}{1-\frac{8\pi}{\tau}\int_0^\tau s^2 3 \sigma^\ast(s))\, ds}
< 9,\ \tau>0.
\]
Let us abbreviate
\[
x(\tau) := \max_{0\leq s \leq \tau} |\sigma(s) - \sigma^\ast(s)| .
\]
Then these estimates together imply that
\begin{eqnarray*}
&&
|\sigma'(\tau)-{\sigma^\ast}'(\tau)|
\leq
\left|\sigma(\tau) + \alpha^2 S(\alpha^{-2}\sigma(\tau)) -
4 \sigma^\ast(\tau)\right|\\
&&
\qquad\qquad\qquad
\frac{1}{1-\frac{8\pi}{\tau}\int_0^\tau s^2 \alpha^2 S(\alpha^{-2}\sigma(s))\, ds}
\left|\frac{4\pi}{\tau^2}\int_0^\tau s^2 \alpha^2 S(\alpha^{-2}\sigma(s))\, ds
+ 4 \pi \tau \sigma(\tau)\right|\\
&&
\qquad\quad{}+
4 \sigma^\ast(\tau)
\left|\frac{1}{1-\frac{8\pi}{\tau}
\int_0^\tau s^2 \alpha^2 S(\alpha^{-2}\sigma(s))\, ds}
-\frac{1}{1-\frac{8\pi}{\tau}\int_0^\tau s^2 3 \sigma^\ast(s))\, ds} \right|\\
&&
\qquad\qquad\qquad
\left|\frac{4\pi}{\tau^2}\int_0^\tau s^2 \alpha^2 S(\alpha^{-2}\sigma(s))\, ds
+ 4 \pi \tau \sigma(\tau)\right|\\
&&
\qquad\quad{}+
\frac{4 \sigma^\ast(\tau)}
{1-\frac{8\pi}{\tau}\int_0^\tau s^2 3 \sigma^\ast(s))\, ds}\\
&&
\qquad\qquad\qquad
\left|\frac{4\pi}{\tau^2}\int_0^\tau s^2 \alpha^2 S(\alpha^{-2}\sigma(s))\, ds
+ 4 \pi \tau \sigma(\tau)
-\frac{4\pi}{\tau^2}\int_0^\tau s^2 3 \sigma^\ast(s))\, ds
- 4 \pi \tau \sigma^\ast(\tau)\right|\\
&&
\qquad \leq
C \left(e^{-2\kappa} + x(\tau)\right)\,\tau +
C \left(e^{-2\kappa} + x(\tau)\right)\,\tau^2 \tau +
C \left(e^{-2\kappa} + x(\tau)\right)\,\tau\\
&&
\qquad \leq
C (\tau + \tau^3)\,\left(e^{-2 \kappa} + x(\tau) \right).
\end{eqnarray*}
Integration of this estimate yields
\[
x(\tau)\leq C (\tau^2 + \tau^4)\,e^{-2 \kappa}
+ C \int_0^\tau (s + s^3)\, x(s)\, ds,
\]
Gronwall's lemma implies that
\[
x(\tau)\leq C (\tau^2 + \tau^4)\,
\exp\left( C \left(\tau^2 + \tau^4\right)\right)\,e^{-2 \kappa},
\]
by \eqref{rescale},
\beas
\left|p(r) -p^\ast(r)\right|
&\leq&
(\tau^2 + \tau^4)\,\exp\left( C \left(\tau^2 + \tau^4\right)\right)\,e^{2 \kappa}
\\
&=&
(e^{4 \kappa}r^2 + e^{8 \kappa}r^4)\,
\exp\left( C \left(e^{4 \kappa}r^2 + e^{8 \kappa}r^4\right)\right)\,e^{2 \kappa},
\eeas
and the proof is complete.
\end{proof}

In order to understand the behavior of $p^\ast_\kappa$ in the limit
$\kappa \to \infty$ the following observation is useful; we recall
\eqref{piv} and Remark~\ref{repara} to avoid miss-interpretation of the
notation $p^\ast_\kappa$.
\begin{lemma} \label{pmlscale}
  Let $p^\ast_\kappa$ denote the solution of \eqref{peqml} with initial data
  $p^\ast_\kappa(0)=e^{4\kappa}/3$.
  Then for all $\kappa \geq 0$,
  \[
  p^\ast_\kappa (r) = e^{4\kappa}p^\ast_0 (e^{2\kappa} r),\ r\geq 0 .
  \]
\end{lemma}

\begin{proof}
As we noted in the proof of Lemma~\ref{p-pml} the rescaled function
$\sigma^\ast(\tau) = e^{-4\kappa}p^\ast_\kappa(r)$ like $p^\ast_\kappa$
solves \eqref{peqml},
but to the initial data $\sigma^\ast(0)=\frac{1}{3}=p^\ast_0(0)$. Hence
$e^{-4\kappa}p^\ast_\kappa(r)=p^\ast_0(\tau)=p^\ast_0(e^{2\kappa}r)$,
which is the assertion.
\end{proof}

In view of Lemma~\ref{pmlscale} we now need to understand the
behavior of $p^\ast_0 (r)$ for large $r$.

\begin{lemma} \label{dynsys}
  Let $\rho^\ast_0$ and $m^\ast_0$ denote the quantities induced by $p^\ast_0$.
  Then for any $\gamma \in ]0,3/2[$,
  \[
  \left|r^2 \rho^\ast_0 (r) - \frac{3}{56 \pi}\right|
  + \left| \frac{m^\ast_0(r)}{r} - \frac{3}{14}\right|
  \leq C r^{-\gamma},\ r>0.
  \]
\end{lemma}

The limiting constants above are the corresponding
values of the BKZ solution \eqref{bkzsolution}.

\begin{proof}
The proof relies on turning the massless steady state equations
into a planar, autonomous dynamical system.
We recall the Tolman-Oppenheimer-Volkov
equation \eqref{tov}, which as we showed above holds in both the Euler and the
Vlasov case, and combine it with the
massless equation of state $p=\rho/3$ and \eqref{yeq}. This leads to
\begin{eqnarray*}
\frac{d\rho}{dr}
&=&
-\frac{4 \rho}{1-\frac{2 m}{r}}
\left( \frac{m}{r^2} + \frac{4\pi}{3}r \rho \right),\\
\frac{d m}{dr}
&=&
4 \pi r^2 \rho.
\end{eqnarray*}
We rewrite this in terms of
$u_1 (r) = r^2 \rho(r),\ u_2(r) = m(r)/r$ to obtain
\begin{eqnarray*}
\frac{d u_1}{dr}
&=&
\frac{2}{r} u_1 - \frac{4 u_1}{1 - 2 u_2}
\left( \frac{1}{r} u_2 + \frac{4\pi}{3} \frac{1}{r} u_1 \right),\\
\frac{d u_2}{dr}
&=&
 4 \pi \frac{1}{r} u_1 - \frac{1}{r} u_2.
\end{eqnarray*}
Now we multiply both equations with $r$ and introduce
$w_1(\tau) = u_1(r),\ w_2(\tau)= u_2(r)$ with $\tau=\ln r$. Then
\begin{eqnarray}
\frac{d w_1}{d\tau}
&=&
\frac{2 w_1}{1 - 2 w_2}\left(1 - 4 w_2 -\frac{8\pi}{3} w_1 \right),\label{w1}\\
\frac{d w_2}{d\tau}
&=&
 4 \pi w_1 - w_2.\label{w2}
\end{eqnarray}
We denote the right hand side of the system
\eqref{w1}, \eqref{w2} by $F(w)$,
which is defined and smooth for $w_1\in \R$ and $w_2\in ]-\infty,1/2[$.
The system has two steady states:
\[
F(w) = 0\ \Leftrightarrow\ w=(0,0) \ \mbox{or}\
w=Z:=\left(\frac{3}{56\pi},\frac{3}{14}\right).
\]
We aim to show that the solution which corresponds to
$p^\ast_0$ converges to $Z$ exponentially fast as $\tau\to\infty$.
To this end we first observe that
\[
DF(Z) =
\begin{pmatrix}
  -\frac{1}{2}& -\frac{3}{4\pi}\\
  4 \pi & -1
\end{pmatrix}
\]
has eigenvalues
\[
\lambda_{1,2} = -\frac{3}{2} \pm \frac{\sqrt{47}}{2} i
\]
so that $Z$ is an exponential sink. On the other hand,
\[
DF(0,0) =
\begin{pmatrix}
  2 & 0 \\
  4 \pi & -1
\end{pmatrix}
\]
with eigenvalues $-1$ and $2$, stable direction $(0,1)$ and unstable direction
$(3,4\pi)$.
For the solution induced by $p^\ast_0$ it holds that $w(\tau) \to (0,0)$ for
$\tau\to -\infty$ which corresponds to $r\to 0$. Since the corresponding
trajectory lies in the first quadrant $\{w_1 >0,\ w_2>0\}$,
it must coincide with
the corresponding branch $T$ of the unstable manifold of $(0,0)$.
We want to show that the trajectory $T$ approaches the point $Z$.
To this end, let $D$ denote the triangular region bounded by the lines
\[
\{w_1=0\},\ \{w_2 =4/9\},\ \{w_1 = w_2\}.
\]
Clearly, $Z\in D\subset ]0,\infty[\times ]0,1/2[$ so that $D$ lies
in the domain where $F$ is defined and smooth. We want to show that
$T\subset D$. Since the unstable direction $(3,4\pi)$ points
into $D$, this is at least true for the part of $T$ close to
the origin. The line $\{w_1=0\}$ is invariant and can not be crossed by $T$; notice that
this line is also the stable manifold of the point $(0,0)$.
Due to Buchdahl's inequality, $T$ must lie below the line
$\{w_2 =4/9\}$. Finally, along the line $\{w_1 = w_2\}$ the vector
$(-1,1)$ is normal to it and points into the domain $D$.
Since along this line,
\[
(-1,1)\cdot F(w_1,w_1) = -F_1(w_1,w_1) + F_2(w_1,w_1)
= \frac{1}{1-2 w_1}
\left((4 \pi -3)\,w_1 + \frac{30-8\pi}{3} w_1^2\right) >0
\]
so that the vector field $F$ points into the domain $D$,
the trajectory $T$ cannot leave this domain.
Hence according to
Poincar\'{e}-Bendixson theory, the $\omega$-limit set of $T$ must
either coincide with $Z$, or with a periodic orbit. But
according to Dulac's negative criterion, the set $]0,\infty[ \times ]0,1/2[$
does not contain a periodic orbit, since
\[
\mathrm{div}\left(\frac{1}{w_1} F(w)\right) =
-\frac{8\pi}{3} \frac{2}{1-2 w_2} - \frac{1}{w_1} < 0.
\]
In view of the real part of the eigenvalues $\lambda_{1,2}$ it follows that
for any $0<\gamma < 3/2$ and all $\tau$ sufficiently large, i.e., all $r$
sufficiently large,
\[
|w(\tau) - Z| \leq C e^{-\gamma \tau}.
\]
When rewritten in terms of the original variables this is the assertion.
\end{proof}

We now combine the previous three lemmas to show that
for $\kappa$ large the steady state of the massive system which
corresponds to $p_\kappa$ is approximated well by the BKZ solution
on some interval close to the origin;
$\rho_\kappa, m_\kappa, \mu_\kappa,\lambda_\kappa$ denote quantities
induced by $p_\kappa$.

\begin{proposition}\label{P:EVTOBKZ}
There exist parameters $0<\alpha_1 < \alpha_2 < \frac{1}{4}$, $\kappa_0>0$
sufficiently large, and constants $\delta>0$ and $C>0$
such that on the interval
\[
[r_\kappa^1, r_\kappa^2] = [\kappa^{\alpha_1} e^{-2\kappa},\kappa^{\alpha_2} e^{-2\kappa}]
\]
and for any $\kappa\ge \kappa_0$
the following estimates hold:
\[
\left|r^2 \rho_\kappa (r) - \frac{3}{56 \pi}\right|,\
\left|r^2 p_\kappa (r) - \frac{1}{56 \pi}\right|,\
\left|\frac{m_\kappa(r)}{r} -\frac{3}{14}\right|,\
\left|2r\mu_\kappa'-1\right|,\
\left|e^{2\lambda_\kappa}-\frac{7}{4}\right|,\  \left|r \lambda_\kappa'\right|
\leq
C\,\kappa^{-\delta}.
\]
\end{proposition}
\begin{proof}
First we note that by Lemma~\ref{pmlscale}, the radiative equation of state
\eqref{radeqstate}, and Lemma~\ref{dynsys},
\be \label{pml-bkz}
\left| r^2 p^\ast_\kappa (r) - \frac{1}{56 \pi}\right|
=
\frac{1}{3}
\left| r^2 \rho^\ast_\kappa (r) - \frac{3}{56 \pi}\right|
=
\frac{1}{3}
\left| (e^{2\kappa} r)^2 \rho^\ast_0(e^{2\kappa}r) - \frac{3}{56 \pi}\right|
\leq C e^{-2\kappa\gamma} r^{-\gamma}
\ee
so that together with Lemma~\ref{p-pml},
\be \label{p-pbkz}
\left| r^2 p_\kappa (r) - \frac{1}{56 \pi}\right|
\leq
C e^{6\kappa} \left(r^4 + e^{4\kappa} r^6 \right)
\exp\left(C\left(e^{4\kappa} r^2 + e^{8\kappa} r^4\right)\right)
+ C e^{-2\kappa\gamma} r^{-\gamma}.
\ee
Below we will simplify the above right hand side on a suitably chosen
$r$ interval close to the origin, but first we estimate various
other quantities in the same fashion. In order to estimate
the difference in $\rho$ we recall the rescaled variables introduced
in \eqref{rescale} and the estimate \eqref{S-3}. This implies that
\begin{eqnarray}
  \left|\rho_\kappa (r) - \rho^\ast_\kappa(r)\right|
  &=&
  \left|S(p_\kappa (r)) - 3 p^\ast_\kappa (r)\right|
  =
  \alpha^{-2}\left|\alpha^2S(\alpha^{-2}\sigma(\tau)) - 3 \sigma^\ast(\tau)\right|
  \nonumber \\
  &\leq&
  \alpha^{-2}\left(C e^{-2\kappa} + 3 |\sigma(\tau)-\sigma^\ast(\tau)|\right)
  \nonumber \\
  &\leq&
  C e^{2\kappa} \left(1+(\tau^2 + \tau^4) \exp(C(\tau^2 + \tau^4)\right).
  \label{rho-rhoml}
\end{eqnarray}
If we combine this with the definition of $\tau$ and \eqref{pml-bkz}
it follows that
\beas
\left| r^2 \rho_\kappa (r) - \frac{3}{56 \pi}\right|
&\leq&
C e^{2\kappa} r^2 + C e^{6\kappa} \left(r^4 + e^{4\kappa} r^6 \right)
\exp\left(C\left(e^{4\kappa} r^2 + e^{8\kappa} r^4\right)\right)
+ C e^{-2\kappa\gamma} r^{-\gamma}\\
&=:&
E_\kappa(r),
\eeas
where notice that $E_\kappa$ also bounds the right hand side in \eqref{p-pbkz}.
The estimate \eqref{rho-rhoml} implies that
\[
\left|\frac{m_\kappa(r)}{r} - \frac{m^\ast_\kappa(r)}{r} \right|
\leq C e^{2\kappa} r^2 + C e^{6\kappa} \left(r^4 + e^{4\kappa} r^6 \right)
\exp\left(C\left(e^{4\kappa} r^2 + e^{8\kappa} r^4\right)\right),
\]
and combining this with the second estimate in Lemma~\ref{dynsys}
and the scaling property in Lemma~\ref{pmlscale} it follows that
\[
\left|\frac{m_\kappa(r)}{r} -\frac{3}{14}\right| \leq C E_\kappa (r);
\]
integrating the $\rho$ estimate directly yields the same
result, but $E_\kappa$ is only integrable at the origin if we
choose $\gamma < 1$.
Next we use \eqref{lambdadef} to find that
\[
\left|e^{-2\lambda_\kappa} -\frac{4}{7}\right|
=\left|-2 \frac{m_\kappa(r)}{r} + 2\frac{3}{14}\right| \leq C E_\kappa (r),
\]
and using Buchdahl's inequality \eqref{buchdahl},
\[
\left|e^{2\lambda_\kappa} -\frac{7}{4}\right|
= e^{2\lambda_\kappa} \frac{7}{4} \left|e^{-2\lambda_\kappa} -\frac{4}{7}\right|
 \leq C E_\kappa (r).
\]
Using \eqref{eemu} for $\mu_\kappa'$ and $\mu_{\mathrm{BKZ}}'$ it follows that
\begin{eqnarray*}
\left| r \mu_\kappa' (r) -\frac{1}{2}\right|
&\leq&
e^{2\lambda_\kappa} \left(\left|\frac{m_\kappa(r)}{r}-\frac{3}{14}\right| +
\left| 4 \pi r^2 p_\kappa(r) - \frac{1}{14}\right|\right)
+
\left(\frac{3}{14} + \frac{1}{14}\right)
\left|e^{2\lambda_\kappa} -\frac{7}{4}\right|\\
&\leq&
C\, E_\kappa (r),
\end{eqnarray*}
where we have again used Buchdahl's inequality \eqref{buchdahl};
note that $\mu_{\mathrm{BKZ}}'$ refers to the BKZ solution~\eqref{bkzsolution}.
Finally, with the established bounds on $r^2\rho_\kappa$ and $e^{2\l_\kappa}$ above,
equation \eqref{ein1} immediately gives
\[
\left| r \lambda_\kappa' (r) - 0\right|
\leq
C\,E_\kappa (r).
\]
We now analyze the error term $E_\kappa(r)$ for $r>0$ and
$\kappa>0$ such that
\[
r_1(\kappa) \leq r e^{2\kappa} \leq r_2(\kappa),
\]
where $0<r_1(\kappa) < r_2(\kappa)$ and $r_2(\kappa) >1$.
Then
\[
E_\kappa(r)
\leq C r_1(\kappa)^{-\gamma} + C r_2 (\kappa)^6 e^{-2\kappa} e^{C r_2 (\kappa)^4}
\leq  C e^{-\gamma \ln r_1(\kappa)} + C e^{C^\ast r_2 (\kappa)^4 -2\kappa}.
\]
We choose
\[
r_1(\kappa) = \kappa^{\alpha_1} \ \mbox{with some}\ 0<\alpha_1 < \frac{1}{4}
\]
and
\[
r_2(\kappa)
= \left(\frac{2\kappa - \alpha_1 \gamma \ln \kappa}{C^\ast}\right)^{1/4}.
\]
Then for $\kappa$ large, $r_1(\kappa) < r_2(\kappa)$ and $1< r_2(\kappa)$
as required,
\[
E_\kappa (r)
\leq C \kappa^{-\gamma \alpha_1}\ \mbox{for}\
\kappa^{\alpha_1} e^{-2\kappa} \leq r \leq
\left(\frac{2\kappa - \alpha_1 \gamma \ln \kappa}{C^\ast}\right)^{1/4}
e^{-2\kappa},
\]
and with $\delta := \gamma \alpha_1$ and any
$\alpha_2 \in ]\alpha_1,\frac{1}{4}[$ the proof is complete.
\end{proof}

In the above proposition no information is provided for $\mu_\kappa$ itself.
Such an estimate is derived next.
\begin{proposition}\label{muass}
  Under the assumptions of Proposition~\ref{P:EVTOBKZ} it follows that
  \[
  C_\kappa \exp\left(-C \kappa^{-\delta}\ln \kappa\right)
  \leq \frac{e^{2\mu_\kappa(r)}}{r} \leq
  C_\kappa \exp\left(C \kappa^{-\delta}\ln \kappa\right)
  \]
  on the $r$ interval specified in Proposition~\ref{P:EVTOBKZ}.
  Here $C_\kappa >0$ does depend on $\kappa$, but $C>0$ does not.
  The point in this estimate is that the exponential terms on both sides
  converge to $1$ as $\kappa\to \infty$.
\end{proposition}
\begin{proof}
Clearly,
\begin{eqnarray*}
  2 \mu_\kappa(r)
  &=&
  2 \mu_\kappa(r_\kappa^1) +\int_{r_\kappa^1}^r \frac{1}{s} ds
  + \int_{r_\kappa^1}^r \left(2\mu_\kappa'(s) - \frac{1}{s}\right) ds\\
  &=&
  2 \mu_\kappa(r_\kappa^1) + \ln\left(\frac{r}{r_\kappa^1}\right)
  + \int_{r_\kappa^1}^r \left(2\mu_\kappa'(s) - \frac{1}{s}\right) ds.
\end{eqnarray*}
Hence
\begin{eqnarray*}
  e^{2 \mu_\kappa(r)}
  &=&
  e^{2 \mu_\kappa(r_\kappa^1)} \frac{r}{r_\kappa^1}
  \exp\left(\int_{r_\kappa^1}^r \left(2\mu_\kappa'(s) - \frac{1}{s}\right) ds\right).
\end{eqnarray*}
Using Proposition~\ref{P:EVTOBKZ} it follows that
\[
\exp\left(\int_{r_\kappa^1}^r \left(2\mu_\kappa'(s) - \frac{1}{s}\right) ds\right)
\leq
\exp\left(\int_{r_\kappa^1}^{r_\kappa^2}
\left(C \kappa^{-\delta} \frac{1}{s}\right) ds\right)
= \left(\frac{r_\kappa^2}{r_\kappa^1}\right)^{C \kappa^{-\delta}}
= e^{C (\alpha_2 - \alpha_1)\kappa^{-\delta}\ln\kappa}
\]
as desired, and the lower estimate is completely analogous.
\end{proof}

\section{Stability analysis for the Einstein-Vlasov system}\label{sec_stab_EV}
\setcounter{equation}{0}

In this section we fix some ansatz function $\Phi$, but we need to strengthen
the assumptions on $\Phi$ as follows.

\smallskip

\noindent
{\bf Assumption ($\Phi$2).} $\Phi$ satisfies $(\Phi\,1)$,
and $\Phi\in C^1([0,\infty[)\cap C^2(]0,\infty[)$
with $\Phi'(\eta)>0$ for $\eta>0$.

\smallskip

\noindent
A typical example is
\[
\Phi(\eta) = c\, \eta^k \ \mbox{for}\ \eta \geq 0,\
\Phi(\eta)=0\ \mbox{for}\ \eta < 0
\]
with constants $c>0$ and $1\leq k<3/2$;
notice that we require the right hand side derivative at $0$ to exist,
but it need not vanish.

We consider a corresponding steady state
$(f_\kappa,\l_\kappa,\mu_\kappa)$ as obtained in Proposition~\ref{ssfamilies}
with some $\kappa >0$.
We aim to prove that this steady state is linearly unstable when $\kappa$
is sufficiently large.
In order to keep the notation short, we write
\[
f_\kappa(x,v) = \Phi\left(1-\frac{E_\kappa}{E_\kappa^0}\right)
= \phi_\kappa(E_\kappa),
\]
where we recall the definition of the particle energy
\[
E_\kappa = e^{\mu_\kappa(x)}\p=e^{\mu_\kappa(x)}\sqrt{1+|v|^2}
\]
and note the fact that the cut-off energy depends on $\kappa$
as well; $E^0_\kappa := e^{y_\kappa(\infty)} = e^{\mu_\kappa(R_\kappa)}$
where $R_\kappa$ is the radius of the spatial support of the steady state,
cf.~Section~\ref{setup}. By ($\Phi$2), $\phi_\kappa\in C^1([0,E^0_\kappa])$ with
\be \label{phiprop}
\phi_\kappa'(E) <0\ \mbox{for}\ 0\leq E<E_\kappa^0,
\ \phi_\kappa'(E) =0\ \mbox{for}\ E>E_\kappa^0.
\ee
The basic strategy on which our instability result relies is
the following. The Einstein-Vlasov system conserves the ADM mass.
The second variation of the ADM mass at the given steady state
is a conserved quantity for the linearized system, restricted
to linearly dynamically accessible states. The key to linear
instability then is to prove that there is a negative energy direction
for the linearized system, i.e., a linearly dynamically accessible state
on which the second variation of the ADM mass is negative.
This is the content of Theorem~\ref{T:NEGATIVEA}.

This fact rather directly implies a linear, exponential instability result,
cf.~Theorem~\ref{lininstabEV}. However, much more precise instability
information, including the existence of a simple growing mode,
can be derived from Theorem~\ref{T:NEGATIVEA}. In order to do so the
Hamiltonian character of the linearized system must be exploited
which leads to the desired spectral properties of the generator of
the $C^0$ group induced by the linearized system,
cf.~Theorem~\ref{thm:main VE}. In order to emphasize the basic mechanism
leading to our instability result we try to bring in the more functional
analytic, abstract tools and terminology only when they are finally
needed to derive Theorem~\ref{thm:main VE}.

\subsection{Dynamic accessibility and the linearized Einstein-Vlasov system}
\label{linEV}

Sufficiently regular solutions of the spherically symmetric
Einstein-Vlasov system conserve the ADM mass
\[ 
\mathcal H_{\text{ADM}}(f) = \iint \p f(x,v)\,dv\,dx.
\]
In addition, the Einstein-Vlasov flow conserves the  Casimir functionals
\be \label{Casimir}
\mathcal C (f) = \iint e^{\l_f}\chi(f)\,dv\,dx.
\ee
Here $\chi\in C^1(\mathbb R)$ is an arbitrary, prescribed function with
$\chi(0)=0$, and $\l_f$ is defined by
\be \label{lambdaf}
e^{-2\lambda_f} = 1 - \frac{2  m_f (r)}{r} =
1 - \frac{8\pi}{r}\int_{|y|\leq r}\int \p f(y,v)\,dv\,dy,
\ee
which is the solution of \eqref{ein1} induced by $f$
and satisfying \eqref{boundc0}, cf.~\eqref{lambdadef}.
In a stability analysis it is natural to restrict the admissible
perturbations of the given steady state to such as preserve all the
Casimir invariants. These {\em dynamically accessible} perturbations
form the symplectic leaf $\mathcal S_{f_\kappa}$ through the given steady state.
Its formal tangent space at $f_\kappa$ is the set of
{\em linearly dynamically accessible perturbations}~\cite{KM}.
In order to proceed we need to recall the usual Poisson bracket
\[
\{f,g\}:=\nabla_x f \cdot\nabla_v g-\nabla_v f \cdot\nabla_x g
\]
of two continuously differentiable functions $f$ and $g$ of
$x,v \in \R ^3$. The product rule for the Poisson bracket reads
\[
\{f,g h\}= \{f,g\} h + \{f,h\} g.
\]
In addition we denote the radial component of a vector $v\in \R^3$ by
\be\label{E:WDEF}
w:= \frac{x\cdot v}{r},\ x,v \in \R^3,\ r=|x|.
\ee
An important subclass of linearly dynamically accessible states are of the form
\be\label{dynacdef}
f_h
= \phi_\kappa'(E_\kappa) \left(e^{-\lambda_\kappa}\{h,E_\kappa\}
+  e^{\mu_\kappa} \frac{w^2}{\p} \lambda_h\right)
= - \mathcal B_\kappa (\phi_\kappa' h);
\ee
the operator $\mathcal B_\kappa$ is defined in Definition~\ref{defXT},
the concept of linearly dynamically accessible states
which we adopt later is made precise in Definition~\ref{D:DYNACCDEF}.
The generating function $h=h(x,v)$ of the perturbation should be
spherically symmetric and $C^1$. In this case, we denote the perturbed metric
field $\l$ by $\l_h$ and we have
\be \label{dynacl}
\lambda_h = 4\pi r  e^{\mu_\kappa+\lambda_\kappa}
\int \phi_\kappa'(E_\kappa)\,h(x,v)\,w\,dv.
\ee
On linearly dynamically accessible states of the form~\eqref{dynacdef}
the second variation of the ADM mass $\mathcal H_{\text{ADM}}$ takes the form
\begin{align} \label{E:ADEF}
\mathcal{A}_\kappa(h,h)
=
\iint \frac{e^{\lambda_\kappa}}{|\phi_\kappa'(E_\kappa)|}(f_h)^2\,dv\,dx
-\int_0^\infty e^{\mu_\kappa - \lambda_\kappa}
\left(2 r \mu_\kappa' +1\right)\, (\lambda_h)^2\,dr,
\end{align}
see~\cite{HaRe2013}.

\begin{remark}
We emphasize that the above variational structure and the discussion of dynamic accessibility is limited to
radially symmetric perturbations; a detailed derivation is given in~\cite{HaRe2013}. However, in our stability analysis of the
steady states of the EV-system
variational methods do not play a role.
\end{remark}

The relation of the quadratic form $\mathcal{A}_\kappa$ to the non-linear
Einstein-Vlasov system is not relevant for the present analysis,
but it will become so in a possible extension of the linear instability
result to the non-linear system. For the present purposes it is important
that $\mathcal{A}_\kappa$ is conserved along the linearized flow, which we
need to recall next for the special class of perturbations of the
form~\eqref{dynacdef}.
The dynamics of the generating function $h$ is determined by
\begin{align} \label{heq}
  \partial_t h + e^{-\lambda_\kappa}\{h,E_\kappa\} +
  e^{\mu_\kappa} \lambda_h \frac{w^2}{\p} - e^{\mu_\kappa}\p \mu_h = 0,
\end{align}
where $\lambda_h$ and $\mu_h$ are determined by
\begin{eqnarray}
e^{-2\lambda_\kappa}(r \lambda_h'- \lambda_h\,(2r\lambda_\kappa'-1))
&=&
4\pi r^2 \rho_h, \label{lineelambda}\\
e^{-2\lambda_\kappa}(r \mu_h'- \lambda_h\,(2r\mu_\kappa'+1))
&=&
4\pi r^2 p_h,\label{lineemu}
\end{eqnarray}
together with the boundary conditions
$\lambda_h(t,0)=0=\mu_h(t,\infty)$,
and
\begin{eqnarray}
\rho_h(t,r)
&=&
\int \p f_h(t,x,v)\,dv, \label{linrhodef}\\
p_h(t,r)
&=&
\int \frac{w^2}{\p } f_h(t,x,v)\,dv.\label{linpdef}
\end{eqnarray}
It can be checked from~\eqref{lineelambda} that $\l_h$ is indeed
given by the formula~\eqref{dynacl}, cf.~\cite{HaRe2013}.

Let us denote
\be \label{suppss}
D_\kappa := \{(x,v)\in \R^6 \mid E_\kappa(x,v) < E_\kappa^0\}
\ee
so that $\overline{D}_\kappa$ is the support of the steady state
under consideration.
A simple iteration argument the details of which are indicated
in \cite{HaRe2013} shows that for any $\mathring{h} \in C^1(\overline{D}_\kappa)$
there exists a unique solution
$h\in C^1([0,\infty[;C(\overline{D}_\kappa)) \cap C ([0,\infty[;C^1(\overline{D}_\kappa))$
of the above system
with $h(0)=\mathring{h}$; in passing we note that Eqn.~(5.10) in
\cite{HaRe2013}, which corresponds to \eqref{heq}, contains a sign error.
Moreover, we note that the characteristic flow of \eqref{heq}
is the one of the steady state solution and leaves the set
$\overline{D}_\kappa$ invariant.
Under the present regularity assumptions on $\Phi$ this characteristic flow
is $C^3$, and the limiting factor concerning regularity
of solutions to the above system is the factor $\phi_\kappa'$
in \eqref{dynacdef}.
We recall from \cite{HaRe2013} that the energy
$\mathcal{A}_\kappa(h,h)$ is conserved along solutions of the linearized
system stated above.

\begin{remark} \label{suppfh}
  It should be noted that by \eqref{phiprop} and \eqref{dynacdef}
  the dynamically accessible perturbation $f_h$ induced by some generator
  $h$ vanishes outside $\overline{D}_\kappa$ and that only the values
  of $h$ on that set matter for the definition of $f_h$ and
  $\mathcal{A}_\kappa(h,h)$.
\end{remark}

In \cite{HaRe2014} it was shown that for $\kappa$ sufficiently small
the quantity $\mathcal{A}_\kappa$ is positive definite, which leads
to a linear stability result. Here we aim for the opposite.

\subsection{A negative energy direction for $\kappa$ sufficiently large}
\label{SS:NEGATIVE}

Our next aim is to show the existence of a linearly dynamically accessible
perturbation for which the bilinear form $\mathcal A_\kappa$ is negative.
\begin{theorem}\label{T:NEGATIVEA}
There exists $\kappa_0>0$ such that for all $\kappa>\kappa_0$
there exists a spherically symmetric function
$h\in C^2(\R^6)$
which is odd in $v$, such that
\[
\mathcal A_\kappa (h,h):= \mathcal A_\kappa (f_h,f_h) < 0,
\]
where $f_h$ is given by~\eqref{dynacdef}.
\end{theorem}
For the proof of this result we first establish two auxiliary results.
It will be convenient to use, in addition to the radial momentum component
$w=x\cdot v/r$, also the modulus of the tangential momentum component, i.e.,
\[
\omega := \left|\frac{x\times v}{r}\right|;\ |v|^2 = w^2 + \omega^2;
\]
$L= r^2 \omega^2$ is the modulus of angular momentum squared mentioned
in the introduction.
\begin{lemma}\label{L:ABIDENTITY}
The following identities hold:
\begin{align}
\int |\phi_\kappa'(E_\kappa)|\frac{w^4}{\p^2}\,dv
& =
3 e^{-\mu_\kappa} p_\kappa
- e^{-\mu_\kappa}\int \phi_\kappa(E_\kappa)\frac{w^4}{\p^3}\,dv,
\label{31}\\ 
\int |\phi_\kappa'(E_\kappa)|\,w^2\,dv
& =
\frac{1}{2} \int |\phi_\kappa'(E_\kappa)|\,\omega^2\,dv
= e^{-\mu_\kappa}(\rho_\kappa + p_\kappa),
\label{11} \\
\int |\phi_\kappa'(E_\kappa)|\frac{\omega^4}{\p^2}\,dv
& =
8 e^{-\mu_\kappa} p_\kappa
- e^{-\mu_\kappa}\int \phi_\kappa(E_\kappa)\frac{\omega^4}{\p^3}\,dv,
\label{31om}\\
\int |\phi_\kappa'(E_\kappa)|\frac{w^2 \omega^2}{\p^2}\,dv
& =
2 e^{-\mu_\kappa} p_\kappa
- e^{-\mu_\kappa}\int \phi_\kappa(E_\kappa)\frac{w^2 \omega^2}{\p^3}\,dv
\label{mixed}.
\end{align}
\end{lemma}

\begin{proof}
  First we observe that all the integrands depend only on $w$ and $\omega$
  so we can change variables using
  \[
  dv = 2\pi\, \omega\, d\omega\, dw.
  \]
  With $E_\kappa =e^{\mu_\kappa} \sqrt{1+w^2+\omega^2}$ it follows that
  \[
  \partial_w \phi_\kappa(E_\kappa)=\phi_\kappa'(E_\kappa)\,e^{\mu_\kappa}\frac{w}{\p},\
  \partial_\omega \phi_\kappa(E_\kappa)=\phi_\kappa'(E_\kappa)\,
  e^{\mu_\kappa}\frac{\omega}{\p}.
  \]
  Using the first identity for the integrals containing $w^2$ or $w^4$ and the
  second one for the others, the assertions follow by integration by parts with
  respect to $w$ or $\omega$ and the definitions of $\rho_\kappa$ and $p_\kappa$;
  no boundary terms arise since $\phi_\kappa(E_\kappa)$ vanishes on the boundary
  of its support and the powers of $\omega$ vanish at $\omega = 0$.
\end{proof}
\begin{lemma}\label{phiprimeint}
  On the interval $[r_\kappa^1, r_\kappa^2]$ introduced in
  Proposition~\ref{P:EVTOBKZ} and for $\kappa$ sufficiently large,
  \[
  \int |\phi_\kappa'(E_\kappa)|\,dv
  \leq \frac{1}{4} e^{-\mu_\kappa} p_{\kappa}.
  \]
\end{lemma}
\begin{proof}
  The same change of variables which yields the formulas \eqref{rhoyrelv}
  implies that
  \[
  \int |\phi_\kappa'(E_\kappa)|\,dv
  = e^{-\mu_\kappa} 4\pi e^{2 y_\kappa} \int_0^{1-e^{-y_\kappa}} \Phi'(\eta)\, (1-\eta)\,
  \sqrt{(1-\eta)^2 - e^{-2y_\kappa}}\, d\eta;
  \]
  the point here is that only the factor $e^{2 y_\kappa}$ appears, while
  for $p_\kappa$ (and $\rho_\kappa$) the factor is $e^{4 y_\kappa}$; recall also that
  $\phi_\kappa(E_\kappa)=\Phi(1-E_\kappa/E_\kappa^0)$, cf.~\eqref{trueansatz}. Thus
  \[
  \int |\phi_\kappa'(E_\kappa)|\,dv \leq e^{-\mu_\kappa} 4\pi e^{2 y_\kappa}
  \int_0^1 \Phi'(\eta)\, (1-\eta)^2  d\eta =: e^{-\mu_\kappa} C e^{2 y_\kappa}
  \leq e^{-\mu_\kappa} C e^{2 \kappa};
  \]
  the constant $C$ is finite, since by Assumption $(\Phi2)$ the last integral
  converges.
  On the other hand, Proposition~\ref{P:EVTOBKZ} implies that on the interval
  $[r_\kappa^1, r_\kappa^2]$ and for $\kappa$ sufficiently large,
  \[
  p_\kappa(r) > \frac{1}{57\pi} r^{-2} \geq \frac{1}{57\pi} (r_\kappa^2)^{-2}
  = \frac{1}{57\pi} \kappa^{-2\alpha_2} e^{4\kappa}
  \]
  with $\alpha_2>0$ from Proposition~\ref{P:EVTOBKZ}. Comparing the two
  estimates yields the assertion.
\end{proof}
\noindent
{\em Proof of Theorem~\ref{T:NEGATIVEA}.}\quad
We shall look for a negative energy direction $h$ of the form
\be\label{E:HG}
h(x,v) = g(r) w,
\ee
where $g\in C^2([0,\infty[)$ is to be chosen. We recall the definition of
the coordinate $w$~\eqref{E:WDEF}.
As a function of $v$, $h$ is clearly odd.
With~\eqref{E:HG} and \eqref{dynacl} in mind
we can simplify the expression for $\lambda_h$:
\begin{align*}
\lambda_{h}& = 4\pi r  e^{\mu_\kappa+\lambda_\kappa}g
\int \phi_\kappa'(E_\kappa)\,w^2\,dv \nonumber \\
& = 4\pi r  e^{\mu_\kappa+\lambda_\kappa}
g \left(-\frac{e^{-2\lambda_\kappa-\mu_\kappa}}{4\pi  r} 
\left(\lambda_\kappa'+\mu_\kappa'\right)\right)\nonumber \\
& =- e^{-\lambda_\kappa}g\left(\lambda_\kappa'+\mu_\kappa'\right),
\end{align*}
where we used \eqref{11} and \eqref{laplusmu}.
On the other hand,
\begin{align*}
\{h,f_\kappa\}
&= \phi_\kappa'(E_\kappa)
\left(\nabla_x h \nabla_v E_\kappa - \nabla_v h \nabla_x E_\kappa\right)
\nonumber\\
& = \phi_\kappa'(E_\kappa)\, e^{\mu_\kappa}
\left( g'(r)\frac{w^2}{\p} - \mu_\kappa'g(r)\p + g(r)\frac{|v|^2-w^2}{r\p}\right).
\end{align*}
Therefore
\begin{align*}
f_{h} = e^{\mu_\kappa-\lambda_\kappa}
\phi_\kappa'(E_\kappa)\left((g'-g(\mu_\kappa'+\lambda_\kappa'))
\frac{w^2}{\p}-\mu_\kappa'g\p+g\frac{|v|^2-w^2}{r\p}\right).
\end{align*}
According to Proposition~\ref{P:EVTOBKZ},
on the spatial interval $[r_\kappa^1, r_\kappa^2]$ the steady state
$(f_\kappa,\lambda_\kappa,\mu_\kappa)$ is well approximated by the BKZ
solution of the massless Einstein-Vlasov system, provided $\kappa$
is sufficiently large. 
Hence we localize the perturbation $h$ given by \eqref{E:HG}
to this interval by setting
\[
g = e^{\frac12\mu_\kappa+\lambda_\kappa}\chi_{\kappa},
\]
where $0\leq \chi_\kappa\leq 1$
denotes a smooth cut-off function supported in the interval
$[r_\kappa^1, r_\kappa^2]$ and identically equal to $1$ on the interval 
$[2 r_\kappa^1, r_\kappa^2/2]$; note that the latter interval is non-trivial
for $\kappa$ sufficiently large. In addition, we require that
\be \label{chipbounds}
|\chi_\kappa'(r)|\le \frac{2}{r_\kappa^1}\ \mbox{for}\
r\in [r_\kappa^1, 2 r_\kappa^1],\quad
|\chi_\kappa'(r)| \le \frac{4}{r_\kappa^2}\
\mbox{for}\ r\in [r_\kappa^2/2,r_\kappa^2].
\ee
Then the perturbation $f_{h}$ takes the form
\begin{align}\label{E:DELTAFSIMPLER}
f_{h} = e^{\frac32\mu_\kappa} \phi_\kappa'(E_\kappa)
\left(-\mu_\kappa' \chi_\kappa\left[\frac{w^2}{2\p}
  +\p-\frac{1}{\mu_\kappa'r}\frac{|v|^2-w^2}{\p}\right]+
\chi_\kappa'\frac{w^2}{\p}\right).
\end{align}
We plug \eqref{E:DELTAFSIMPLER} into~\eqref{E:ADEF}
and obtain the following identity:
\begin{align}
&
\mathcal{A}_\kappa(h,h) \notag \\
&
= \int_{r_\kappa^1}^{r_\kappa^2} e^{2\mu_\kappa-\lambda_\kappa} \chi_\kappa^2
\biggl[ 4\pi r^2 e^{\mu_\kappa+2\lambda_\kappa}(\mu_\kappa')^2
\int |\phi_\kappa'(E_\kappa)|
\left(\frac{w^2}{2\p}+\p-\frac{1}{\mu_\kappa'r}\frac{|v|^2-w^2}{\p}\right)^2\,dv
\notag \\
&
\qquad \qquad \qquad \qquad  - (2r\mu_\kappa'+1)
(\mu_\kappa'+\lambda_\kappa')^2\biggr]\,dr \notag\\
&
\quad
+ 4\pi \int_{\text{supp}\chi_\kappa'} r^2 e^{3\mu_\kappa+\lambda_\kappa}(\chi_\kappa')^2
\int |\phi_\kappa'(E_\kappa)|\frac{w^4}{\p^2}\,dv\,dr  \notag \\
&
\quad
+ 8\pi \int_{\text{supp}\chi_\kappa'}
r^2 e^{3\mu_\kappa+\lambda_\kappa} \mu_\kappa' \chi_\kappa \chi_\kappa'
\int |\phi_\kappa'(E_\kappa)| \left(\frac{w^4}{2\p^2}+w^2
-\frac{1}{\mu_\kappa'r}\frac{w^2|v|^2-w^4}{\p^2}\right)\,dv\,dr \notag\\
&= : A_1 +A_2 +A_3. \label{E:ABREAKDOWN}
\end{align}
The idea now is as follows: If we replace the steady state
$(f_\kappa,\lambda_\kappa,\mu_\kappa)$ in $A_1$ by the corresponding
limiting quantities according to Proposition~\ref{P:EVTOBKZ},
then a strictly negative term arises, together with various error
terms, which like  $A_2$ and $A_3$ do not destroy the overall negative
sign of $\mathcal{A}_\kappa(h,h)$, provided $\kappa$ is sufficiently
large.

Since by Proposition~\ref{P:EVTOBKZ} $\frac{1}{r \mu_\kappa'} \approx 2$ on the support of $\chi_\kappa$,
we split $A_1$ further:
\begin{align}
A_1 
&
= \int_{r_\kappa^1}^{r_\kappa^2} e^{2\mu_\kappa-\lambda_\kappa} \chi_\kappa^2
\biggl[ 4\pi r^2 e^{\mu_\kappa+2\lambda_\kappa}(\mu_\kappa')^2
\int |\phi_\kappa'(E_\kappa)|
\left(\frac{w^2}{2\p}+\p-2\frac{|v|^2-w^2}{\p}\right)^2 dv
\notag \\
&
\qquad \qquad \qquad \qquad  - (2r\mu_\kappa'+1)
(\mu_\kappa'+\lambda_\kappa')^2\biggr]\,dr \notag\\
&
\quad + \int_{r_\kappa^1}^{r_\kappa^2} e^{2\mu_\kappa-\lambda_\kappa} \chi_\kappa^2
4\pi r^2 e^{\mu_\kappa+2\lambda_\kappa}(\mu_\kappa')^2
\int |\phi_\kappa'(E_\kappa)|
\Biggl[
\left(\frac{w^2}{2\p}+\p-\frac{1}{\mu_\kappa'r}\frac{|v|^2-w^2}{\p}\right)^2
\notag \\
&
\qquad \qquad \qquad \qquad
- \left(\frac{w^2}{2\p}+\p-2\frac{|v|^2-w^2}{\p}\right)^2\Biggr]\, dv\,dr \notag\\
&
=: A_{11} + A_{12}.
\end{align}
Recalling the definition of $\omega$ some algebra implies that
\[
\left[\frac{w^2}{2\p}+\p-2\frac{|v|^2-w^2}{\p}\right]^2
=
1+ \frac{w^4}{4\p^2} + 2 w^2 -3 \omega^2 + 4 \frac{\omega^4}{\p^2}
-2\frac{w^2 \omega^2}{\p^2}. 
\]
If we apply Lemma~\ref{L:ABIDENTITY} and Lemma~\ref{phiprimeint} and collect terms it
follows that
\begin{align*}
&
\int |\phi_\kappa'(E_\kappa)|
\left[\frac{w^2}{2\p}+\p-2\frac{|v|^2-w^2}{\p}\right]^2 dv\\
&\qquad\qquad
\leq e^{-\mu_\kappa}\left( 25 p_{\kappa} - 4 \rho_\kappa\right)
- e^{-\mu_\kappa}\int \phi_\kappa(E_\kappa)\frac{1}{\p^3}
\left[\frac{1}{4} w^4 + 4 \omega^4 - 2 w^2 \omega^2\right]\,dv\\
&\qquad\qquad
= 4 e^{-\mu_\kappa}\left(p_{\kappa} + \rho_\kappa\right)
+ e^{-\mu_\kappa}\left(21 p_{\kappa} -8 \rho_\kappa\right)
-e^{-\mu_\kappa}\int \phi_\kappa(E_\kappa)\frac{1}{\p^3}\,
\left[\frac{1}{2} w^2 - 2 \omega^2\right]^2 dv\\
&\qquad\qquad
\leq 4 e^{-\mu_\kappa}\left(\rho_\kappa + p_\kappa\right) - 3 e^{-\mu_\kappa} p_\kappa;
\end{align*}
in the last estimate we used the fact that $3 p_\kappa \leq \rho_\kappa$.
Using \eqref{laplusmu} it follows that
\begin{align}
A_{11}
&
< - 12 \pi \int_{r_\kappa^1}^{r_\kappa^2} e^{2\mu_\kappa+ \lambda_\kappa} \chi_\kappa^2 r^2
\left(\mu_\kappa'\right)^2 p_\kappa dr\notag \\
&
\quad + \int_{r_\kappa^1}^{r_\kappa^2} e^{2\mu_\kappa- \lambda_\kappa}
\chi_\kappa^2 \left(\lambda_\kappa'+\mu_\kappa'\right)
\left(4 r \left(\mu_\kappa'\right)^2 - \left(\lambda_\kappa'+\mu_\kappa'\right)
\left(2 r \mu_\kappa'+1\right)\right)\, dr\notag \\
&
=:A_{111} + A_{112}.
\end{align}
In order to estimate the various parts into which
$\mathcal{A}_\kappa(h,h)$ has now been split we observe
that by Proposition~\ref{P:EVTOBKZ} the following estimates hold,
provided $\kappa$ is sufficiently large:
\be \label{bounds}
1\leq e^{\lambda_\kappa} \leq 2,\ \frac{1}{4}\leq r \mu_\kappa' \leq 1,\
|r\lambda'_\kappa| \leq 1,\
\frac{1}{57\pi}\leq r^2 p_\kappa,\ r^2\rho_\kappa \leq 1.
\ee
Hence using the fact that $\chi_\kappa=1$ on $[2r_\kappa^1,r_\kappa^2/2]$ and
Proposition~\ref{muass},
\bea
A_{111}
&\leq&
- 12 \pi \int_{2 r_\kappa^1}^{r_\kappa^2/2} e^{2\mu_\kappa+ \lambda_\kappa} r^2
\left(\mu_\kappa'\right)^2 p_\kappa dr \nonumber \\
&\leq&
- C C_\kappa e^{-C \kappa^{-\delta}\ln \kappa}\int_{2 r_\kappa^1}^{r_\kappa^2/2} \frac{dr}{r}
\nonumber\\
&=&
-  C C_\kappa e^{-C \kappa^{-\delta}\ln \kappa} \ln(r_\kappa^2/(4r_\kappa^1))\nonumber\\
&=&
-  C C_\kappa e^{-C \kappa^{-\delta}\ln \kappa} (\ln \kappa - \ln 4)\nonumber\\
&\leq& 
- C C_\kappa e^{-C \kappa^{-\delta}\ln \kappa} \ln \kappa.\label{A111est}
\eea
Here and in what follows $C_\kappa$ always denotes the constant introduced
in Proposition~\ref{muass}, while $C$ denotes a generic constant which
does not depend on $\kappa$, but which may change its value from line
to line. 

We now proceed to prove that all the remaining terms 
are smaller in modulus than the negative term just obtained, provided
$\kappa$ is sufficiently large. In order to estimate $A_{112}$ we note that
by Proposition~\ref{P:EVTOBKZ},
\begin{align*}
4 r \left(\mu_\kappa'\right)^2 - & \left(\lambda_\kappa'+ \mu_\kappa'\right)
\left(2 r \mu_\kappa'+1\right)\\
&
\leq
4 r \left(\frac{1}{2r} + \frac{C}{r}\kappa^{-\delta}\right)^2 -
\left(\frac{1}{2r} - \frac{C}{r}\kappa^{-\delta}\right)
\left(1- C \kappa^{-\delta}+1\right)\\
&
= \frac{1}{r}
\left( \left(1 + C \kappa^{-\delta}\right)^2 -
\left(\frac{1}{2} - C\kappa^{-\delta}\right)
\left(2- C \kappa^{-\delta}\right)\right)\\
&
\leq \frac{C}{r} \kappa^{-\delta}.
\end{align*}
Combining this with \eqref{bounds} and Proposition~\ref{muass},
\be \label{A112est}
A_{112} \leq C C_\kappa e^{C \kappa^{-\delta}\ln \kappa}
\kappa^{-\delta} \int_{r_\kappa^1}^{r_\kappa^2}\frac{dr}{r}
= C C_\kappa e^{C \kappa^{-\delta}\ln \kappa} \kappa^{-\delta} \ln\kappa.
\ee
In order to estimate $A_{12}$ we first observe that
\begin{align*}
&
\left|
\left(\frac{w^2}{2\p}+ \p-\frac{1}{\mu_\kappa'r}\frac{|v|^2-w^2}{\p}\right)^2
- \left(\frac{w^2}{2\p}+\p-2\frac{|v|^2-w^2}{\p}\right)^2\right|\\
&
\qquad =
\Biggl|
- 2 \left(\frac{w^2}{2\p}+\p\right) \frac{1}{\mu_\kappa'r}\frac{|v|^2-w^2}{\p}
+ \left(\frac{1}{\mu_\kappa'r}\right)^2 \left(\frac{|v|^2-w^2}{\p}\right)^2\\
&
\qquad\qquad {}+ 2 \left(\frac{w^2}{2\p}+\p\right) 2 \frac{|v|^2-w^2}{\p}
- 4 \left(\frac{|v|^2-w^2}{\p}\right)^2 \Biggr|\\
&
\qquad \leq
\left|\frac{1}{r \mu_\kappa'}-2\right|
2 \left(\frac{w^2}{2\p}+\p\right) \frac{|v|^2-w^2}{\p}
+ \left|\frac{1}{(r \mu_\kappa')^2}-4\right| \left(\frac{|v|^2-w^2}{\p}\right)^2\\
&
\qquad \leq
\frac{1}{r \mu_\kappa'}\left|1 -2 r \mu_\kappa' \right|
\frac{|v|^2-w^2}{\p}\left(\frac{w^2}{\p} + 2 \p + 6 \frac{|v|^2-w^2}{\p}\right)\\
&
\qquad \leq
C |v|^2 \kappa^{-\delta}.
\end{align*}
Using Lemma~\ref{L:ABIDENTITY} we find that
\[
\int |\phi_\kappa'(E_\kappa)| |v|^2 dv =  3 e^{-\mu_\kappa} (\rho_\kappa + p_\kappa)
\leq 4 e^{-\mu_\kappa} \rho_\kappa.
\]
We insert the above two estimates into the definition of $A_{12}$, and
using \eqref{bounds} and Proposition~\ref{muass} we conclude that
\begin{align}
A_{12}
&
\leq C \kappa^{-\delta}
\int_{r_\kappa^1}^{r_\kappa^2} e^{2\mu_\kappa-\lambda_\kappa}
r^2 (\mu_\kappa')^2
\rho_\kappa \,dr
\leq
C C_\kappa e^{C \kappa^{-\delta}\ln \kappa}
\kappa^{-\delta} \int_{r_\kappa^1}^{r_\kappa^2} \frac{dr}{r}\notag \\
&
= C C_\kappa e^{C \kappa^{-\delta}\ln \kappa}  \kappa^{-\delta} \ln\kappa .\label{A12est}
\end{align}
To estimate $A_2$ we notice that 
\[
\text{supp}\chi_\kappa'=[r_\kappa^1,2 r_\kappa^1] \cup [r_\kappa^2/2,r_\kappa^2],
\]
we use the bounds \eqref{chipbounds} for $\chi_\kappa'$
and the estimate
\[
\int |\phi_\kappa'(E_\kappa)|\frac{w^4}{\p^2}\,dv \le \int |\phi_\kappa'(E_\kappa)|w^2\,dv
= \frac{e^{-2\lambda_\kappa-\mu_\kappa}}{4\pi  r} 
\left(\lambda_\kappa'+\mu_\kappa'\right);
\]
cf.~\eqref{11} and \eqref{laplusmu}.
Together with \eqref{bounds} this implies that
\begin{align}
A_2
&\le
\int_{\text{supp}\chi_\kappa'} r e^{2\mu_\kappa-\lambda_\kappa}(\chi_\kappa')^2
\left(\lambda_\kappa'+\mu_\kappa'\right)\,dr \notag \\
&\le
C C_\kappa e^{C \kappa^{-\delta}\ln\kappa}
\int_{\text{supp}\chi_\kappa'} (\chi_\kappa')^2 r\,dr \notag \\
& =
C C_\kappa e^{C \kappa^{-\delta}\ln\kappa}
\left(\int_{r_\kappa^1}^{2 r_\kappa^1} (\chi_\kappa')^2 r\,dr 
+\int_{r_\kappa^2/2}^{r_\kappa^2} (\chi_\kappa')^2 r\,dr\right) \notag \\
& \leq
C C_\kappa e^{C \kappa^{-\delta}\ln\kappa}
\left(\int_{r_\kappa^1}^{2 r_\kappa^1} \frac{1}{(r_\kappa^1)^2} r\,dr 
+\int_{r_\kappa^2/2}^{r_\kappa^2} \frac{1}{(r_\kappa^2)^2} r\,dr\right) \notag \\
& \le
C C_\kappa e^{C \kappa^{-\delta}\ln\kappa}. \label{E:A2BOUND}
\end{align}
In order to estimate $A_3$ we first estimate the $v$-integral contained
in that term:
\begin{align*}
\int |\phi_\kappa'(E_\kappa)| \left|\frac{w^4}{2\p^2}+w^2
-\frac{1}{\mu_\kappa'r}\frac{w^2|v|^2-w^4}{\p^2}\right|\,dv
& \leq
C \int |\phi_\kappa'(E_\kappa)| w^2 \,dv\\
&
= C \frac{e^{-2\lambda_\kappa -\mu_\kappa}}{r}\left(\lambda'_\kappa +\mu'_\kappa\right),
\end{align*}
where we used \eqref{11} and \eqref{laplusmu}.
Hence
\begin{align}
A_3
&\le
C \int_{\text{supp}\chi_\kappa'} r e^{2\mu_\kappa-\lambda_\kappa}|\chi_\kappa'|
\, \mu_\kappa'\left(\lambda_\kappa'+\mu_\kappa'\right)\,dr \notag \\
&\le
C C_\kappa e^{C \kappa^{-\delta}\ln\kappa}
\int_{\text{supp}\chi_\kappa'} |\chi_\kappa'|\,dr \notag \\
& =
C C_\kappa e^{C \kappa^{-\delta}\ln\kappa}
\left(\int_{r_\kappa^1}^{2 r_\kappa^1} |\chi_\kappa'|\,dr 
+\int_{r_\kappa^2/2}^{r_\kappa^2} |\chi_\kappa'|\,dr\right) \notag \\
& \leq
C C_\kappa e^{C \kappa^{-\delta}\ln\kappa}
\left(\int_{r_\kappa^1}^{2 r_\kappa^1} \frac{1}{r_\kappa^1}\,dr 
+\int_{r_\kappa^2/2}^{r_\kappa^2} \frac{1}{r_\kappa^2}\,dr\right) \notag \\
& \le
C C_\kappa e^{C \kappa^{-\delta}\ln\kappa}. \label{E:A3BOUND}
\end{align}
We add up \eqref{A111est}, \eqref{A112est}, \eqref{A12est},
\eqref{E:A2BOUND}, and \eqref{E:A3BOUND} to find
that besides the constants
$C_\kappa$ and $C$ introduced in Proposition~\ref{muass}
there exist positive constants $C_1, C_2, C_3$ such that 
\[
\mathcal{A}_\kappa(h,h) <
- C_\kappa e^{-C \kappa^{-\delta}\ln \kappa}
\left( C_1 \ln \kappa - C_2 e^{2 C \kappa^{-\delta}\ln \kappa} \kappa^{-\delta} \ln\kappa
-C_3 e^{2 C \kappa^{-\delta}\ln \kappa}\right) < 0,
\]
provided $\kappa$ is sufficiently large, and the proof is complete.
\prfe

\begin{remark}
The logarithmic gain in~\eqref{A111est} is fundamental to
the proof of Theorem~\ref{T:NEGATIVEA}. The perturbation
$f_{h}$ is carefully engineered to produce such a gain, taking
advantage of the exact $\kappa\to\infty$ asymptotics provided by
Proposition~\ref{P:EVTOBKZ}.
\end{remark}

\subsection{Linear exponential instability}\label{SS:LEI}

Theorem~\ref{T:NEGATIVEA} is sufficient for showing a linear exponential
instability result, without proving the existence of an exponentially
growing mode.
We present this simple argument first,
before we turn to the growing mode in the following
sections.

In order to exploit Theorem~\ref{T:NEGATIVEA} we further analyze
Eqn.~\eqref{heq}, which governs the dynamics of the linearly dynamically
accessible perturbations.
We split $h$ into even and odd parts with respect to $v$,
i.e., $h_+ (x,v) = \frac{1}{2}\left(h(x,v) + h(x,-v)\right)$ and
$h_- (x,v) = \frac{1}{2}\left(h(x,v) - h(x,-v)\right)$. It is then
easy to see that $\lambda_h =\lambda_{h_-}$,
$\mu_h =\mu_{h_-}$, and \eqref{heq} can be turned into the system
\begin{align}
\partial_t h_- & = \mathcal T h_+,\label{hsystem1}\\
\partial_t h_+ &= \mathcal T h_- - \mathcal C h_-,\label{hsystem2}
\end{align}
where we define
\be\label{E:MATHCALTDEF}
\mathcal T h := - e^{-\lambda_\kappa}\{h,E_\kappa\},\qquad
\mathcal C h :=
e^{\mu_\kappa}\lambda_{h} \frac{w^2}{\p} - e^{\mu_\kappa} \mu_{h} \p;
\ee
the operator $\mathcal T$ essentially represents the transport
along the characteristic flow of the steady state under consideration.
\begin{remark}
  It is tempting to turn the system \eqref{hsystem1}, \eqref{hsystem2}
  into a second order equation for the odd
  part $h_-$ alone, but under the regularity assumption on $\Phi$ the
  corresponding second order derivatives of $h_-$ need not exist.
  To do this rigorously we would have to interpret the equation in a
  distributional sense.
\end{remark}
In what follows we use the weight
$W : = e^{\lambda_\kappa}|\phi_\kappa'(E_\kappa)|$
and denote by $L^2_{W}=L^2_{W}(D_\kappa)$ the corresponding
weighted $L^2$ space on the set $D_\kappa$, cf.~\eqref{suppss};
$(\cdot,\cdot)_{L^2_{W}}$ denotes the corresponding scalar product.
We obtain the following linear, exponential instability result:
\begin{theorem} \label{lininstabEV}
  There exist initial data
  $\mathring h_+, \mathring h_-$ and constants $c_1,c_2>0$
  such that for the corresponding solution to the
  system \eqref{hsystem1}, \eqref{hsystem2},
  \[
  \|h_-(t)\|_{L^2_{W}} ,  \|\mathcal T h_+(t)\|_{L^2_{W}}
  \ge c_1 e^{c_2 t}.
  \]
\end{theorem}
\begin{proof}
The proof relies on the fact that solutions to
\eqref{hsystem1}, \eqref{hsystem2} preserve energy, i.e.,
\begin{align}\label{E:ENERGY}
  (\mathcal T h_+,\mathcal T h_+)_{L^2_W} + \mathcal A_\kappa (h_-,h_-)
  = const,
\end{align}
and on the virial identity
\begin{align}\label{E:VIRIAL}
\frac{1}{2}\frac{d^2}{dt^2}(h_-,h_-)_{L^2_{W}} =
- \mathcal A_\kappa(h_-,h_-) + (\mathcal T h_+,\mathcal T h_+)_{L^2_{W}}.
\end{align}
As to \eqref{E:ENERGY} we recall from \cite{HaRe2013} that the energy
$\mathcal A_\kappa (h,h)$ is conserved,
but if we substitute $h=h_+ + h_-$ and use
the fact that $\lambda_h = \lambda_{h_-}$ it follows easily
that
\[
\mathcal A_\kappa(h,h) = \mathcal A_\kappa(h_-,h_-)
+ (\mathcal T h_+,\mathcal T h_+)_{L^2_{W}}.
\]
Now
\[
\frac{1}{2}\frac{d}{dt}(h_-,h_-)_{L^2_{W}}
= (\partial_t h_-,h_-)_{L^2_{W}}
= (\mathcal T h_+,h_-)_{L^2_{W}},
\]
and
\beas
\frac{1}{2}\frac{d^2}{dt^2}(h_-(t),h_-(t))_{L^2_{W}}
&=&
\lim_{\delta \to 0}
\left(\mathcal T \frac{h_+(t+\delta)-h_+(t)}{\delta},
h_-(t)\right)_{L^2_{W}}
+
\left(\mathcal T h_+(t),\partial_t h_-(t)\right)_{L^2_{W}}\\
&=&
- \lim_{\delta \to 0}
\left(\frac{h_+(t+\delta)-h_+(t)}{\delta},\mathcal T h_-(t)\right)_{L^2_{W}}
+ \left(\mathcal T h_+(t),\partial_t h_-(t)\right)_{L^2_{W}}\\
&=&
-
\left(\partial_t h_+(t),\mathcal T h_-(t)\right)_{L^2_{W}}
+ \left(\mathcal T h_+(t),\partial_t h_-(t)\right)_{L^2_{W}}\\
&=&
- \left(\mathcal T h_-(t)- \mathcal C h_-(t),
\mathcal T h_-(t)\right)_{L^2_{W}}
+ \|\mathcal T h_+(t)\|_{L^2_{W}}^2;
\eeas
here we used the anti-symmetry of $\mathcal T$ stated
in Lemma~\ref{intbyparts} below.
The identity
\[
\lambda_{h_-} = 4\pi r  e^{\mu_\kappa+\lambda_\kappa}
\int \phi_\kappa'(E_\kappa)\,h_- (x,v)\,w\,dv
\]
implies that
\[
\mathcal T h_-
= \frac{1}{|\phi_\kappa'(E_\kappa)|}\left( f_{h_-} -
e^{\mu_\kappa}\phi_\kappa'(E_\kappa) \frac{w^2}{\p} \lambda_{h_-}\right).
\]
If we substitute this into the above expression and use the linearized
field equations \eqref{lineelambda} and \eqref{lineemu} to replace the
terms $\rho_h$ and $p_h$ it follows by a straightforward
computation that
\beas
&&
- \left(\mathcal T h_-(t)- \mathcal C h_-(t),
        \mathcal T h_-(t)\right)_{L^2_{W}}\\
&&
\qquad = -\mathcal A_\kappa(h_-,h_-) + \int_0^\infty e^{\mu_\kappa -\lambda_\kappa}
\left[\left(r \left(\lambda'_\kappa - \mu'_\kappa\right)-1\right)
  \lambda_h \mu_h
  - r\left(\mu'_h \lambda_h + \mu_h \lambda'_h\right)
  \right]\, dr.
\eeas
If we integrate the $\mu'_h \lambda_h$ term by parts
the $r$ integral vanishes and \eqref{E:VIRIAL} follows.

To prove the linear exponential instability we follow the approach used by
Laval, Mercier, Pellat~\cite{LaMePe} in the context of
stability of plasma flows.
Since Jeans' Theorem does not hold for the Einstein-Vlasov system, cf.\
\cite{Sch}, we must modify the argument of~\cite{LaMePe}.

By Theorem~\ref{T:NEGATIVEA} and Remark~\ref{suppfh}
there exists an odd-in-$v$ function
$\mathring h_-\in C^2 (\overline{D}_\kappa)$
such that
\begin{align}\label{E:HCHOICE}
\mathcal A_\kappa(\mathring h_-,\mathring h_-) < 0,
\end{align}
and to obtain initial data for the system \eqref{hsystem1}, \eqref{hsystem2}
we supplement this by
\begin{align*}
\mathring h_+ = -\varepsilon \mathcal T \mathring h_-,
\end{align*}
where $\varepsilon>0$ is chosen so small that
\be\label{E:CHOICE2}
\|\mathcal T \mathring h_+\|_{L^2_{W}}^2 =
\varepsilon^2 \|\mathcal T^2 \mathring h_-\|_{L^2_{W}}^2 <
- \mathcal A_\kappa(\mathring h_-,\mathring h_-).
\ee
Conservation of energy \eqref{E:ENERGY} and \eqref{E:CHOICE2} imply that
\begin{align}\label{E:ENERGY1}
  \mathcal A_\kappa(h_-,h_-) +  \|\mathcal T h_+\|_{L^2_{W}}^2 =
  \mathcal A_\kappa(\mathring h_-,\mathring h_-) +
  \|\mathcal T \mathring h_+\|^2_{L^2_{W}} = c
  <0.
\end{align}
By the Cauchy-Schwarz inequality,
\[
  \frac{1}{4} \left(\frac{d}{dt} \|h_-\|_{L^2_{W}}^2\right)^2
  = (\mathcal T h_+,h_-)^2_{L^2_{W}} \leq
  \|\mathcal T h_+\|_{L^2_{W}}^2 \|h_-\|_{L^2_{W}}^2.
\]
Combining this with the virial identity~\eqref{E:VIRIAL} and
\eqref{E:ENERGY1} it follows that
\begin{align}\label{E:KEYINEQUALITY}
\frac{d^2}{dt^2}\|h_-\|_{L^2_{W}}^2 = -2 c
+ 4  \|\mathcal T h_+\|_{L^2_{W}}^2
\ge
\frac{\left[\frac{d}{dt}\|h_-\|_{L^2_{W}}^2\right]^2}{\|h_-\|_{L^2_{W}}^2}
\end{align}
Setting
\[
y = \ln \left(\frac{ \|h_-\|_{L^2_{W}}^2}
{\|\mathring h_-\|_{L^2_{W}}^2}\right)
\]
we see that $y(0)=0$ and
\beas
\dot y(t)
&=&
\frac{\frac{d}{dt}\|h_-\|_{L^2_{W}}^2}{\|h_-\|_{L^2_{W}}^2}
=2 \frac{(\partial_t h_-,h_-)_{L^2_{W}}}{\|h_-\|_{L^2_{W}}^2},\\
\ddot y(t)
&=& \frac{1}{\|h_-\|_{L^2_{W}}^4}
\left[\frac{d^2}{dt^2} \|h_-\|_{L^2_{W}}^2 \|h_-\|_{L^2_{W}}^2
  - \left(\frac{d}{dt} \left(\|h_-\|_{L^2_{W}}^2\right)\right)^2\right]
\geq 0
\eeas
by \eqref{E:KEYINEQUALITY}. Hence
\[
\dot y(t) \geq \dot y(0) =
2 \frac{(\mathcal T \mathring h_+,\mathring h_-)_{L^2_{W}}}
{\|\mathring h_-\|_{L^2_W}^2}
= 2 \epsilon \frac{\|\mathcal T \mathring h_-\|^2_{L^2_{W}}}
{\|\mathring h_-\|^2_{L^2_{W}}} =: c_\epsilon >0,
\]
and thus $y(t) \geq c_\epsilon t$. By definition of $y$ this is the assertion
for $\|h_-(t)\|_{L^2_{W}}$. By Cauchy-Schwarz,
\[
c_\epsilon\leq \dot y(t)
=2 \frac{(\partial_t h_-,h_-)_{L^2_{W}}}{\|h_-\|_{L^2_{W}}^2}\leq
2 \frac{\|\mathcal T h_+(t)\|_{L^2_{W}}}{\|h_-(t)\|_{L^2_{W}}},
\]
which implies the estimate for $\|\mathcal T h_+(t)\|_{L^2_{W}}$,
and the proof is complete.
\end{proof}

\begin{remark}
  It should be noted that $\mathring h_+$ defined in this proof
  is indeed even so that $\mathring h_\pm$
  are the even and odd parts of some initial data $\mathring h$ for the
  original linearized system \eqref{heq}--\eqref{linpdef}.
\end{remark}

\begin{lemma} \label{intbyparts}
Let $\chi \in C (]0,E^0_\kappa[)$, $\chi>0$, and $g,h \in C^1(D_\kappa)$.
Then
\[
\left(\mathcal T g,h\right)_{L^2_{e^{\lambda_\kappa} \chi(E_\kappa)}}
= - \left(g,\mathcal T h\right)_{L^2_{e^{\lambda_\kappa} \chi(E_\kappa)}},
\]
provided both integrals exist.
\end{lemma}

\begin{proof}
Clearly,
\[
\left(\mathcal T g,h\right)_{L^2_{e^{\lambda_\kappa} \chi(E_\kappa)}}
= - \iint_{D_\kappa} \{g,E_\kappa\} h \chi(E_\kappa)
= - \lim_{\epsilon \to 0} \iint_{E_\kappa \leq E_\kappa^0 -\epsilon}
\{g,E_\kappa\} h \chi(E_\kappa).
\]
For fixed $\epsilon >0$ we now integrate by parts:
\beas
\iint_{E_\kappa \leq E_\kappa^0 -\epsilon}
\{g,E_\kappa\} h \chi(E_\kappa)
&=&
\iint_{E_\kappa = E_\kappa^0 -\epsilon}
\left(\nabla_v E_\kappa \cdot \nu_x - \nabla_x E_\kappa \cdot \nu_v\right)\,
g h \chi(E_\kappa)\, dS(x,v)\\
&&
- \iint_{E_\kappa \leq E_\kappa^0 -\epsilon}
\{h,E_\kappa\} g \chi(E_\kappa);
\eeas
all the other terms which appear actually drop out, because
they contain either
$\partial_{x_i}\partial_{v_i} E_\kappa - \partial_{v_i}\partial_{x_i} E_\kappa$
or
$\nabla_v E_\kappa \cdot\nabla_x E_\kappa-\nabla_x E_\kappa \cdot\nabla_v E_\kappa$,
and we may temporarily replace $\chi$ by a smooth approximation.
But the outer unit normal $\nu = (\nu_x,\nu_v)$ which appears in the boundary
integral is parallel to $(\nabla_x E_\kappa,\nabla_v E_\kappa)$ so that the
boundary integral vanishes as well, and the assertion follows.
\end{proof}

It should be noted that the above integration-by-parts formula does
not require any boundary conditions on either $g, h$ or on the weight $\chi$.

\subsection{The Hamiltonian structure of the linearized Einstein-Vlasov system}
\label{sec-hamilton}
We begin by defining a Hilbert space which is the suitable state space for the
linearized dynamics, and a certain transport operator.
\begin{definition}[The Hilbert space $X$ and the operators
                   $\mathcal T_\kappa$ and $\mathcal B_\kappa$.]
\label{defXT}
\begin{itemize}
\item[(a)]
  The Hilbert space $X$ is the space of all spherically symmetric
  functions in the weighted
  $L^2$ space on the set $D_\kappa$ with weight
  $e^{\lambda_\kappa}/|\phi _{\kappa }^{\prime }|$ and inner product
  \be\label{E:INNERPRODUCT}
  \left( f_{1},f_{2}\right)_X : =\iint_{D_\kappa}
  \frac{e^{\lambda_\kappa}}{|\phi_\kappa'(E_\kappa)|} f_{1} f_{2} dv\,dx.
  \ee
  Functions $f\in X$ are extended by $0$ to all of $\R^6$.
\item[(b)]
  For a function $f\in X$ the transport term
  $\mathcal T_\kappa f := - e^{-\lambda_\kappa}\{f,E_\kappa\}$
  {\em exists weakly in $X$}, iff there exists a function $\eta \in X$ such that
  for all test functions $\psi \in C^\infty_c (D_\kappa)$,
  \[
  \left(\eta, \psi\right)_X = - \left(f, \mathcal T_\kappa \psi\right)_X.
  \]
  Let
  \[
  D(\mathcal T_\kappa)
  := \{f\in X \mid \mathcal T_\kappa f \ \mbox{exists weakly in}\ X\},
  \]
  and $\mathcal T_\kappa f:=\eta$ for $f\in D(\mathcal T_\kappa)$.
  \item[(c)]
  The operator $\mathcal B_\kappa \colon D(\mathcal B_\kappa)\subset X \to X$
  is defined by $D(\mathcal B_\kappa) = D(\mathcal T_\kappa)$ and
  \begin{align}
    \mathcal B_\kappa f & : = \mathcal T_\kappa f
    + 4\pi r \phi_\kappa'(E_\kappa) e^{2 \mu_\kappa+\l_\kappa}
    \left(w \int  f \frac{\tilde w^2}{\tp}\,d\tilde v
    - \frac{w^2}{\p}\int  f \tilde w\,d\tilde v\right), \label{E:BDEF}
  \end{align}
  where we recall the definition of $w$~\eqref{E:WDEF}.
\end{itemize}
\end{definition}
\begin{remark}\label{Tasa}
\begin{itemize}
\item[(a)]
 When it exists, the function $\eta$ in Definition~\ref{defXT}~(b)
 is uniquely determined by $f$.
\item[(b)]
  The fact that the operator $\mathcal T_\kappa$ is anti-self-adjoint (or
  skew-adjoint) is shown in detail in \cite{RStr}.
\end{itemize}
\end{remark}


In what follows we denote perturbations of $f_\kappa,\mu_\kappa,\l_\kappa$ by
$f, \mu, \l$ respectively, i.e., $f_\kappa +f,\mu_\kappa +\mu,\l_\kappa+\l$
is the solution to the Einstein-Vlasov system \eqref{vlasov}--\eqref{p_T},
which is then linearized in $f, \mu, \l$. In particular,
$\l=\l_f$ depends on $f$ through the non-local relationship
\be\label{E:LAMBDAOFF}
\l = 4\pi \frac{e^{2\l_\kappa}}r \int_0^r s^2 \int f \p \, dv\,ds,
\ee
the linearization of \eqref{lambdaf}.

\begin{remark}\label{R:LAMBDAREGULARITY}
  A simple application of the Cauchy-Schwarz inequality implies that $\l$
  is well-defined for any $f\in X$, i.e. there exists a $C>0$ such that
  $ \left\vert\int_0^r s^2 \int f \p \, dv\,ds\right\vert \le C$, $r\in[0,\infty[$,
  as $f$ is supported on $D_\kappa$. In particular $|\l(r)|\le \frac Cr$ and thus
  $\int_0^\infty \l(s)^2\,ds <\infty$.
\end{remark}

As shown in \cite{HaRe2013} the formal variation of~\eqref{Casimir} yields the relation
\begin{align*}
  \delta\mathcal C(f_\kappa)  f =
  \int_{D_\kappa} \frac{e^{\l_\kappa}}{|\phi_\kappa'(E_\kappa)|} \chi'(f_\kappa)
  \left(|\phi_\kappa'(E_\kappa)| f -
  e^{\mu_\kappa}\frac{w^2}{\p}|\phi_\kappa'(E_\kappa)| \l\right)\,dv\,dx.
\end{align*}
Therefore, the formal tangent space of linearly dynamically accessible
perturbations can be interpreted as the set of perturbations $f$ satisfying
the orthogonality condition
\begin{align}\label{E:ORTHOGONALITY}
  \left( \chi'(f_\kappa),
  \, |\phi_\kappa'(E_\kappa)|f - e^{\mu_\kappa}\frac{w^2}{\p}|\phi_\kappa'(E_\kappa)|
  \l\right)_X = 0
\end{align}
for all $\chi \in C^1(\mathbb R)$, $\chi(0)=0$.
 Here $\l$ is defined through~\eqref{E:LAMBDAOFF}.
By Proposition 3.2 in~\cite{HaRe2013} for any $h\in C^1(D_\kappa)$
the orthogonality condition~\eqref{E:ORTHOGONALITY} is satisfied
if $f=\mathcal B_\kappa h$. A somewhat lengthy, but standard density
argument shows
that~\eqref{E:ORTHOGONALITY} is true for any
$f\in \overline{R(\mathcal B_\kappa)}$.
This motivates the following definition.

\begin{definition}[Linearly dynamically accessible perturbations]
  \label{D:DYNACCDEF}
  A function $f\in X$ is a {\em linearly dynamically accessible perturbation}
  if $f\in \overline{R(\mathcal B_\kappa)}$.
\end{definition}

\begin{remark}
  \begin{itemize}
  \item[(a)]
    It is an open question whether
    $R(\mathcal B_\kappa)=  \overline{R(\mathcal B_\kappa)}$.
    In the context of the Vlasov-Poisson system, the
    analogous statement is true and can be inferred from~\cite{GuRe2007}.
    The above question is intimately tied to the validity of
    Jeans' theorem; if the former is true then one can indeed show that
    $R(\mathcal B_\kappa)=  \overline{R(\mathcal B_\kappa)}$.
    However, Jeans' theorem is known to be false in general for the
    Einstein-Vlasov system~\cite{Sch}.
\item[(b)]
  We shall see in Lemma~\ref{L:EVLIN1} that the operator
  $\mathcal B_\kappa$ and the Hilbert space $X$ arise naturally in the
  study of the linearized Einstein-Vlasov system around the steady
  state $(f_\kappa,\mu_\kappa,\l_\kappa)$.
  \end{itemize}
\end{remark}

A key ingredient in the following analysis is a modified potential
induced by a state $f\in X$.

\begin{definition}[The modified potential $\bar \mu$ and the
operator $\bar{\mathcal L}_\kappa$.] \label{D:XDEF}
\begin{itemize}
\item[(a)]
For $f\in X$ the {\em induced modified potential $\bar \mu$} is defined as
\begin{align}\label{E:MUBARDEFNEW}
\bar\mu(r) = \bar\mu_f(r)
: =- e^{-\mu_\kappa-\l_\kappa}\int_r^\infty \frac1s\,
e^{\mu_\kappa(s)+\l_\kappa(s)}(2s\mu_\kappa'(s)+1)\,\l(s)\,ds
\end{align}
with $\lambda$ defined by \eqref{E:LAMBDAOFF}.
\item[(b)]
The operator $\bar{\mathcal L}_\kappa \colon X \to X$ is defined by
\begin{align}
\bar{\mathcal L}_\kappa f & : =
f - \phi_\kappa'(E_\kappa)E_\kappa \bar\mu. \label{E:LKAPPA1DEF}
\end{align}
\end{itemize}
\end{definition}
As it shall be often used
in the remainder of this paper, we observe that
\begin{equation}
2 r \mu_\kappa' +1 \geq 1,\ r\geq 0;
\label{lower bound mu}
\end{equation}
this follows since $\mu_\kappa' \geq 0$.
\begin{lemma} \label{mubarlemma}
For $f\in X$, $\bar\mu = \bar\mu_f \in C([0,\infty[)\cap C^1(]0,\infty[)$,
and $|\bar \mu (r)| \leq C \| f\|_X$, $r\geq 0$, with some $C>0$ independent
of $f$. Moreover,
\begin{align}\label{E:BASICVLASOV}
\frac{e^{-\mu_\kappa-\l_\kappa}r}{2r\mu_\kappa'+1}(e^{\mu_\kappa+\l_\kappa} \bar\mu)' = \l,
\ r\geq 0,
\end{align}
$\bar\mu \in \dot H^1_r$,
the subspace of spherically symmetric functions in the homogeneous Sobolev space
$\dot H^1(\R^3)$,  and
\begin{align}\label{E:BARMUELLIPTIC}
\frac{1}{4\pi r^2}
\frac{d}{dr}\left(\frac{e^{-\mu_\kappa-3\l_\kappa}r^2}{2r\mu_\kappa'+1}
\frac{d}{dr}\left(e^{\mu_\kappa+\l_\kappa} \bar\mu\right)\right) = \rho
= \int f\p\,dv\ \mbox{a.~e.}
\end{align}
in the weak sense.
\end{lemma}
\begin{proof}
The properties of the steady state imply that the quantity
$e^{\mu_\kappa+\l_\kappa}(2 r \mu_\kappa'+1)$ is bounded;
this is seen from the field equation~\eqref{ein2} and the boundedness of
$p_\kappa,\mu_\kappa,\l_\kappa$ which is explained in Section~\ref{setup}.
For $f\in X$ the
Cauchy-Schwarz inequality implies that
\begin{align}\label{E:REFQUERYONE}
\left|\int_0^r s^2 \int f\, \p dv\, ds \right|
\leq C \min\{1,r^{3/2}\}\| f\|_X
\end{align}
which implies the estimate for $\bar\mu$. Its continuity is obvious,
and since $\l$ is continuous, $\bar\mu$ is continuously differentiable
for $r>0$ with
\be \label{mubarder}
\bar \mu' = - (\mu_\kappa' + \l_\kappa') \, \bar \mu
+ \frac{2r\mu_\kappa' +1}{r} \lambda,
\ee
which follows after multiplying~\eqref{E:MUBARDEFNEW} by $e^{\mu_\kappa+\l_\kappa}$ and then differentiating with respect to $r$.
Eqn.~\eqref{E:BASICVLASOV} follows directly from \eqref{mubarder}.
By Remark~\ref{R:LAMBDAREGULARITY} and~\eqref{lower bound mu} we conclude
from~\eqref{E:BASICVLASOV}
that $\frac{d}{dr}\left(e^{\mu_\kappa+\lambda_\kappa}\bar\mu\right)\in L^2_r$.
Since $\mu_\kappa'+\l_\kappa'=0$ outside $D_\kappa$, it follows that
$\bar\mu\in \dot H^1_r$.
Multiplying both sides of~\eqref{E:BASICVLASOV}
by $e^{-2\l_\kappa}r$ and taking one more radial derivative we obtain
\eqref{E:BARMUELLIPTIC}; note that $r^2 \rho\in L^1([0,\infty[)$.
\end{proof}

\begin{lemma}\label{L:KCOMPACT}
  The operator $K\colon X\to X$, $Kf: = \phi_\kappa'(E_\kappa)E_\kappa \bar\mu_f$ with
  $\bar\mu_f$ given by~\eqref{E:MUBARDEFNEW} is bounded, self-adjoint, and compact.
\end{lemma}

\begin{proof}
  Boundedness follows from Lemma~\ref{mubarlemma} and self-adjointness is obvious
  from an integration-by-parts argument and~\eqref{E:BARMUELLIPTIC}.
  To prove compactness, assume that $(f_n)_{n\in\mathbb N}\subset X$ converges
  weakly to some $f\in X$. From~\eqref{E:LAMBDAOFF},
  \[
  \l_{f_n-f}(s) =
  \frac{e^{2\l_\kappa}}{s} \left(f_n-f,e^{-\l_\kappa}|\phi_\kappa'(E_\kappa)|\p {\bf 1}_{B_s(0)}\right)_X .
  \]
  If we denote
  $C_n(s):=e^{2\l_\kappa}\left(f_n-f,
  e^{-\l_\kappa}|\phi_\kappa'(E_\kappa)| \p {\bf 1}_{B_s(0)}\right)_X$,
  it follows from the weak convergence of $(f_n)_{n\in\mathbb N}$
  that $\lim_{n\to\infty}C_n(s)=0$ for any $s\ge0$. Moreover for any $s\ge R_\kappa$,
  $C_n(s)=C_n(R_\kappa)$, where $[0,R_\kappa]$ is the spatial support of the steady state.
  For any $s,s'\in[0,R_\kappa]$,
  \[
  \left\vert C_n(s)-C_n(s') \right\vert =
  \left\vert e^{2\l_\kappa}\left(f_n-f,
  e^{-\l_\kappa}|\phi_\kappa'(E_\kappa)| \p {\bf 1}_{B_s(0)\setminus B_{s'}(0)}\right)_X\right\vert
  \le C |s-s'|,
  \]
  where the last bound follows from H\"older's inequality.
  In particular, the sequence of functions $(C_n)_{n\in\mathbb N}$
  is equicontinuous on $[0,R_\kappa]$. By the Arzela-Ascoli theorem it follows
  that along some subsequence $(f_{n_k})_{k\in \mathbb N}$,
  $\lim_{k\to\infty} c_{n_k} = 0$,
  where $c_n:=\sup_{s\ge0}|C_n(s)|=\max_{s\in[0,R_\kappa]}|C_n(s)|$.
  From this we conclude that
  \begin{align*}
    &
    \|Kf_{n_k}-Kf\|_X^2 \\
    &
    \quad= \iint_{D_\kappa} |\phi_\kappa'(E_\kappa)|e^{-\l_\kappa}\p^2
    \left\vert \int_r^\infty \frac1s\,e^{\mu_\kappa(s)+\l_\kappa(s)}(2s\mu_\kappa'(s)+1)\,
    \l_{f_{n_k}-f}(s)\,ds \right\vert^2\,dv\,dx \\
    &
    \quad\le C  \iint_{D_\kappa} |\phi_\kappa'(E_\kappa)|\left(\int_r^\infty \frac1s\,
    \left\vert\l_{f_{n_k}-f}(s)\right\vert\,ds\right)^2 \,dv\,dx \\
    &
    \quad\le C c_{n_k}
    \iint_{D_\kappa} |\phi_\kappa'(E_\kappa)|\left(\int_r^\infty \frac1{s^2}\,ds\right)^2 \,dv\,dx \\
    &
    \quad\le C c_{n_k}\iint_{D_\kappa} |\phi_\kappa'(E_\kappa)| \frac1{r^2} \,dv\,dx \\
    &
    \quad
    \le C c_{n_k} \to 0 \ \ \text{ as } \ k\to\infty.
  \end{align*}
  This shows the compactness of the map $K$.
\end{proof}

\begin{lemma}
The operator $\mathcal B_\kappa$ defined in Definition~\ref{defXT} is
densely defined and anti self-adjoint on $X$.
The operator $\bar{\mathcal L}_\kappa$
is bounded and symmetric on $X$.
\end{lemma}

\begin{proof}
We write $\mathcal B_\kappa = \mathcal T_\kappa + \mathcal R_\kappa$
where for $f\in X$,
\[
\mathcal R_\kappa f := 4\pi r \phi_\kappa'(E_\kappa) e^{2 \mu_\kappa+\l_\kappa}
\left( w \int  f \frac{\tilde w^2}{\tp}\,d\tilde v
- \frac{w^2}{\p}\int  f \tilde w\,d\tilde v \right);
\]
notice that $f$ is extended by $0$ to all of $\R^6$.
Since the functions $r, \mu_\kappa, \lambda_\kappa, \phi'(E_\kappa), v$
are all bounded on the set $D_\kappa$ a straightforward computation shows
that $\mathcal R_\kappa$ is bounded on $X$, and its anti-symmetry is just as easy
to check. Hence the assertion on $\mathcal B_\kappa$ follows from
Remark~\ref{Tasa}.
The estimate for $\bar \mu$ in Lemma~\ref{mubarlemma}
implies that $\bar{\mathcal L}_\kappa$ is bounded on $X$, and
it is easy to check that  $\bar{\mathcal L}_\kappa$ is symmetric.
\end{proof}

We now establish a first order Hamiltonian formulation of the
linearized Einstein-Vlasov system.
\begin{lemma} \label{L:EVLIN1}
The linearized, spherically symmetric Einstein-Vlasov system takes the form
\begin{align}\label{E:EVLINFIRST}
\pa_t f = \mathcal B_\kappa \bar{\mathcal L}_\kappa f,
\end{align}
where the operators $\mathcal B_\kappa$ and $\bar{\mathcal L}_\kappa$
are defined in~\eqref{E:BDEF} and~\eqref{E:LKAPPA1DEF} respectively.
Furthermore $D(\mathcal B_\kappa \bar{\mathcal L}_\kappa)=D(\mathcal B_\kappa)$.
The operator $\bar{\mathcal L}_\kappa$ induces a
quadratic form on $X$ which takes the form
\be\label{E:QUADRATICFORM}
\mathcal A_\kappa(f,f) :=
(\bar{\mathcal L}_\kappa f,f)_X
= \iint \frac{e^{\lambda _{\kappa }}}{|\phi _{\kappa }^{\prime }(E_{\kappa })|}
f^2\,dv\,dx
- \int_0^\infty e^{\mu_\kappa-\lambda_\kappa}
\left( 2r\mu_\kappa'+1\right) \,\lambda^2\,dr,
\ee
where $\l$ depends on $f$ through \eqref{E:LAMBDAOFF}.
The flow of \eqref{E:EVLINFIRST} preserves $\mathcal A_\kappa(f,f)$.
\end{lemma}

\begin{proof}
First we notice that for $f\in X$ at least formally,
\be\label{inX}
\mathcal T_\kappa (\phi_\kappa'(E_\kappa) E_\kappa \bar\mu)
= - e^{2 \mu_\kappa-\l_\kappa} \phi_\kappa'(E_\kappa)\, w\, \bar\mu',
\ee
where recall that the coordinate $w$ is given by~\eqref{E:WDEF}.
Since $\l_f$ is given by~\eqref{E:LAMBDAOFF} and the bound~\eqref{E:REFQUERYONE} holds, we conclude that
$|\l_f(r)| \leq C r^{1/2}$. This bound together with~\eqref{E:BASICVLASOV} implies that
$|\bar\mu'(r)| \leq C r^{-1/2}$, and the right hand side of
\eqref{inX} lies in $X$. Hence $\bar{\mathcal L}_\kappa$
maps $D(\mathcal B_\kappa)=D(\mathcal T_\kappa)$ into itself,
and the operator product $\mathcal B_\kappa \bar{\mathcal L}_\kappa$
is defined on $D(\mathcal B_\kappa)$.

To linearize the equations we expand $f = f_\kappa+\epsilon\delta f, \mu = \mu_\kappa+\epsilon\delta \mu, \l = \l_\kappa+\epsilon\delta \l$
into~\eqref{vlasov}, use~\eqref{ein3},
and neglect terms of quadratic order and higher.
By a slight abuse of notation we rename $\delta f, \delta\l, \delta\mu$ into $f,\l,\mu$ respectively, and
obtain the linearized Vlasov equation
\begin{align}\label{E:EVLIN1}
\pa_t  f =- e^{-\l_\kappa}\{f, E_\kappa\}
- 4\pi r \phi_\kappa'(E_\kappa) e^{2\mu_\kappa+\l_\kappa}\frac{w^2}{\p}
\int f \tilde w\,d\tilde v
+ e^{2\mu_\kappa-\l_\kappa}\phi_\kappa'(E_\kappa) w \mu'.
\end{align}
From the field equation~\eqref{ein2} and the definition~\eqref{p} of $p$
it follows that
\[
\mu' = 4\pi r e^{2\l_\kappa} \int f\frac{w^2}{\p}\,dv  +\frac{2r\mu_\kappa'+1}{r}\l.
\]
Plugging this back into~\eqref{E:EVLIN1} we obtain the equation
\begin{align}\label{E:EVLIN2}
\pa_t  f = \mathcal B_\kappa  f +
\phi_\kappa'(E_\kappa)e^{2\mu_\kappa-\l_\kappa} w \frac{2r\mu_\kappa'+1}{r} \l,
\end{align}
where $\mathcal B_\kappa$ is given by~\eqref{E:BDEF}.
Using \eqref{inX}, \eqref{11}, and \eqref{mubarder} it follows that
\[
\mathcal{B}_{\kappa }\left(
\phi _{\kappa }^{\prime }E_{\kappa }\bar{\mu}\right)=
-\phi _{\kappa }^{\prime }(E_{\kappa })e^{2\mu _{\kappa }-\lambda _{\kappa }}w
\frac{2r\mu _{\kappa }^{\prime }+1}{r}\lambda.
\]
Plugging this into~\eqref{E:EVLIN2} we obtain
\[
\partial _{t}f=\mathcal{B}_{\kappa }\left( f-\phi _{\kappa }^{\prime}
E_\kappa \bar{\mu}\right) = \mathcal{B}_\kappa \bar{\mathcal L}_\kappa f,
\]
which proves~\eqref{E:EVLINFIRST}.

The conservation of $(\bar{\mathcal L}_\kappa f,f)_X$ is a direct
consequence of the antisymmetry of $\mathcal B_\kappa$
and~\eqref{E:EVLINFIRST}.
Finally, from the definition of $\bar{\mathcal L}_\kappa$ we have
\begin{align*}
(\bar{\mathcal{L}}_{\kappa} f,f)_X & =
(f-\phi _{\kappa }^{\prime }E_{\kappa }\bar{\mu},\ f)_X \\
& = \iint \frac{e^{\lambda_\kappa }}{|\phi_\kappa'(E_{\kappa })|} f^{2}\,dv\,dx
+ \int e^{\mu_\kappa+\l_\kappa} \bar\mu \int f\p,\,dv\,dx \\
& =  \iint \frac{e^{\lambda_\kappa}}{|\phi_\kappa'(E_{\kappa })|} f^{2}\,dv\,dx
+ 4\pi \int_0^\infty r^2\bar\mu\,  e^{\mu_\kappa+\l_\kappa}\rho\,dr .
\end{align*}
Replacing $4\pi r^2 \rho$ by $\left(e^{-2\l_\kappa}r\l\right)'$
and using~\eqref{E:BASICVLASOV} we obtain~\eqref{E:QUADRATICFORM}.
\end{proof}

The above first order Hamiltonian formulation
appears to be new and not exploited elsewhere in the literature.
We now derive a second order formulation of the linearized
Einstein-Vlasov system, which was
first given in 1968 by Ipser and Thorne~\cite{IT68}.


\begin{lemma} \label{L:SECONDORDEREV}
A formal linearization of the spherically symmetric Einstein-Vlasov
system takes the form
\begin{align}\label{E:LINEARISEDEV}
\pa_{tt} f_- + \mathcal L_\kappa f_- = 0,
\end{align}
where $f_-$ denotes the odd part of $f$ with respect to the variable $v$,
\begin{align} \label{E:MATHCALLDEF}
\mathcal L_\kappa f & : =
- \mathcal B_\kappa\mathcal B_\kappa f
+ 4\pi e^{3\mu_\kappa} \phi_\kappa'(E_\kappa) w \,(2r\mu_\kappa'+1)
\int \tilde w\, f\,d\tilde v
= - \mathcal B_\kappa \bar{\mathcal L}_\kappa \mathcal B_\kappa f,
\end{align}
and $\mathcal B_\kappa$ is defined in~\eqref{E:BDEF}.
The operator $\mathcal L_\kappa$ is self-adjoint on $X$
with
\[
D(\mathcal L_\kappa) : = \{f\in X \mid  \mathcal B_\kappa\mathcal B_\kappa f
\in X\ \text{ {\em in a weak sense}}\},
\]
the latter being defined in analogy to Definition~\ref{defXT}~(b).
Therefore $\|\pa_t f_-\|_X^2 + (\mathcal L_\kappa f_-,f_-)_X$
is formally conserved along the flow of~\eqref{E:LINEARISEDEV},
and
\[
\|\pa_t f_-\|_X^2 + (\mathcal L_\kappa f_-,f_-)_X =
\mathcal A_\kappa(\pa_t f,\pa_t f).
\]
\end{lemma}

\begin{proof}
We begin by examining the operator $\mathcal L_\kappa$.
First we note that
$\mathcal B_\kappa\mathcal B_\kappa f \in X$
is defined in a pointwise, classical sense, provided
$f\in C^\infty_c(D_\kappa)$; here the regularity assumption
$\Phi \in C^2(]0,\infty[)$, i.~e., $\phi_\kappa(E_\kappa)\in C^2(D_\kappa)$
enters. Thus $\mathcal B_\kappa\mathcal B_\kappa$ is densely defined
and self-adjoint by arguments analogous to those in Remark~\ref{Tasa}.
If we define
\[
\mathcal R_\kappa f : =
4\pi e^{3\mu_\kappa} \phi_\kappa'(E_\kappa) w (2r\mu_\kappa'+1)
\int \tilde w\, f\,d\tilde v,
\]
then it follows easily that $\mathcal R_\kappa$ is bounded on $X$ and
symmetric, which proves the assertion for $\mathcal L_\kappa$.

Since $\phi_\kappa'(E_\kappa) E_\kappa \bar\mu$ is even in $v$ it is immediate
from~\eqref{E:EVLINFIRST} that
\begin{eqnarray*}
\partial _{t}f_{+} &=&\mathcal{B}_{\kappa }f_{-},\ \  \\
\partial _{t}f_{-} &=&\mathcal{B}_{\kappa }f_{+}+\phi_\kappa' (E_{\kappa })\,
e^{2\mu_\kappa - \lambda_\kappa} w \frac{2r\mu_\kappa'+1}{r} \lambda .
\end{eqnarray*}
We take the time derivative of the second equation and use the relation
\be \label{dtlambdalin}
\partial _{t}\lambda =-4\pi re^{\mu _{\kappa }+\lambda _{\kappa }}\int w\,f\,dv
=-4\pi re^{\mu _{\kappa }+\lambda _{\kappa }}\int w\, f_{-} dv,
\ee
which is the linearization of the field equation~\eqref{ein3}. Thus
\[
\partial _{tt}f_{-}=\mathcal{B}_{\kappa }^{2}f_{-}
-4\pi \phi _{\kappa}^{\prime }(E_{\kappa })e^{3\mu_\kappa} w\,
(2r\mu_\kappa'+1)\int \tilde w\ f_{-}d\tilde v =-\mathcal{L}_{\kappa }f_{-} .
\]
Next,
\begin{eqnarray*}
\left(\mathcal{L}_{\kappa }f_{-},f_{-}\right)_X
&=&
\left(\mathcal{B}_{\kappa }f_- ,\mathcal{B}_{\kappa }f_-\right)_X
-4\pi \int e^{3\mu_{\kappa }+\lambda _{\kappa }}
\left( 2r\mu _{\kappa }^{\prime }+1\right)
\left( \int w\ f_{-}dv\,\right) ^{2}\,dx  \\
&=&
\|\partial_t f_+ \|_X^2
-\int_{0}^{\infty}e^{\mu _{\kappa }-\lambda _{\kappa }}\left( 2r\mu_\kappa'+1\right) \,
(\lambda_{\partial_t f})^2\,dr;
\end{eqnarray*}
here we recall \eqref{dtlambdalin} and observe that
$\partial_t \lambda = \lambda_{\partial_t f}$.
Therefore
\begin{align*}
\|\pa_t f_-\|_X^2 + (\mathcal L_\kappa f_-,f_-)_X
&=\|\partial_t f_-\|_X^2 + \|\partial_t f_+\|_X^2
-\int_{0}^{\infty}e^{\mu _{\kappa }-\lambda _{\kappa }}\left( 2r\mu_\kappa'+1\right) \,
(\lambda_{\partial_t f})^2\,dr\\
& = \|\partial_t f \|_X^2
-\int_{0}^{\infty}e^{\mu _{\kappa }-\lambda _{\kappa }}\left( 2r\mu_\kappa'+1\right) \,
(\lambda_{\partial_t f})^2\,dr\\
& = \mathcal A_\kappa(\pa_t f,\pa_t f),
\end{align*}
by \eqref{E:QUADRATICFORM} and the fact that $f_+$ and $f_-$ are orthogonal
to each other in the space $X$.
\end{proof}

\begin{remark}
\begin{itemize}
\item[(a)]
The fact that $\mathcal A_\kappa(\pa_t f,\pa_t f)$ is conserved along
the linearized flow can also be seen directly from the first order equation
satisfied by $\partial _{t}f$,
\[
\partial _{t}\left( \partial _{t}f\right) =\mathcal{B}_{\kappa }\mathcal{L}
_{\kappa }\partial _{t}f.
\]
Note that by~\eqref{E:EVLINFIRST},
$\partial _{t}f\in R\left(\mathcal{B}_{\kappa }\right)$,
the range of the operator $\mathcal{B}_{\kappa }$.
\item[(b)]
As noted at the beginning of Section~\ref{linEV}
the quadratic form~\eqref{E:QUADRATICFORM} arises as the
second variation of a suitably chosen energy-Casimir
functional~\cite{HaRe2013,HaRe2014,KM}, which is conserved
along the flow of the nonlinear Einstein-Vlasov system.
\end{itemize}
\end{remark}

\subsection{The reduced Einstein-Vlasov operator}
\label{SS:REDUCED}

The main result of this section is Theorem~\ref{T:VLASOVREDUCED},
which states in a precise manner that the operator $\mathcal L_\kappa$
is bounded from below and above by a simpler,
Schr\"odinger-type operator. Note that this operator acts only on functions that depend on
the spatial variable $r$ and we therefore refer to it as a macroscopic operator.
In Section~\ref{SS:PROOFOFTHEOREM1} we shall use this to obtain
refined spectral information about the operator $\mathcal L_\kappa$.

An attempt to derive such a reduced operator was made by
Ipser in~\cite{IP1969}, yet the operator derived there bounds the
operator $\mathcal L_\kappa$ from below only under an additional
assumption on the steady state, which appears to be hard to verify.
Our approach is new and the key ingredient is the use of the modified
potential $\bar\mu$ introduced in \eqref{E:MUBARDEFNEW}.

The operator $\Delta_\kappa$ is given by
\be\label{lap_kappa}
\Delta_\kappa \psi : =
\frac{e^{\mu_\kappa+\l_\kappa}}{4\pi r^2}
\frac{d}{dr}\left(\frac{e^{-\mu_\kappa-3\l_\kappa}r^2}{2r\mu_\kappa'+1}
\frac{d}{dr} \left(e^{\mu_\kappa+\l_\kappa}\psi\right)\right).
\ee
The {\em reduced operator} is given by
\be\label{E:SKAPPADEFINITION}
S_\kappa\psi : =
- \Delta_\kappa \psi
- e^{\l_\kappa}\int |\phi_\kappa'(E_\kappa)|E_\kappa^2\,dv\,\psi.
\ee
The {\em non-local reduced operator}
$\tilde S_\kappa$ is given by
\be\label{E:STILDEKAPPADEF}
\tilde S_\kappa \psi : =
- \Delta_\kappa \psi
- e^{\l_\kappa}\int \left(\mathrm{id}-\Pi\right)
\left(|\phi_\kappa'(E_\kappa)|\,E_\kappa \psi \right) E_\kappa\,dv,
\ee
where $\Pi$ denotes the projection onto the orthogonal
complement of $R(\mathcal B_\kappa)$ in $X$, and $\mathrm{id}$ is the identity.
On a flat background, i.e., for $\lambda_\kappa=\mu_\kappa=0$ the
operator $ 4\pi \Delta_\kappa$ is the Laplacian applied to
spherically symmetric functions.

To make sense classically of the above operators it is natural to
require $\psi\in H^2_r=H^2(\R^3)\cap L^2_r$, where $L^2_r$
denotes the space of all spherically symmetric $L^2$ functions.
However, this is too strong for our purposes, since the
modified potentials $\bar \mu_f$ from Definition~\ref{D:XDEF} must belong
to the domain of definition of $S_\kappa$ and $\tilde S_\kappa$,
but they in general belong to $\dot H^1_r$. Therefore we define
$S_\kappa$ and $\tilde S_\kappa$ using duality:

\begin{definition}[The operators $S_\kappa$ and $\tilde S_\kappa$ as symmetric quadratic forms]
  \label{redEVdef}
  Let $(\dot H^1_r)'$ denote the dual space of $\dot H^1_r$,
  the radial subspace of the homogeneous Sobolev space $\dot H^1(\R^3)$.
  The dual pairing is denoted by $\langle\cdot,\cdot\rangle$.
  Then  $S_\kappa\colon \dot H^1_r \to  (\dot H^1_r)'$ is defined by
  \begin{align*}
    \langle S_\kappa \phi,\psi\rangle
    = & \int_0^\infty \frac{e^{-\mu_\kappa-3\l_\kappa}}{2r\mu_\kappa'+1}
    \frac{d}{dr} \left(e^{\mu_\kappa+\l_\kappa}\phi\right)
    \frac{d}{dr} \left(e^{\mu_\kappa+\l_\kappa}\psi\right)r^2\,dr \\
    & - \iint_{D_\kappa} e^{\l_\kappa} |\phi_\kappa'(E_\kappa)|E_\kappa^2 \phi \psi\,dv\,dx.
  \end{align*}
  The operator $\tilde S_\kappa\colon \dot H^1_r \to  (\dot H^1_r)'$
  is defined analogously.
\end{definition}

We use a standard definition of $\dot H^1$ as given e.g. in~\cite{Ge1998},
and we recall that $D_\kappa = \supp f_\kappa$.
As defined above the operators $S_\kappa$ and $\tilde S_\kappa$ are clearly self-dual.
Moreover, taking the $r$ derivative of
$\rho_\kappa = \int \phi_\kappa(E_\kappa) \p \,dv$ it follows that
\be\label{E:IDLOC}
\int |\phi_\kappa'(E_\kappa)|\p^2\,dv
= -e^{-\mu_\kappa} \frac{\rho_\kappa'}{\mu_\kappa'}.
\ee
Thus the operator $S_\kappa$ can formally be written in the form
\be\label{redEVmacro}
S_\kappa \psi  =
-\Delta_\kappa \psi
+ e^{\mu_\kappa+\l_\kappa} \frac{\rho_\kappa'}{\mu_\kappa'} \psi.
\ee
The key result of this section is the following theorem.


\begin{theorem}\label{T:VLASOVREDUCED}
\begin{enumerate}
\item[(a)]
  For every $\mu\in \dot H^1_r$ and
  $f = f_\mu := \phi_\kappa'(E_\kappa)E_\kappa\mu$ it holds that
  \be  \label{E:UPPERBOUND1}
  \langle S_\kappa\mu,\,\mu\rangle \ge {\mathcal A}_\kappa( f_\mu, f_\mu).
  \ee
  For every $f\in X$ and $\bar\mu_f$ as defined in \eqref{E:MUBARDEFNEW},
  \be  \label{E:LOWERBOUND1}
  {\mathcal A}_\kappa( f, f) \ge \langle\S_\kappa\bar\mu_f,\,\bar\mu_f\rangle.
  \ee
\item[(b)]
  For every $\mu\in \dot H^1_r$ and
  $\tilde f = \tilde f_\mu :=
  (\mathrm{id}-\Pi)(\phi_\kappa'(E_\kappa)E_\kappa\mu)\in R(\mathcal B_\kappa)$
  it holds that
  \be  \label{E:UPPERBOUND2}
  \langle \tilde S_\kappa\mu,\,\mu\rangle \ge
  {\mathcal A}_\kappa( \tilde f_\mu, \tilde f_\mu).
  \ee
  For every $f\in X$,
  \be  \label{E:LOWERBOUND2}
  {\mathcal A}_\kappa( f, f) \ge
  \langle\tilde S_\kappa\bar\mu_f,\,\bar\mu_f\rangle.
  \ee
\end{enumerate}
\end{theorem}

\begin{proof}
We first observe that for $\mu\in \dot H^1_r$,
\begin{equation}\label{lap_kappa_mu}
  \langle-\Delta_\kappa \mu,\mu\rangle
  = \int_0^\infty
  \frac{e^{-\mu_\kappa-3 \l_\kappa}}{2r\mu_\kappa'+1}
  \left(\left(e^{\mu_\kappa+\l_\kappa} \mu\right)'\right)^2 r^2 \,dr,
\end{equation}
in particular, using \eqref{E:MUBARDEFNEW} and \eqref{E:BASICVLASOV},
\begin{eqnarray}\label{lap_kappa_mubar1}
  \langle-\Delta_\kappa \bar\mu,\bar\mu\rangle
  &=&
  \int_0^\infty
  e^{\mu_\kappa-\l_\kappa} \left(2r\mu_\kappa'+1\right)
  \left(\frac{e^{-\mu_\kappa-\l_\kappa}}{2r\mu_\kappa'+1} r
  \left(e^{\mu_\kappa+\l_\kappa} \bar\mu\right)'\right)^2  \,dr
  \nonumber \\
  &=&
  \int_0^\infty e^{\mu_\kappa-\l_\kappa} \left(2r\mu_\kappa'+1\right)\, \lambda^2 dr,
\end{eqnarray}
and on the other hand by \eqref{E:BARMUELLIPTIC},
\begin{equation}\label{lap_kappa_mubar2}
  \langle-\Delta_\kappa \bar\mu,\bar\mu\rangle
  =
  - 4\pi \int_0^\infty e^{\mu_\kappa+\l_\kappa} \bar\mu\, \rho\, r^2 dr
  = -\iint e^{\mu_\kappa+\l_\kappa} \bar\mu \p f \,dv\,dx.
\end{equation}
The definitions of ${\mathcal A}_\kappa$ and $S_\kappa$ together with
\eqref{lap_kappa_mubar1}--\eqref{lap_kappa_mubar2} imply that
\begin{eqnarray} \label{skappa<}
  \langle S_\kappa \bar\mu,\bar\mu\rangle
  &=&
  \int_0^\infty e^{\mu_\kappa-\l_\kappa} \left(2r\mu_\kappa'+1\right)\, \lambda^2 dr
  -\iint e^{\l_\kappa+2\mu_\kappa}
  |\phi_\kappa'(E_\kappa)|\p^2 \bar\mu^2 \,dv\,dx\nonumber \\
  &=&
  {\mathcal A}_\kappa(f,f)
  -\iint f^2 \frac{e^{\lambda_\kappa}}{|\phi'(E_\kappa)|}\,dv\,dx
  - 2 \iint e^{\mu_\kappa+\l_\kappa} \bar\mu \p f \,dv\,dx\nonumber\\
  &&
  {}-\iint e^{\l_\kappa+2\mu_\kappa}
  |\phi_\kappa'(E_\kappa)|\p^2 \bar\mu^2 \,dv\,dx\nonumber\\
  &=&
  {\mathcal A}_\kappa(f,f)
  - \iint e^{\l_\kappa} \left(\frac{f}{|\phi_\kappa'(E_\kappa)|^{1/2}}
  + e^{\mu_\kappa} \bar\mu |\phi_\kappa'(E_\kappa)|^{1/2} \p\right)^2 \,dv\,dx
  \nonumber\\
  &\leq&
  {\mathcal A}_\kappa(f,f), \nonumber
\end{eqnarray}
and  the lower bound \eqref{E:LOWERBOUND1} is proven.
To obtain the upper bound \eqref{E:UPPERBOUND1}
we observe that by \eqref{lap_kappa_mu} for $\mu\in H^2_r$ and
$\lambda=\lambda_f$
with $f\in X$,
\begin{eqnarray} \label{upper1}
  \langle-\Delta_\kappa \mu,\mu\rangle
  &=&
  \int_0^\infty
  \frac{e^{-\mu_\kappa-3 \l_\kappa}}{2r\mu_\kappa'+1}
  \left(\left(e^{\mu_\kappa+\l_\kappa} \mu\right)'\right)^2 r^2 \,dr
  \nonumber \\
  &=&
  \int_0^\infty
  e^{\mu_\kappa-\l_\kappa} \left(2r\mu_\kappa'+1\right)
  \left(\frac{r e^{-\mu_\kappa-\l_\kappa} \left(e^{\mu_\kappa+\l_\kappa} \mu\right)'}
       {2r\mu_\kappa'+1} -\lambda\right)^2\,dr\nonumber \\
  &&
  {} + 2  \int_0^\infty
  r e^{-2\l_\kappa} \left(e^{\mu_\kappa+\l_\kappa} \mu\right)' \lambda \,dr
  - \int_0^\infty e^{\mu_\kappa-\l_\kappa}\left(2r\mu_\kappa'+1\right)\, \lambda^2 \,dr.
  \quad \nonumber
\end{eqnarray}
Now
\begin{eqnarray} \label{mulambdarho}
  \int_0^\infty r e^{-2\l_\kappa} \left(e^{\mu_\kappa+\l_\kappa} \mu\right)'
  \lambda \,dr
  &=&
  - \int_0^\infty e^{\mu_\kappa+\l_\kappa} \mu
  \left(r e^{-2\l_\kappa}\lambda\right)' dr
  =
  -\int e^{\mu_\kappa+\l_\kappa} \mu \rho\,dx\qquad \nonumber\\
  &=&
  -\iint e^{\mu_\kappa+\l_\kappa} \mu \p f\, dv\,dx
  = (f,\phi_\kappa'(E_\kappa) E_\kappa \mu)_X \nonumber.
\end{eqnarray}
Together with \eqref{upper1} this implies that
\begin{eqnarray}\label{skappa>}
  \langle S_\kappa\mu, \mu\rangle
  &\geq&
  2 (f,\phi_\kappa'(E_\kappa) E_\kappa \mu)_X + \mathcal A_\kappa(f,f) - \|f\|_X^2
  - \iint |\phi_\kappa'(E_\kappa)| E_\kappa^2 \mu^2 e^{\lambda_\kappa} dv\,dx
  \nonumber\\
  &=&
  \mathcal A_\kappa(f,f) + 2 (f,\phi_\kappa'(E_\kappa) E_\kappa \mu)_X - \|f\|_X^2
  - \| \phi_\kappa'(E_\kappa) E_\kappa\mu\|_X^2 \nonumber\\
  &=&
  \mathcal A_\kappa(f,f), \nonumber
\end{eqnarray}
provided $f= \phi_\kappa'(E_\kappa) E_\kappa\mu$, and \eqref{E:UPPERBOUND1}
is proven.

We now turn to part (b)
and assume that $f\in R(\mathcal B_\kappa)$.
By definition of the projection $\Pi$ and \eqref{lap_kappa_mubar1}--\eqref{lap_kappa_mubar2},
\begin{align}\label{mix1}
  &\iint f\,
  (\mathrm{id} - \Pi)\left(\phi_\kappa'(E_\kappa) E_\kappa\bar\mu\right)
  \frac{e^{\l_\kappa}}{ |\phi_\kappa'(E_\kappa)|}\,dv\,dx
  = - \iint f\, \phi_\kappa'(E_\kappa)\bar\mu E_\kappa
  \frac{e^{\l_\kappa}}{ \phi_\kappa'(E_\kappa)}\,dv\,dx \notag \\
  & \qquad = -  \iint e^{\mu_\kappa+\l_\kappa}   f\, \bar\mu \p\,dv\,dx
  = \int_0^\infty e^{\mu_\kappa-\l_\kappa} \left(2r\mu_\kappa'+1\right)\, \lambda^2 dr,
\end{align}
and
\begin{align}\label{mix2}
  &
  \iint e^{\l_\kappa}
  (\mathrm{id} - \Pi)
  \left(\phi_\kappa'(E_\kappa) E_\kappa \bar\mu\right) E_\kappa \bar\mu\,dv\,dx
  \notag \\
  &
  \qquad = - \iint
  (\mathrm{id} - \Pi)\left(\phi_\kappa'(E_\kappa) E_\kappa \bar\mu\right)
  \phi_\kappa'(E_\kappa) E_\kappa \bar\mu
  \frac{e^{\l_\kappa}}{|\phi_\kappa'(E_\kappa)|} \,dv\,dx \notag \\
  &
  \qquad = - \iint
  \left((\mathrm{id} - \Pi)\left(\phi_\kappa'(E_\kappa)
  E_\kappa \bar\mu\right)\right)^2
  \frac{e^{\l_\kappa}}{|\phi_\kappa'(E_\kappa)|} \,dv\,dx .
\end{align}
Using \eqref{mix1} and \eqref{mix2} it follows that
\begin{eqnarray*}
  \langle\tilde S_\kappa \bar\mu,\bar\mu\rangle
  &=&
  {\mathcal A}_\kappa(f,f)
  - \iint f^2 \frac{e^{\lambda_\kappa}}{|\phi'(E_\kappa)|}\,dv\,dx
  + 2 \int_0^\infty e^{\mu_\kappa-\l_\kappa}
  \left(2r\mu_\kappa'+1\right)\, \lambda^2 dr \\
  &&
  {}+\iint e^{\l_\kappa}
  (\mathrm{id} - \Pi)\left(\phi_\kappa'(E_\kappa) E_\kappa \bar\mu\right)
  E_\kappa \bar\mu\,dv\,dx \\
  &=&
  {\mathcal A}_\kappa(f,f)
  - \iint
  \left(f- (\mathrm{id} - \Pi)
  \left(\phi_\kappa'(E_\kappa) E_\kappa \bar\mu\right)\right)^2
  \frac{e^{\l_\kappa}}{|\phi_\kappa'(E_\kappa)|}\,dv\,dx \\
  &\leq&
  {\mathcal A}_\kappa(f,f),
\end{eqnarray*}
which proves \eqref{E:LOWERBOUND2}.
To obtain the upper bound \eqref{E:UPPERBOUND2}
we let $f =(\mathrm{id}-\Pi)\left(\phi_\kappa'(E_\kappa)E_\kappa \mu\right)$.
Then in particular, $(f,\phi_\kappa'(E_\kappa)E_\kappa \mu)_X = \|f\|_X^2$.
An argument analogous to the one used for
\eqref{skappa>} implies that
\begin{eqnarray*}
  \langle\tilde S_\kappa\mu, \mu\rangle
  &\geq&
  \mathcal A_\kappa(f,f) + 2 (f,\phi_\kappa'(E_\kappa) E_\kappa \mu)_X - \|f\|_X^2\\
  &&
  {}
  + \iint (\mathrm{id}-\Pi)(\phi_\kappa'(E_\kappa) E_\kappa \mu)\, E_\kappa \mu
  e^{\lambda_\kappa} dv\,dx \\
  &=&
  \mathcal A_\kappa(f,f) + \|f\|_X^2
  - \iint f \phi_\kappa'(E_\kappa)\, E_\kappa \mu
  \frac{e^{\lambda_\kappa}}{|\phi_\kappa'(E_\kappa)|} dv\,dx \\
  &=&
  \mathcal A_\kappa(f,f),
\end{eqnarray*}
and the proof is complete.
\end{proof}

\subsection{Main results for the Einstein-Vlasov system}
\label{SS:PROOFOFTHEOREM1}

We first need to comment on some notation.
 Let $L$
be either a self-adjoint, linear operator $L\colon H\supset D(L) \to H$ on some
Hilbert space $H$ with scalar product $(\cdot,\cdot)_H$
or  a self-dual operator $L\colon H\to H'$. In either case we
we denote by $n^{-}\left(L\right)$
its negative Morse index, which by definition
is the maximal dimension of subspaces of
$H$ on which
$(L\cdot,\cdot)_H <0$ or $\langle L\cdot,\cdot\rangle < 0$ in the self-dual case
with $\langle \cdot,\cdot\rangle$ the dual pairing.  For any decomposition
$H=H_{-}\oplus \ker L\oplus H_{+}$ such that
$(L\cdot,\cdot)_H <0$ on
$H_{-}\setminus\left\{  0\right\}$ and
$(L u,u)_H \geq\delta\left\Vert u\right\Vert_{H}^{2}$
on $H_{+}$ it
can be shown that $\dim H_{-}=n^{-}\left( L\right)$, see Remark 2.2 and
Lemma 12.1 in \cite{LinZeng2017}. For any subspace $S\subset H$, we also use
$n^{-}\left(L|_{S}\right)$ to denote the Morse index of
$L$ restricted to $S$.

The self-adjointness of the operator $\mathcal L_\kappa$ and
Theorem~\ref{T:VLASOVREDUCED}
imply that the negative part of the spectrum has to be discrete
and that part (a) of Theorem~\ref{mainev} holds,
more precisely:


\begin{theorem}\label{thm:main VE2}
  The negative part of the spectrum of
  $\mathcal L_\kappa \colon X \supset D(\mathcal L_\kappa)\to X$
  is either empty or consists of at most finitely many eigenvalues
  with finite multiplicities. For $\kappa$ sufficiently large, there
  exists at least one negative eigenvalue and therefore a growing mode
  for the linearization of the Einstein-Vlasov system around
  $(f_\kappa, \l_\kappa,\mu_\kappa)$.
\end{theorem}


\begin{proof}
  By Lemma~\ref{L:SECONDORDEREV} the operator $\mathcal L_\kappa$ is self-adjoint. 
  We note that for $f\in D(\mathcal L_\kappa)$,
  \begin{equation*}
    (\mathcal{L}_{\kappa }f,f)_X =(\bar{\mathcal{L}}_{\kappa }
    \mathcal{B}_{\kappa }f ,\mathcal{B}_{\kappa}f)_X .
  \end{equation*}
  By Theorem~\ref{T:VLASOVREDUCED}, 
  \[
  n^{-}\left( \mathcal{L}_{\kappa }\right) \leq
  n^{-}\left( \bar{\mathcal{L}}_{\kappa }\right)
  \le n^{-}\left( S_{\kappa }\right) <\infty;
  \]
  for the sake of completeness we recall the general argument behind
  the first two estimates in Lemma~\ref{op_est_morse} in the appendix.
  In order to show that $n^{-}\left( S_{\kappa }\right) <\infty$ we first observe
  that $\langle \S_\kappa \phi,\phi\rangle \geq C \langle \S'_\kappa \phi,\phi\rangle$
  for $\phi \in \dot H^1_r$, where the self-dual operator
  $S'_\kappa \colon  \dot H^1_r\to (\dot H^1_r)'$ is formally given as
  $S'_\kappa = -\Delta -V_\kappa$ with a non-negative, continuous potential $V_\kappa$
  which has compact support. This follows from the corresponding bounds on
  $\lambda_\kappa$ and $\mu_\kappa$. Now we observe that the mapping
  $(-\Delta)^{1/2} \colon \dot H^1(\R^3) \to L^2(\R^3)$,
  $\phi\mapsto (2\pi |\xi|\hat\phi)\check{\phantom{\phi}}$ is an isomorphism which respects
  spherical symmetry. Passing to
  $\psi = (-\Delta)^{1/2}\phi$ we arrive at the relation
  that
  $4 \pi \langle \S'_\kappa \phi,\phi\rangle = ((\mathrm{id}-K)\psi,\psi)_{L^2}$,
  where $K=(-\Delta)^{-1/2}V_\kappa(-\Delta)^{-1/2}\colon L^2(\R^3) \to  L^2(\R^3)$
  is compact, since $V_\kappa$ is bounded and supported on the compact set
  $\bar B_{R_\kappa}(0)$, and the mapping
  $\dot H^1(\R^3) \ni f \mapsto {\bf 1}_{\bar B_{R_\kappa}(0)} f\in L^2(\R^3)$ is compact;
  notice that $h = V_\kappa(-\Delta)^{-1/2}\psi\in L^1\cap L^2(\R^3)$ so that
  $\hat h\in L^\infty \cap L^2(\R^3)$, and hence $\frac{1}{2 \pi|\xi|} \hat h$ and its inverse
  Fourier transform are in $L^2(\R^3)$.
  The spectral properties of compact operators imply that $n^-(\mathrm{id}-K) <\infty$,
  and invoking Lemma~\ref{op_est_morse} again it follows that
  $n^{-}\left( S_{\kappa }\right) <\infty$ as claimed.
  
  The assertion on the spectrum of $\mathcal L_\kappa$ follows from its spectral
  representation; for the sake of completeness we include the argument
  in Proposition~\ref{saspec} in the appendix.
  
  Each negative eigenvalue of $\mathcal{L}_{\kappa }$ gives rise to a
  pair of stable and unstable eigenvalues for~\eqref{E:LINEARISEDEV}.
  For $\kappa$ sufficiently large, $n^-(\mathcal L_\kappa)\ge 1$ by
  Theorem~\ref{T:NEGATIVEA}, and the claim follows.
  Note that the operator $\mathcal L_\kappa$ is non-negative when
  restricted to the subspace 
  of $X$ of all even-in-$v$ functions in $X$, since for any
  $f\in X$, $\int_{\mathbb R^3} f_+ w \,dv=0$, 
  where $f_+(x,v)=\frac12(f(x,v)+f(x,-v))$ is the even part of $f$.
  In particular, since the 
  negative eigenvalues are of finite multiplicity,
  the associated eigenspace necessarily contains only odd-in-$v$ 
  functions. 
\end{proof}

\begin{remark}
The finiteness of $n^-(\bar{\mathcal L}_\kappa)$ follows alternatively from the observation that 
the difference $\text{id} - \bar{\mathcal L}_\kappa = K$, which is a compact operator by Lemma~\ref{L:KCOMPACT}.
Since $\bar{\mathcal L}_\kappa$ is a compact perturbation of the identity, the only possible accumulation point
is $1$ and therefore $n^-(\bar{\mathcal L}_\kappa)<\infty$. 
\end{remark}

\begin{remark}
  By the same arguments as in the proof above we
  can conclude that
  \begin{equation}\label{E:DIMENSIONFORMULA}
  n^{-}\left( \mathcal{L}_{\kappa }\right) \leq
  n^{-}\left( \bar{\mathcal{L}}_{\kappa }|_{\overline{R(\mathcal{B}_{\kappa })}}\right)
  \le n^{-}\left( \tilde{S}_{\kappa }\right) <\infty,
  \end{equation}
  which will be used below.
\end{remark}


We now come to the existence of the linearized flow and
the corresponding exponential trichotomy decomposition into a
stable, unstable, and center space. Under a non-degeneracy assumption on
the operator $\tilde S_\kappa$ we in fact obtain a Lyapunov stability of the flow
on the center space for any $\kappa>0$. The key in our analysis will be the
first order formulation~\eqref{E:EVLINFIRST}.


\begin{theorem} \label{thm:main VE}
  The operator $\mathcal{B}_{\kappa }\bar{\mathcal{L}}_{\kappa }$
  generates a $C^{0}$ group $e^{t\mathcal{B}_{\kappa }\bar{\mathcal{L}}_{\kappa }}$
  of bounded linear operators on $X$, and there exists a decomposition
  \begin{equation*}
    X=E^{u}\oplus E^{c}\oplus E^{s}
  \end{equation*}
  with the following properties:
  \begin{itemize}
  \item[(i)]
    $E^{u}\left( E^{s}\right) $ is the linear subspace spanned by the
    eigenvectors corresponding to
    positive (negative) eigenvalues of
    $\mathcal{B}_{\kappa }\bar{\mathcal{L}}_{\kappa }$, and
    \begin{equation}
      \dim E^{u}=\dim E^{s}=n^{-}\left( \tilde{S}_{\kappa }\right) .
      \label{unstable-dimension-formula}
    \end{equation}
  \item[(ii)]
    The quadratic form
    $\left(\bar{\mathcal{L}}_{\kappa }\cdot ,\cdot \right)_X$
    vanishes on $E^{u,s}$, but is non-degenerate on $E^{u}\oplus E^{s}$, and
    \begin{equation*}
      E^{c}=\{f\in X \mid {\mathcal{A}}_{\kappa }(f,g)=
      \left( \bar{\mathcal{L}}_{\kappa }f,\ g\right)_X =0\ \mbox{for all}\
      g\in E^{s}\oplus E^{u}
      \}.
    \end{equation*}
  \item[(iii)]
    $E^{c}$, $E^{u}$, $E^{s}$ are invariant under
    $e^{t\mathcal{B}_{\kappa }\bar{\mathcal{L}}_{\kappa }}$.
  \item[(iv)]
    Let $\lambda _{u}=
    \min \{\lambda \mid \lambda \in \sigma (\mathcal{B}_{\kappa }
    \bar{\mathcal{L}}_{\kappa }|_{E^{u}})\}>0$.
    Then there exist $M>0$
    such that
    \begin{equation}
      \begin{split}
        & \left\vert e^{t\mathcal{B}_{\kappa }\bar{\mathcal{L}}_\kappa}|_{E^{s}}\right\vert
        \leq M e^{-\lambda_{u}t},\ t\geq 0,\\
        & \left\vert e^{t\mathcal{B}_{\kappa }\bar{\mathcal{L}}_\kappa}|_{E^{u}}\right\vert
        \leq M e^{\lambda_{u}t},\ t\leq 0,
      \end{split}
      \label{estimate-stable-unstable}
    \end{equation}
    \begin{equation}
      \left\vert e^{t\mathcal{B}_{\kappa }\bar{\mathcal{L}}_{\kappa }}|_{E^{c}}\right\vert
      \leq M (1+\left\vert t\right\vert )^{k_{0}},\ t\in \mathbb{R},
      \label{estimate-center}
    \end{equation}
    where
    \begin{equation}
      k_{0}\leq
      1+2\left( n^{-}\left( S_{\kappa }\right) -
      n^{-}\left( \tilde{S}_{\kappa }\right) \right) .  \label{bound-k-0}
    \end{equation}
  \item[(v)]
    When $\ker \tilde{S}_{\kappa }=\left\{ 0\right\} $, then $E^{c}$ can be
    further decomposed as
    \begin{equation}
      E^{c}=\ker \left( \mathcal{B}_{\kappa }\bar{\mathcal{L}}_{\kappa }\right)
      \oplus \left( E^{c}\cap \overline{R(\mathcal{B}_{\kappa })}\right),
      \label{decom-E-c}
    \end{equation}
    and
    $\bar{\mathcal{L}}_{\kappa }|_{E^{c}\cap \overline{R(\mathcal{B}_{\kappa})}}>\delta $
    for some positive constant $\delta $. As a consequence,
    $|e^{t\mathcal{B}_{\kappa }\bar{\mathcal{L}}_{\kappa }}|_{E^{c}}|\leq M$ for some
    constant $M$. In particular, when $\tilde{S}_{\kappa }>0$, then the
    steady state is stable in the sense that for all perturbations
    $f\in X$,
    \begin{equation}
      \left\Vert e^{t\mathcal{B}_{\kappa }\bar{\mathcal{L}}_{\kappa }}f\right\Vert_{X}
      \leq M\left\Vert f\right\Vert_{X}.
      \label{stability-general data}
    \end{equation}
    \end{itemize}
\end{theorem}


\begin{remark}
  \begin{itemize}
    \item[(a)]
      Eqn.~\eqref{unstable-dimension-formula} gives the sharp linear
      stability condition, namely $\tilde{S}_{\kappa }\geq 0$.
    \item[(b)]
      We expect the exponential trichotomy estimates in (iii) to become
      useful for the possible
      construction of invariant (stable, unstable and center) manifolds,
      as well as for a possible proof of nonlinear instability in the future.
    \item[(c)]
      In the Newtonian limit the operator $\tilde{S}_{\kappa}$
      formally converges to its Newtonian counterpart, which was
      proved to be positive e.g.\ in \cite{GuoLin}.  This can be shown rigorously
      and for $\kappa\ll1$ it can be then shown that
      $\tilde{S}_{\kappa}\ge 0$.
      By~\eqref{stability-general data} we therefore obtain linear stability
      against {\em general initial data} in $X$.
      This improves the result in \cite{HaRe2013}, where linear
      stability was proved for dynamically accessible initial data,
      i.e., when $f \in R\left( \mathcal{B}_{\kappa }\right)$.
  \end{itemize}
\end{remark}


To prove Theorem~\ref{thm:main VE} we first
need to prove several preparatory statements.

\begin{proposition}
  \label{prop-L-decompo}
  There exists a decomposition of $X$ into the direct
  sum of three closed subspaces
  \begin{equation*}
    X=X_{-}\oplus \ker \bar{\mathcal{L}}_{\kappa }\oplus X_{+}
  \end{equation*}
  with the following properties:
  \begin{itemize}
  \item[(H1)]
    $\left( \bar{\mathcal{L}}_{\kappa }f,f\right)_X <0$ for all
    $f\in X_{-}\setminus \{0\}$, and
  \begin{equation}
    n^{-}\left( \bar{\mathcal{L}}_{\kappa }\right) := \dim
    X_{-}=n^{-}\left( S_{\kappa }\right)<\infty ,\ \ \
    \dim \ker \bar{\mathcal{L}}_{\kappa }=\dim \ker S_{\kappa }.
    \label{dim-equality-L-kappa}
  \end{equation}
\item[(H2)]
  There exists $\delta >0$ such that for all $f\in X_{+}$,
  \begin{equation*}
    \left( \bar{\mathcal{L}}_{\kappa }f,f\right)_X \geq
    \delta \left\Vert f\right\Vert_{X}^{2}.
\end{equation*}
\end{itemize}
\end{proposition}

\begin{proof}
  By Lemma~\ref{L:KCOMPACT} and the definition~\eqref{E:LKAPPA1DEF} of $\bar{\mathcal L}_\kappa$
  the operator $\bar{\mathcal L}_\kappa$ is a compact perturbation of the identity, and
  therefore the only possible accumulation point of its spectrum is $1$.
  Therefore the negative Morse index and the dimension of
  $\ker \bar{\mathcal L}_\kappa$ are finite and the claimed decomposition follows,
  as well as the bound (H2).

  To prove (\ref{dim-equality-L-kappa}), first we note that
  \begin{equation}
    \dim \ker \bar{\mathcal{L}}_{\kappa }=\dim \ker S_{\kappa },
    \label{equ-dim-kernel}
  \end{equation}
  since it is easy to see that $f\in \ker \bar{\mathcal{L}}_{\kappa }$
  implies $\mu \in \ker S_{\kappa }$ where $\mu$ is defined from $f$
  via~\eqref{E:BARMUELLIPTIC}, and $\mu \in \ker S_{\kappa }$ implies that
  $\phi_{\kappa }'(E_{\kappa })E_{\kappa }\mu \in \ker \bar{\mathcal{L}}_{\kappa }$.
  Let $n^{\leq 0}\left( \bar{\mathcal{L}}_{\kappa }\right) $ and
  $n^{\leq 0}\left( S_{\kappa }\right)$ be the non-positive dimensions of
  $\left( \bar{\mathcal{L}}_{\kappa }\cdot ,\cdot \right)_X$ and
  $\left(S_{\kappa }\cdot ,\cdot \right)_{L^2_r}$. By
  Theorem~\ref{T:VLASOVREDUCED},
  $n^{\leq 0}\left(\bar{\mathcal{L}}_{\kappa }\right)=n^{\leq 0}\left(S_{\kappa }\right)$
  which combined with
  (\ref{equ-dim-kernel}) implies that
  $n^{-}\left( \bar{\mathcal{L}}_{\kappa}\right) =
  n^{-}\left( S_{\kappa }\right)$.
  This finishes the proof of the proposition.
\end{proof}

In the study of linear stability of (\ref{E:EVLINFIRST}), it is important to
restrict to the dynamically accessible
space $R\left( \mathcal{B}_{\kappa }\right) $, for which we have the
following result.

\begin{lemma}
  \label{lemma-quadratic-L-R}
  The quadratic form
  $\left( \bar{\mathcal{L}}_{\kappa }\cdot ,\cdot \right)_X $
  is non-degenerate on the quotient space
  $\overline{R(\mathcal{B}_{\kappa })} /
  \left( \overline{R(\mathcal{B}_{\kappa })}\cap
  \ker \bar{\mathcal{L}}_{\kappa }\right)$
  if and only if $\ker \tilde{S}_{\kappa}\subset \ker S_{\kappa }$.
  Moreover,
  \begin{equation} \label{equ-dim-ker-S-tilde}
  \dim \ker \left( \left( \bar{\mathcal{L}}_{\kappa }\cdot ,\cdot \right)_X
  |_{\overline{R(\mathcal{B}_{\kappa })}}\right) =\dim \ker \tilde{S}_{\kappa },
\end{equation}
and
\begin{equation}
n^{-}\left( \bar{\mathcal{L}}_{\kappa }|_{\overline{R(\mathcal{B}_{\kappa })}
/ \left( \overline{R(\mathcal{B}_{\kappa })}\cap \ker \bar{\mathcal{L}}
_{\kappa }\right) }\right)
=n^{-}\left(  \bar{\mathcal{L}}_{\kappa }
|_{\overline{R\left( \mathcal{B}_{\kappa }\right) }}\right) =
n^{-}\left( \tilde{S}_{\kappa }\right).
\label{equality-negative-modes-L-R}
\end{equation}
\end{lemma}

\begin{proof}
Since
\begin{equation*}
  \ker \left( \left( \bar{\mathcal{L}}_{\kappa }\cdot ,\cdot \right)_X
  |_{\overline{R(\mathcal{B}_{\kappa })}}\right)
  =\overline{R(\mathcal{B}_{\kappa})}\cap
  \ker \mathcal{B}_{\kappa }\bar{\mathcal{L}}_{\kappa },
\end{equation*}
it implies that
$f\in
\ker \left( \left( \bar{\mathcal{L}}_{\kappa }\cdot,\cdot \right)_X
|_{\overline{R(\mathcal{B}_{\kappa })}}\right) $
if and only if $f\in \overline{R(\mathcal{B}_{\kappa })}$ and
\begin{equation}
  0=\left( \mathrm{id}-\Pi \right) \bar{\mathcal{L}}_{\kappa }f
  =f-\left( \mathrm{id}-\Pi \right)
\left( \phi _{\kappa }^{\prime }(E_{\kappa })E_{\kappa }\mu \right)
\label{eqn-ker--f-R-B};
\end{equation}
here $\Pi$ is the projection onto $\ker \mathcal{B}_{\kappa }$ in
$X$. In view of the the definition
of $\mu $ by \eqref{E:BARMUELLIPTIC}, we deduce from
(\ref{eqn-ker--f-R-B}) that $\tilde{S}_{\kappa }\mu =0$.
On the other hand, if
$\tilde{S}_{\kappa }\mu =0$, then it is easy to check that
\begin{equation*}
  \left( \mathrm{id}-\Pi \right)
  \left( \phi _{\kappa }^{\prime }(E_{\kappa })E_{\kappa}\mu \right)
  \in \overline{R(\mathcal{B}_{\kappa })}\cap \ker \mathcal{B}_{\kappa }
  \bar{\mathcal{L}}_{\kappa }.
\end{equation*}
Thus (\ref{equ-dim-ker-S-tilde}) is proved. By Theorem \ref{T:VLASOVREDUCED},
\begin{equation*}
  n^{\leq 0}\left( \bar{\mathcal{L}}_{\kappa }|_{\overline{R(\mathcal{B}_{\kappa })}}\right)
  =n^{\leq 0}\left( \tilde{S}_{\kappa }\right) ,
\end{equation*}
which together with (\ref{equ-dim-ker-S-tilde}) implies
(\ref{equality-negative-modes-L-R}).
\end{proof}

Now we are ready to prove Theorem~\ref{thm:main VE}.

\begin{proof}[Proof of Theorem \protect\ref{thm:main VE}]
To prove (i) we
note that each pair of stable and unstable eigenvalues for the operator
$\mathcal L_\kappa$
provided by Theorem~\ref{thm:main VE2} also
gives a pair of stable and unstable eigenvalues for (\ref{E:EVLINFIRST}).
Together with~\eqref{equality-negative-modes-L-R}
we obtain the dimension formula~\eqref{unstable-dimension-formula}.

Consider the Hamiltonian formulation (\ref{E:EVLINFIRST}). The operator
$\mathcal{B}_{\kappa }$ can be defined as a self-dual operator
$X'\supset D\left( \mathcal{B}_{\kappa }\right) \rightarrow X$.
Moreover, by Proposition \ref{prop-L-decompo}, the quadratic form
$\left( \bar{\mathcal{L}}_{\kappa }\cdot ,\cdot \right)_X$
satisfies
the assumptions (H1)--(H3) in \cite{LinZeng2017}. Thus the conclusions
(ii)--(iv)
follow directly from \cite[Thm.\ 2.2]{LinZeng2017} for general
Hamiltonian PDEs.

It remains to show (v). When $\ker \tilde{S}_{\kappa }=\left\{ 0\right\} $,
by Lemma \ref{lemma-quadratic-L-R},
$\left( \bar{\mathcal{L}}_{\kappa }\cdot,\cdot \right)_X |_{\overline{R(\mathcal{B}_{\kappa })}}$
is non-degenerate and
\begin{equation*}
  \overline{R(\mathcal{B}_{\kappa })}\cap
  \ker \mathcal{B}_{\kappa }\bar{\mathcal{L}}_{\kappa }=\left\{ 0\right\}.
\end{equation*}
By \cite[Prop.~2.8]{LinZeng2017} we have the direct sum
decomposition
\begin{equation}
  X=\overline{R(\mathcal{B}_{\kappa })}\oplus
  \ker \mathcal{B}_{\kappa }\bar{\mathcal{L}}_{\kappa }.  \label{decom-X}
\end{equation}
Since $\ker \mathcal{B}_{\kappa }\bar{\mathcal{L}}_{\kappa }$ is the steady
solution space of (\ref{E:EVLINFIRST}) and is therefore contained in $E^{c}$, we
get the decomposition (\ref{decom-E-c}) for $E^{c}$. Since
\begin{equation*}
n^{-}\left( \bar{\mathcal{L}}_{\kappa }|_{E^{s}\oplus E^{u}}\right) =\dim
E^{u}=n^{-}\left( \tilde{S}_{\kappa }\right)
\end{equation*}
by \cite[Cor. 6.1]{LinZeng2017} and $E^{c}$ is the orthogonal complement
of $E^{s}\oplus E^{u}$ in $X$, it follows that
\begin{equation*}
  n^{-}\left( \mathcal{\bar{L}}_{\kappa }
  |_{E^{c}\cap \overline{R(\mathcal{B}_{\kappa })}}\right)
  =n^{-}\left( \mathcal{\bar{L}}_{\kappa }
  |_{\overline{R(\mathcal{B}_{\kappa })}}\right)
  -n^{-}\left( \mathcal{\bar{L}}_{\kappa}|_{E^{u}\oplus E^{s}}\right)
  =n^{-}\left( \bar{S}_{\kappa }\right)
-n^{-}\left( \bar{S}_{\kappa }\right) =0.
\end{equation*}
Thus
$\mathcal{\bar{L}}_{\kappa }|_{E^{c}\cap \overline{R(\mathcal{B}_{\kappa })}}
\geq \delta _{1}>0$, and by using
$\left( \mathcal{\bar{L}}_{\kappa }f,f\right)_X$
as the Lyapunov functional, we get the stability
estimate
$\left\Vert e^{t\mathcal{B}_{\kappa }\bar{\mathcal{L}}_{\kappa}}f\right\Vert _{X}
\leq M\left\Vert f\right\Vert _{X}$
for
$f \in E^{c}\cap \overline{R(\mathcal{B}_{\kappa })}$. Then the stability
estimate
$|e^{t\mathcal{B}_{\kappa }\bar{\mathcal{L}}_{\kappa}}|_{E^{c}}|\leq M$
follows by the decomposition (\ref{decom-E-c}) and the
fact that $e^{t\mathcal{B}_{\kappa }\bar{\mathcal{L}}_{\kappa }}f=f$ when
$f\in \ker \mathcal{B}_{\kappa }\bar{\mathcal{L}}_{\kappa }$.
When $\tilde{S}_{\kappa }>0$, (\ref{stability-general data})
follows from the decomposition
(\ref{decom-X}) by the same arguments as above. This finishes the proof of
the theorem.
\end{proof}

\begin{remark}
  \begin{itemize}
  \item[(a)]
    The polynomial growth on $E^{c}$ can be shown (\cite{LinZeng2017b}) to be at most quadratic
    by using the second order formulation (\ref{E:LINEARISEDEV}).
  \item[(b)]
    It is tempting to find the most unstable eigenvalue $\lambda _{0}<0\ $
    of $\mathcal L_\kappa$
    by trying to minimize
    $\left( \mathcal{L}_{\kappa }f,f\right)_X$
    over the constraint set
    $\{f\in D\left(\mathcal{L}_{\kappa }\right) \mid
    \left\Vert f\right\Vert _{X}=1\}$.
    However, in the current case it is difficult to study this variational
    problem directly due to the lack of compactness.
    It is the self-adjointness of $\mathcal{L}_{\kappa }$ and the
    finiteness of the negative Morse index that are crucial to the proof of
    the existence of unstable eigenvalues.
  \end{itemize}
\end{remark}


\section{Stability analysis for the Einstein-Euler system}
\label{S:EE}
\setcounter{equation}{0}


\subsection{The number density $n$}\label{SS:NUMBERDENSITY}

We need to introduce an additional quantity $n=N(\rho)$,
the number density. We define the function $N$ as
\be \label{Ndef}
N(\rho):= \exp\left(\int_1^\rho \frac{ds}{s+P(s)}\right),\
\rho \in ]0,\infty[.
\ee
Clearly,
$N\in C^2 (]0,\infty[)$ with
\be \label{Ndiffeq}
\frac{dN}{d\rho} = \frac{N}{\rho + P(\rho)},
\ee
and
\be \label{N2diff}
\frac{d^2 N}{d\rho^2} = \frac{dN/d\rho}{\rho + P(\rho)}
- \frac{N}{(\rho + P(\rho))^2} \left(1+ \frac{dP}{d\rho}\right)
= - \frac{dP}{d\rho} \frac{N}{(\rho + P(\rho))^2} < 0,
\ee
in particular, $\frac{dN}{d\rho}$ is a positive and strictly decreasing
function on $]0,\infty[$.
By $({\mathrm P}1)$, $P(s) \leq c s^\gamma$
for $s\in [0,1]$, where $\gamma >1$. Hence for $0<\rho\leq 1$,
\[
\int_1^\rho \frac{ds}{s+P(s)}\leq -\int_\rho^1 \frac{ds}{s(1+c s^{\gamma-1})}
= -\frac{1}{\gamma-1}\int_{\rho^{\gamma-1}}^1 \frac{d\sigma}{\sigma(1+c \sigma)}
= \ln\left(\frac{(1+c)^{1/(\gamma-1)}\rho}{(1+c \rho^{\gamma-1})^{1/(\gamma-1)}}\right).
\]
Together with
\eqref{Ndef}, \eqref{Ndiffeq}, and the monotonicity of $\frac{dN}{d\rho}$
this implies that
$\lim_{\rho\searrow 0} N(\rho) = 0$ and
$\lim_{\rho\searrow 0} \frac{dN}{d\rho}(\rho)$ exists in $]0,\infty[$,
in particular, $N\in C^1([0,\infty[)$.

The asymptotic behavior \eqref{radiation} implies that
$\lim_{\rho \to \infty} N(\rho) =\infty$ so that
$N\colon [0,\infty[ \to [0,\infty[$ is one-to-one and onto. Given any
solution of the Einstein-Euler system we define the corresponding
number density $n= N(\rho)$ so that $\rho=N^{-1}(n) = \rho(n)$,
and by \eqref{Ndiffeq},
\begin{equation}\label{E:NDEF}
n \frac{d \rho}{d n} = \rho(n) + p(n).
\end{equation}
Since the functional relationship between the pressure and the energy
density $\rho$ is prescribed by~\eqref{eqstate}, the expression $p(n)$
simply means $P(\rho(n))$.

Using \eqref{N2diff},
the Tolman-Oppenheimer-Volkov equation \eqref{tov}
and the fact that
$p'_\kappa = \frac{dP}{d\rho}(\rho_\kappa) \rho'_\kappa$ it follows that
\[
\left(\frac{dN}{d\rho}(\rho_\kappa)\right)'
= - \frac{dP}{d\rho}(\rho_\kappa)
\frac{N(\rho_\kappa)}{(\rho_\kappa + p_\kappa)^2}\rho'_\kappa
= - \frac{N(\rho_\kappa)}{(\rho_\kappa + p_\kappa)^2} p'_\kappa
= \frac{N(\rho_\kappa)}{\rho_\kappa + p_\kappa} \mu'_\kappa
= \frac{dN}{d\rho}(\rho_\kappa) \mu'_\kappa
\]
on the interval $[0,R_\kappa[$, where $[0,R_\kappa]$ is the support of the steady state.
We integrate this differential equation
to find that for all $0\leq r, r_0 < R_\kappa$,
\[
\frac{dN}{d\rho}(\rho_\kappa(r))
= \frac{dN}{d\rho}(\rho_\kappa(r_0)) e^{\mu_\kappa(r) - \mu_\kappa(r_0)} .
\]
We take the limit
$r_0\nearrow R_\kappa$ and
observe that $\rho_\kappa(R_\kappa)=0$ so that
\[
\frac{dN}{d\rho}(\rho_\kappa(r))
= \frac{dN}{d\rho}(0) \, e^{\mu_\kappa(r) - \mu_\kappa(R_\kappa)}.
\]
Now we observe that the function
\be\label{E:NKAPPADEF0}
N_\kappa := c_\kappa N
\ee
with $c_\kappa>0$ has the same properties
as the function $N$ stated above, and we can choose the normalization
constant $c_\kappa$ such that
$\frac{dN_\kappa}{d\rho}(0)= e^{\mu_\kappa(R_\kappa)}$.
For the given steady state and with $n_\kappa = N_\kappa(\rho_\kappa)$
the identity
\be \label{Nmu}
\frac{n_\kappa}{\rho_\kappa + p_\kappa} =
\frac{dN_\kappa}{d\rho}(\rho_\kappa)
= e^{\mu_\kappa}
\ee
holds on the interval $[0,R_\kappa]$.

\begin{remark}
  The normalization can not be achieved by taking $r\to\infty$,
  since outside the support of the steady state the middle
  term is necessarily constant, but $\mu_\kappa$ is not.
\end{remark}

Before we linearize the Einstein-Euler system we collect a
few properties of the steady state under consideration.


\begin{lemma} \label{ssprops}
  Let $(\rho_{\kappa},\lambda_{\kappa},\mu_{\kappa})$
  be a steady state given by Proposition~\ref{ssfamilies}
  for some $\kappa >0$.
  Then
  \begin{equation}\label{E:TOV3}
    \frac{n_\kappa}{\rho_\kappa+p_\kappa}
    = \frac1{\frac{d \rho}{d n}(n_\kappa)}
    = e^{\mu_\kappa}.
  \end{equation}
  Let
  \begin{equation} \label{E:PSIDEF}
  \Psi_\kappa(r) :=\frac1{n_\kappa(r)}\frac{d P}{d \rho}(\rho_\kappa(r)),
  \ 0 \leq r < R_\kappa,
  \end{equation}
  denote the specific enthalpy of the steady state; $[0,R_\kappa]$
  is the support of $\rho_\kappa$. Then
  \begin{equation}  \label{E:TOV2}
    \Psi_\kappa n_\kappa' = -\mu_\kappa'.
  \end{equation}
  and $1/\Psi_\kappa \in C([0,R_\kappa])$ with
  \begin{equation}\label{E:PSIKAPPAPROP}
    \lim_{r\nearrow R_{\kappa }}\frac{1}{\Psi _{\kappa }(r) }=0.
  \end{equation}
\end{lemma}


\begin{proof}
  The identity~\eqref{E:TOV3} is a simple consequence of~\eqref{Nmu}.
  The proof of~\eqref{E:TOV2} relies on the
  Tolman-Oppenheimer-Volkov equation \eqref{tov}.
  Replacing
  $\rho_\kappa+p_\kappa$ by $n_\kappa\frac{d \rho}{d n}(n_\kappa)$
  and $p_\kappa'$ by
  $\frac{d P}{d \rho}(\rho_\kappa)\frac{d \rho}{d n}(n_\kappa) n_\kappa'$
  we obtain the result.
  The function $\Psi_\kappa$ is continuous on $]0,R_\kappa[$, and
  $\Psi_\kappa(0) >0$. Finally, \eqref{E:PSIKAPPAPROP} follows if we
  express $1/\Psi_\kappa$ via \eqref{E:TOV2} and observe that
  $\mu_\kappa'(R_\kappa)>0$ and
  $n_\kappa'(R_\kappa) = \frac{d N_\kappa}{d \rho}(0)\, \rho_\kappa'(R_\kappa)=0$.
\end{proof}


\subsection{Linearization}

Let us now linearize the Einstein-Euler
system \eqref{eelambda}--\eqref{ud} around a given steady state
$(n_\kappa, u_\kappa\equiv 0, \l_\kappa,\mu_\kappa)$.
We write $n,  u, \rho, p, \l,  \mu$ for
the Eulerian perturbations, i.e., $n_\kappa + n, u_\kappa + u, \rho_\kappa + \rho$
etc.\ correspond to the solution of the original, non-linear system. Linearizing
$\rho_\kappa +\rho = N_\kappa^{-1}(n_\kappa + n)$ and
$p_\kappa + p = P(N_\kappa^{-1}(n_\kappa + n))$ we have the relations
\begin{align}
  \rho & = \frac{d N_\kappa^{-1}}{d n}(n_\kappa)\,  n
  = \frac{1}{\frac{d N_\kappa}{d \rho}(\rho_\kappa)}\,  n
  = e^{-\mu_\kappa} n, \label{deltarhon} \\
p & =\frac{d P}{d \rho}(\rho_\kappa) \frac{d N_\kappa^{-1}}{d n}(n_\kappa)\,  n
= \Psi_\kappa (\rho_\kappa + p_\kappa)\, n,
\label{E:DELTAPFORMULA}
\end{align}
where we have used \eqref{Nmu} and \eqref{E:PSIDEF}.


Linearizing~\eqref{rhod} we arrive at
\be\label{E:FOLIN}
\dot{ \rho} + n_\kappa'  u
+ n_\kappa\left(\dot{ \l}+\left(\l_\kappa'+\mu_\kappa'+\frac2r\right) u +u'\right)
=0.
\ee
From~\eqref{eelambdad} it follows that
$\dot{ \l}=- 4 \pi  r e^{\mu_\kappa + 2\lambda_\kappa}\,
u\, (\rho_\kappa + p_\kappa)$.
Together with the relation~\eqref{laplusmu}
this yields
\be\label{E:DOTLAMBDAEQN}
\dot{ \l}+e^{\mu_\kappa}(\l_\kappa'+\mu_\kappa') u = 0,
\ee
and \eqref{E:FOLIN} simplifies to
\be
\dot{ \rho} + \frac{1}{r^2} \frac{d}{dr}\left(r^2n_\kappa  u\right)
= 0. \label{E:NLIN2}
\ee
Linearizing~\eqref{ud} we obtain the equation
\begin{align*}
  0&=e^{2\l_\kappa} \dot{ u}
  +e^{\mu_\kappa} \mu\left(\mu_\kappa'+\Psi_\kappa n_\kappa'\right)
+ e^{\mu_\kappa}\mu'+e^{\mu_\kappa}\left(e^{\mu_\kappa}\Psi_\kappa \rho\right)'\\
& = e^{2\l_\kappa} \dot{ u}
+ e^{\mu_\kappa}\mu'+e^{\mu_\kappa}\left(e^{\mu_\kappa}\Psi_\kappa \rho\right)',
\end{align*}
where we have used \eqref{E:TOV2}. Multiplying by $e^{\l_\kappa}$
we find that
\begin{align}
  0 &=
  e^{3\l_\kappa} \dot u  + e^{\l_\kappa+\mu_\kappa}\mu'
  + e^{\l_\kappa+\mu_\kappa}\left(e^{\mu_\kappa}\Psi_\kappa \rho\right)'  \notag \\
  &=
  e^{3\l_\kappa} \dot u + e^{\mu_\kappa+\l_\kappa}\mu'
  + \left(e^{2\mu_\kappa+\l_\kappa}\Psi_\kappa \rho\right)'
  - e^{2\mu_\kappa+\l_\kappa}(\mu_\kappa'+\l_\kappa')\Psi_\kappa \rho \notag \\
  &=
  e^{3\l_\kappa} \dot u
  + e^{\mu_\kappa+\l_\kappa} \left(\l \left(2\mu_\kappa'+\frac1r\right)
  + \Psi_\kappa(\l_\kappa'+\mu_\kappa') e^{\mu_\kappa}\rho\right)
  + \left(e^{2\mu_\kappa+\l_\kappa}\Psi_\kappa \rho\right)' \notag \\
  &\quad {}
  - e^{2\mu_\kappa+\l_\kappa}(\mu_\kappa'+\l_\kappa')\Psi_\kappa \rho \notag \\
&=
  e^{3\l_\kappa} \dot u
  + \frac1r e^{\mu_\kappa+\l_\kappa} (2 r \mu_\kappa'+ 1) \, \l
  + \left(e^{2\mu_\kappa+\l_\kappa}\Psi_\kappa \rho\right)', \label{linveleq}
\end{align}
where we have used the  field equation~\eqref{eemu}
in the third line to conclude that
\be\label{E:RMUPRIME}
r \mu' =  (2r\mu_\kappa'+1)\, \l
+ r\Psi_\kappa(\l_\kappa'+\mu_\kappa')e^{\mu_\kappa} \rho.
\ee

In order to proceed we define a suitable phase space and a modified
potential in complete analogy to Definition~\ref{D:XDEF}.

\begin{definition}\label{D:XDEF_EE}
\begin{itemize}
\item[(a)]
  The Hilbert space $X_1$ is the space of all spherically symmetric
  functions in the weighted $L^2$ space on the set
  $B_\kappa=B_{R_\kappa}$ (the ball with radius $R_\kappa$ which is
  the support of $\rho _{\kappa }$) with weight
  $e^{2\mu_\kappa +\lambda_\kappa}\Psi_{\kappa }$ and the corresponding inner product,
  $X_{2}$ is the space of
  radial functions in $L^2 \left( B_\kappa\right)$, and the phase space for the
  linearized Einstein-Euler system is
  \[
  X:=X_{1}\times X_{2}.
  \]
\item[(b)]
  For $\rho\in X_1$ the {\em induced modified potential $\bar \mu$}
  is defined as
  \begin{align}\label{E:MUBARDEFNEW_EE}
    \bar\mu(r) = \bar\mu_\rho(r)
    : =- e^{-\mu_\kappa-\l_\kappa}\int_r^\infty \frac1s\,
    e^{\mu_\kappa(s)+\l_\kappa(s)}(2s\mu_\kappa'(s)+1)\,\l(s)\,ds,
  \end{align}
  where
  \be\label{E:LAMBDAOFF_EE}
  \l(r) = 4\pi \frac{e^{2\l_\kappa}}r \int_0^r s^2 \rho(s)\,ds.
  \ee
\item[(c)]
  The operators
  $\tilde{\mathcal L}_\kappa \colon X_1
  \to X_1'$ and $\mathcal L_\kappa \colon X_1 \to X_1$
  are defined by
  \be\label{E:TILDELKAPPADEF}
  \tilde{\mathcal L}_\kappa \rho : =
  e^{2\mu_\kappa+\lambda_\kappa}\Psi_{\kappa }\rho +e^{\mu_\kappa +\lambda_\kappa}\bar{\mu}_\rho.
  \ee
  \be\label{E:MATHCALLKAPPADEF}
  \mathcal L_\kappa \rho : = e^{-2\mu_\kappa-\lambda_\kappa}\Psi^{-1}_{\kappa }
  \tilde{\mathcal L}_\kappa \rho
  = \rho + e^{-\mu_\kappa}\Psi_\kappa^{-1}\bar{\mu}_\rho
  \ee
\end{itemize}
\end{definition}
Here the dual pairing is realized through the $L^2$-inner product, so that
\[
\langle \tilde{\mathcal L}_\kappa \rho, \bar\rho\rangle
= \left(\tilde{\mathcal L}_\kappa \rho, \bar\rho\right)_{L^2},
\quad \rho,\bar\rho\in X_1.
\]
To  check that the operators $\tilde{\mathcal L}_\kappa $ and
${\mathcal L}_\kappa$
are well-defined and self-dual or self-adjoint respectively,
it clearly suffices to show only the self-adjointness
of $\mathcal L_\kappa$.

\begin{lemma}
  \label{lemma-self-adjoint-L-kappa}
  The operator $\mathcal{L}_{\kappa}$
  defined by (\ref{E:MATHCALLKAPPADEF}) is well-defined, self-adjoint,
  and $\mathcal{L}_{\kappa}-\mathrm{id}$ is compact.
\end{lemma}

\begin{proof}
The operator $\mathcal{L}_{\kappa}$ is clearly symmetric, i.e.,
\[
\left(  \mathcal{L}_{\kappa}\rho_{1},\rho_{2}\right) _{X_1}
=\left(  \rho_{1},\mathcal{L}_{\kappa}\rho_{2}\right)_{X_1}.
\]
We define the operator $K \colon X_{1}\rightarrow X_{1}$ by
$K\rho=\frac{1}{e^{\mu_{\kappa}}\Psi_{\kappa}}\bar{\mu}$.
It is straightforward to check that the modified potential
$\bar \mu$ has the properties stated in Lemma~\ref{mubarlemma},
in particular
\begin{equation}\label{eqn-mu-bar}
  \frac{1}{4\pi r^{2}}
  \frac{d}{dr}\left(  \frac{e^{-\mu_{\kappa}-3\lambda_{\kappa}}
    r^{2}}{2r\mu_{\kappa}^{\prime}+1}
  \frac{d}{dr}\left(  e^{\mu_{\kappa}+\lambda_{\kappa}}\bar{\mu}\right)  \right)
  =\rho.
\end{equation}
Evaluating the $L^2(\mathbb R^3)$-inner product of~(\ref{eqn-mu-bar}) with
$e^{\mu_{\kappa}+\lambda_{\kappa}}\bar{\mu}$, we get
\[
\int_{0}^{\infty}
\frac{e^{-\mu_{\kappa}-3\lambda_{\kappa}}r^{2}}{2r\mu_{\kappa}^{\prime}+1}
\left(\frac{d}{dr}\left(  e^{\mu_{\kappa}+\lambda_{\kappa}}\bar{\mu}\right)\right)^{2} dr
=-\int_{0}^{R_{\kappa}}e^{\mu_{\kappa}+\lambda_{\kappa}}\bar{\mu}\rho\ 4\pi r^{2}dr.
\]
Therefore
\begin{align*}
\left\Vert \nabla\left(  e^{\mu_{\kappa}+\lambda_{\kappa}}\bar{\mu}\right)
\right\Vert _{L^{2}}^{2}
& \leq C
\int_{0}^{\infty}
\frac{e^{-\mu_{\kappa}-3\lambda_{\kappa}}r^{2}}{2r\mu_{\kappa}^{\prime}+1}\left(
\frac{d}{dr}\left(  e^{\mu_{\kappa}+\lambda_{\kappa}}\bar{\mu}\right)\right)^{2}dr\\
& \leq C
\left\Vert \rho\right\Vert_{X_{1}}
\left\Vert \frac{1}{\sqrt{\Psi_{\kappa}}}e^{\mu_{\kappa}+\lambda_{\kappa}}
\bar{\mu}\right\Vert _{L^{2}}
\leq C
\left\Vert \rho\right\Vert _{X_{1}}\left\Vert e^{\mu_{\kappa}+\lambda_{\kappa}}
\bar{\mu}\right\Vert _{L^{6}}\\
& \leq C
\left\Vert \rho\right\Vert _{X_{1}}
\left\Vert \nabla\left(e^{\mu_{\kappa}+\lambda_{\kappa}}\bar{\mu}\right)
\right\Vert_{L^{2}}
\end{align*}
and thus
$\left\Vert\nabla\left(e^{\mu_\kappa+\lambda_\kappa}\bar{\mu}\right)\right\Vert_{L^{2}}
\leq C \left\Vert \rho\right\Vert_{X_{1}}$.
Here we made use of the compact support of $\frac{1}{\Psi_{\kappa}}$
and the  embedding $\dot{H}^{1}\left(\R^{3}\right)  \hookrightarrow
L^{2}\left( B_\kappa\right)$.
Since
\[
\left\Vert K\rho\right\Vert _{X_{1}}\leq C \left\Vert \frac{1}{\sqrt
{\Psi_{\kappa}}}e^{\mu_{\kappa}+\lambda_{\kappa}}\bar{\mu}\right\Vert _{L^{2}}
\]
and the mapping
$\dot{H}_{r}^{1}\left(\R^{3}\right)\ni \rho \mapsto \rho_{|S}\in L^{2}(S)$
is compact for any compact set $S\subset\mathbb R^3$, the operator $K$ is compact.
Therefore, $\mathcal{L}_{\kappa}=\mathrm{id}+K$ is self-adjoint on $X_{1}$.
\end{proof}

With the definition of $\tilde{\mathcal L}_\kappa$
the linearized velocity equation \eqref{linveleq} takes the form
\be
\dot u+ e^{-3\lambda _{\kappa }} \frac{d}{dr}(\tilde{\mathcal L}_\kappa \rho)
= 0 \label{E:ULIN2}.
\ee
It easy to check that formally the energy
\begin{equation*}
I\left( \rho ,u\right) =\int e^{3\lambda _{\kappa }}n_{\kappa }u^{2}\
dx+\left( \mathcal{\tilde{L}}_{\kappa }\rho ,\rho \right)_{L^2}
\end{equation*}
is conserved along solutions
of~\eqref{E:NLIN2},~\eqref{E:ULIN2}.
The functional $I$ maps
$X_1 \times L_{e^{3\lambda _{\kappa }}n_{\kappa }}^{2}\left( B_{\kappa}\right)$ to $\mathbb R$.
From the definition~\eqref{E:TILDELKAPPADEF} one readily checks that
for $\rho \in X_1$,
\begin{equation}\label{E:ENERGYEE1}
\left( \mathcal{\tilde{L}}_{\kappa }\rho ,\rho \right)_{L^2} =\int
e^{2\mu _{\kappa }+\lambda _{\kappa }}\Psi _{\kappa }\rho
^{2}dx-\int_{0}^{\infty }e^{\mu _{\kappa }-\lambda _{\kappa }}\left( 2r\mu
_{\kappa }^{\prime }+1\right) \,\lambda ^{2}\,dr,
\end{equation}
where we used \eqref{E:BASICVLASOV} and \eqref{E:LAMBDAOFF_EE}.

The quadratic form $I(\rho,u)$ motivates the definition
\[
v:=e^{\frac{3}{2}\lambda _{\kappa }}n_{\kappa }^{\frac{1}{2}}u
\]
so that
$v\in X_{2}$ for any
$u\in L_{e^{3\lambda _{\kappa }}n_{\kappa }}^{2}\left(B_{\kappa}\right)$.
Then the system ~\eqref{E:NLIN2},~\eqref{E:ULIN2} can be equivalently
rewritten in the form
\begin{align}
  \dot{\rho} +
  \frac{1}{r^{2}}\frac{d}{dr}
  \left(r^2 e^{-\frac{3}{2}\lambda_{\kappa }}n_{\kappa }^{\frac{1}{2}}v\right)
  & =0,   \label{eqn-LEE-rho} \\
  \dot{v}+e^{-\frac{3}{2}\lambda _{\kappa }} n_{\kappa }^{\frac{1}{2}}
  \frac{d}{dr}\left(\mathcal{\tilde{L}}_{\kappa }\rho\right) & =0.
  \label{eqn-LEE-v}
\end{align}
In order to exhibit the Hamiltonian structure of this system
it is convenient to define the modified divergence operator
$A_{\kappa }\colon D\left( A_{\kappa }\right) \subset X_{2}=X_2^\ast\rightarrow X_{1}$
by
\be \label{Akappadef}
D\left( A_{\kappa }\right) := \{ v\in X_2 \,\big| \, A_\kappa v
\text{ exists weakly in } X_1 \},\quad
  A_{\kappa }v
  :=-\frac{1}{r^{2}} \frac{d}{dr}\left( r^{2}e^{-\frac{3}{2}\lambda_{\kappa }}
  n_{\kappa }^{\frac{1}{2}}v\right),
\ee
and its dual operator (the modified gradient)
$A_{\kappa }^{\prime}\colon X_{1}' \supset D( A_{\kappa }') \to X_{2}$ by
\be \label{Akappasdef}
D\left( A_{\kappa }^{\prime}\right)
:= \{ v\in X_1' \,\big| \, A_\kappa' v \text{ exists weakly in } X_2 \},\quad
  A_{\kappa }^\prime \rho :=e^{-\frac{3}{2}\lambda _{\kappa }}
  n_{\kappa }^{\frac{1}{2}} \frac d{dr}\rho .
\ee
Then for any $v\in X_{2}$ and $\rho\in X_{1}'$ it holds that
\[
\left\langle A_{\kappa}v,\rho\right\rangle
=\left\langle v,A_{\kappa}^{\prime}\rho\right\rangle ,
\]
where $\left\langle \cdot,\cdot\right\rangle $ is the dual bracket. We also
define the adjoint operator
$A_{\kappa}^{\ast}\colon X_{1} \supset D\left(  A_{\kappa}^{\ast}\right) \to X_{2}$ by
\[
A_{\kappa}^{\ast}\rho : =e^{-\frac{3}{2}\lambda_{\kappa}}n_{\kappa}^{\frac{1}{2}}
\partial_{r}\left(  e^{2\mu_{\kappa}+\lambda_{\kappa}}\Psi_{\kappa}
\rho\right)  .
\]
Then for any $v\in X_{2}$ and $\rho\in X_{1}$
\[
\left(  A_{\kappa}v,\rho\right)_{X_1}=\left(  v,A_{\kappa}^{\ast}\rho\right)_{X_2}.
\]
We claim that both $A_\kappa$ and $A_\kappa^\prime$ are densely defined and closed.
The density is clear, and to show that they are closed we consider
the isometries
\[
S_1 = \text{id} \colon L^2(B_\kappa) \to X_2' \, \ \
S_2 = e^{\mu_\kappa+\l_\kappa/2}\Psi_\kappa^{\frac12}\colon X_1\to L^2(B_\kappa).
\]
Then to show that $A_\kappa\colon X_2'\supset D(A_\kappa)\to X_1$
is closed it suffices to show that
$\mathscr A_\kappa:=S_2A_\kappa S_1 \colon L^2(B_\kappa)\supset D(\mathscr A_\kappa)\to L^2(B_\kappa)$
is closed, where
\[
D(\mathscr A_\kappa) =
\{ v\in L^2(B_\kappa)\, \big| \,
e^{\mu_\kappa+\l_\kappa/2}\Psi_\kappa^{\frac12}\frac{1}{r^{2}}
\frac{d}{dr}\left( r^{2}e^{-\frac{3}{2}\lambda_{\kappa}n_\kappa^{\frac12}} v\right) \
\text { is weakly in } \ L^2(B_\kappa)\}.
\]
The adjoint of $\mathscr A_\kappa$ is given by
$
{\mathscr A}_\kappa^\ast \rho =
e^{-\frac32\l_\kappa}n_\kappa^{\frac12} \frac{d}{dr}\left(e^{\mu_\kappa+\l_\kappa/2}
\Psi_\kappa^{\frac12}\rho\right)
$
with
\[
D({\mathscr A}_\kappa^\ast) =
\{ \rho\in L^2(B_\kappa)\, \big| \, e^{-\frac32\l_\kappa}n_\kappa^{\frac12}
\frac{d}{dr}\left(e^{\mu_\kappa+\l_\kappa/2}\Psi_\kappa^{\frac12}\rho\right)
\ \text { is weakly in } \ L^2(B_\kappa)\}.
\]
Clearly ${\mathscr A}_\kappa^\ast$ is also densely defined.
To show that $\mathscr A_\kappa$ is closed it suffices to show
$\mathscr A_\kappa={\mathscr A_\kappa}^{\ast\ast}$. To justify the latter we
need to show that the adjoint of ${\mathscr A}_\kappa^\ast$ is precisely
$\mathscr A_\kappa$, which is a simple consequence of integration-by-parts
since by~\eqref{E:PSIDEF}
$n_\kappa \Psi_\kappa = \frac{dP}{d\rho}(\rho_\kappa)\sim_{\rho\ll1} \rho_\kappa^{\gamma-1}$
and thus the boundary term vanishes.

\begin{remark}[Notational conventions]
We adopt the notation from Yosida's book~\cite{Yosida}: For an
operator $A\colon X\rightarrow Y$ between two Hilbert spaces $X$ and $Y$, we use
$A'\colon Y'\rightarrow X'$ to
denote the dual operator and $A^{\ast}\colon Y\rightarrow X$ to denote the
adjoint operator of $A$. The operators $A^{\prime}$ and $A^{\ast}$ are related
by
\[
A^{\ast}=I_{X}A^{\prime}I_{Y}^{-1},
\]
where $I_{X}\colon X'\rightarrow X$ and
$I_{Y}\colon Y'\rightarrow Y$ are the
isomorphisms defined by the Riesz representation theorem.
\end{remark}


\begin{lemma}\label{L:LINEAREE}
  The formal linearization of the spherically symmetric Einstein-Euler system
  takes the Hamiltonian form
  \begin{equation}\label{eqn-LEE-hamiltonian}
    \frac{d}{dt}
    \begin{pmatrix} \rho \\ v \end{pmatrix}
    =J^{\kappa }L^{\kappa }
    \begin{pmatrix} \rho \\ v \end{pmatrix},
  \end{equation}
where $\left( \rho ,v\right) \in X$, and
\begin{equation}
  J^{\kappa } :=
  \begin{pmatrix}
  0 & A_{\kappa } \\
  -A_{\kappa }^{\prime } & 0
  \end{pmatrix},\quad
  L^{\kappa } :=
  \begin{pmatrix}
  \mathcal{\tilde{L}}_{\kappa } & 0 \\
  0 & \mathrm{id}
  \end{pmatrix}.
\label{defn-J-K-kappa}
\end{equation}
Then $J^{\kappa}\colon X'\rightarrow X$ and $L^{\kappa}\colon X\rightarrow X'$
are anti-self-dual and self-dual respectively.
The conserved energy functional takes the form
\begin{equation*}
I\left( \rho ,v\right) =\left( L^{\kappa }
\begin{pmatrix} \rho \\ v \end{pmatrix},
\begin{pmatrix} \rho \\ v \end{pmatrix}
\right)_{X_2}
=\int v^{2}dx+\left( \mathcal{\tilde{L}}_{\kappa}\rho ,
\rho \right)_{X_2} , \ \  \left( \rho ,v\right) \in X.
\end{equation*}
Alternatively, one can rewrite~\eqref{eqn-LEE-hamiltonian} in the
second-order form
\begin{equation}
  \frac{d^2}{dt^2}v
  +A_{\kappa }^{\prime}\tilde{\mathcal{L}}_{\kappa }A_{\kappa}v=0,
  \label{eqn-2nd order-LEE}
\end{equation}
or equivalently
\begin{equation}
  \frac{d^2}{dt^2}v
  +A_{\kappa }^{\ast}\mathcal{L}_{\kappa }A_{\kappa}v=0.
  \label{eqn-2nd order-LEE2}
\end{equation}
\end{lemma}

\begin{proof}
Eqn.~\eqref{eqn-LEE-hamiltonian} is just a rephrasing of
\eqref{eqn-LEE-rho}, \eqref{eqn-LEE-v}.
The formal conservation of energy $t\mapsto I(\rho,v)(t)$ is a direct
consequence of the Hamiltonian structure~\eqref{eqn-LEE-hamiltonian}.
That $J^\kappa$ is anti-self-dual follows easily from the definition of
$A_\kappa$ and $A_\kappa^\prime$. The self-duality of $L^\kappa$ follows from
the self-duality of $\tilde{\mathcal L}_\kappa$,
which in turn, is a trivial consequence of the self-adjointness of
$\mathcal L_\kappa$.
Finally,~\eqref{eqn-2nd order-LEE} follows formally by applying $\frac{d}{dt}$
to the
$v$-component of~\eqref{eqn-LEE-hamiltonian} and then plugging in the
equation satisfied by the $\rho$-component of~\eqref{eqn-LEE-hamiltonian}.
Equation~\eqref{eqn-2nd order-LEE2} is just a restatement
of~\eqref{eqn-2nd order-LEE2}.
\end{proof}

We will show below that the operator
$A_{\kappa}^{\ast}\mathcal{L}_{\kappa} A_{\kappa}$ is in fact self-adjoint,
cf.~Lemma~\ref{prop-self-adjoint-composition}.


\subsection{Variational approach to stability}

The variational picture explained in this section is naturally constrained to radially symmetric perturbations. However,
by comparison to the corresponding variational picture for the EV-system in Section~\ref{linEV},
we need to provide more details, as it will play an important role in our formulation of the main results for the Einstein-Euler
system in Section~\ref{SS:EEMAIN}.
The fundamental conserved quantities associated with the Einstein-Euler
system are the total ADM mass
\begin{align}\label{E:MASS}
\M(\rho) =  4\pi \int_0^\infty r^2 \rho(r)\,dr
\end{align}
and the total particle (or baryon) number
\begin{align}\label{E:NKAPPADEF}
\N_\kappa(\rho) =  4\pi \int_0^\infty r^2 e^{\l}N_\kappa(\rho) \,dr.
\end{align}
Recalling~\eqref{E:NKAPPADEF0}
we write
\be\label{E:NZERODEF}
\N_\kappa(\rho) = c_\kappa \mathscr N_0(\rho),
\ \ \mathscr N_0(\rho):= 4\pi \int_0^\infty r^2 e^{\l}N(\rho) \,dr.
\ee
By~\eqref{lambdadef}, $e^{\l} = \left(1-\frac{2m(r)}{r}\right)^{-\frac12}$
with $m(r)=4\pi\int_0^r s^2\rho(s)\, ds$,
and therefore
\[
\mathscr N_\kappa(\rho)
=   4\pi  \int_0^\infty r^2 \left(1-\frac{2m(r)}{r}\right)^{-\frac12}
N_\kappa(\rho)\,dr.
\]
We define the open set
\[
U:= \left\{ \rho\in X_1\,\big| \, \max_{r\in[0,R_\kappa]}\frac{2m(r)}{r} <1,
\ m(r)=4\pi\int_0^r s^2\rho(s)\, ds \right\}.
\]
Since by Cauchy-Schwarz $X_1 \subset L^1(B_\kappa)$,
it follows that both $\M$ and $\N_\kappa$
are well-defined functionals on $U$.

\begin{remark}
  For any $\kappa>0$ the associated fluid density $\rho_\kappa$ belongs to $U$.
  The estimate
  $\max_{r\in[0,R_\kappa]}\frac{2m(r)}{r} <1$ follows from the Buchdahl inequality,
  while the integral bound
  $\int_{B_\kappa}e^{2\mu_\kappa+\l_\kappa}\rho_\kappa^2 \Psi_\kappa <\infty$
  follows from the boundary asymptotics $\Psi_\kappa\sim \rho_\kappa^{\gamma-1}$
  in the neighborhood of the vacuum boundary $r=R_\kappa$, for some $\gamma>0$.
\end{remark}

Since $\M$ is a linear map it is clearly twice Fr\'echet differentiable and
by a standard argument so is $\N_\kappa$.
In order to compute the first and the second Fr\'echet derivative of these
conserved quantities we will as above replace $\rho$ by
$\rho_\kappa +\rho$ etc., i.e.,
$\rho$ now stands for the corresponding perturbation.
Since by definition all perturbations $\rho$ belong to $X_1$
we have $\supp \rho\subset B_\kappa$.
\begin{lemma}\label{L:VARIATIONEE}
  Let $\kappa>0$. The steady state $(\rho_\kappa,\l_\kappa,\mu_\kappa)$
  of the Einstein-Euler system given by Proposition~\ref{ssfamilies} is
  a critical point of the binding energy
  \[
  \mathscr E_\kappa : = \mathscr M- \mathscr N_\kappa,
  \]
  and
  \begin{equation*}
    D^2\mathscr E[\rho_\kappa](\rho,\rho)
    = \frac12  (\tilde{\mathcal L}_\kappa \rho, \rho)_{L^2},
  \end{equation*}
  cf.\ \eqref{E:ENERGYEE1}.
\end{lemma}

\begin{proof}
  Since $\mathscr M$ is linear its expansion is trivial. For
  $\mathscr N_\kappa$ we recall \eqref{E:LAMBDAOFF_EE} so that
  \begin{align} \label{Nlamexpand}
    &
    N_\kappa(\rho_\kappa+\rho)
    \left(1-\frac{2m_\kappa(r)}{r}-\frac{2m(r)}{r}\right)^{-\frac12}\notag\\
    &\qquad =
    e^{\lambda_\kappa} \frac{d N_\kappa}{d \rho}(\rho_\kappa) \rho
    + e^{\lambda_\kappa} n_\kappa + e^{\lambda_\kappa} n_\kappa \lambda \notag\\
    &\quad\qquad
    +\frac{1}{2} e^{\lambda_\kappa} \frac{d^2 N_\kappa}{d^2 \rho}(\rho_\kappa) \rho^2
    + e^{\lambda_\kappa} \frac{d N_\kappa}{d \rho}(\rho_\kappa) \rho \lambda
    + \frac{3}{2} e^{\lambda_\kappa} n_\kappa \lambda^2 +\ldots.
  \end{align}
  By \eqref{Nmu}, $\frac{dN}{d\rho}(\rho_\kappa) = e^{\mu_\kappa}$.
  Hence
  \[
  D \mathscr E_\kappa [\rho_\kappa](\rho)
  =
  4\pi \int_0^{R_\kappa}  r^2 \rho\,dr
  - 4\pi \int_0^{R_\kappa}  r^2 e^{\l_\kappa+\mu_\kappa} \rho\,dr
  - 4\pi \int_0^{R_\kappa} r^2 e^{\l_\kappa} n_\kappa \l \,dr.
  \]
  Using~\eqref{E:TOV3} and~\eqref{laplusmu} it follows that
  \be\label{E:NKAPPA1}
  n_\kappa = e^{\mu_\kappa-2\l_\kappa}\frac1{4\pi r}(\l_\kappa'+\mu_\kappa'),
  \ee
  which we use to rewrite the last integral as follows:
  \begin{align*}
    4\pi \int_0^{R_\kappa}  r^2 e^{\l_\kappa} n_\kappa \l \,dr
    & =
    \int_0^{R_\kappa} r e^{\mu_\kappa-\l_\kappa} (\mu_\kappa'+\l_\kappa')\l \,dr \\
    & =
    \int_0^{R_\kappa} r e^{-2\l_\kappa} \l \left(e^{\mu_\kappa+\l_\kappa}\right)' \,dr  \\
    & =
    - \int_0^{R_\kappa} \left( r e^{-2\l_\kappa} \l \right)'
    \left(e^{\mu_\kappa+\l_\kappa} - e^{\mu_\kappa+\l_\kappa}\big|_{r=R_\kappa}\right) \,dr \\
    & =
    - 4\pi \int_0^{R_\kappa} r^2 \rho \left(e^{\mu_\kappa+\l_\kappa} - 1\right) \,dr;
  \end{align*}
  in the last line we used \eqref{E:LAMBDAOFF_EE} and the fact that due to
  \eqref{laplusmu},
  $\mu_\kappa(R_\kappa)+\l_\kappa(R_\kappa)=\mu_\kappa(\infty)+\l_\kappa(\infty)=0$.
  It follows that
  $D \mathscr E_\kappa[\rho_\kappa] =0$,
  i.e., $\rho_\kappa$ is a critical point of $\mathscr E_\kappa$.

  Next we compute $D^2\mathscr E_\kappa[\rho_\kappa]$, using \eqref{Nlamexpand}:
  \begin{align*}
    D^2\mathscr E_\kappa[\rho_\kappa](\rho,\rho)
    & =
    - 2\pi \int_0^{R_\kappa} r^2 e^{\l_\kappa}\frac{d^2N_\kappa}{d\rho^2}(\rho_\kappa)\,
    \rho^2 \,dr
    - 4\pi \int_0^{R_\kappa}  r^2 e^{\l_\kappa} \frac{dN_\kappa}{d\rho}(\rho_\kappa)
    \rho \, \lambda \,dr \\
    &\quad
    -6\pi \int_0^{R_\kappa} r^2 e^{\l_\kappa}  n_\kappa \lambda^2\,dr \\
    & = : I_1 + I_2 + I_3.
  \end{align*}
  From \eqref{N2diff}, \eqref{Nmu}, and \eqref{E:PSIDEF} it follows that
  \[
    \frac{d^2N_\kappa}{d\rho^2} = -e^{2\mu_\kappa} \Psi_\kappa .
  \]
  Therefore,
  \[
  I_1 = 2\pi \int_0^{R_\kappa}  r^2 e^{2\mu_\kappa+\l_\kappa} \Psi_\kappa \rho^2  \,dr.
  \]
  In order to compute $I_2$ we recall \eqref{Nmu} and note that
  $
  \rho = \frac{1}{4\pi r^2}\left(e^{-2\lambda_\kappa} r \lambda\right)'.
  $
  Hence
  \begin{align*}
    I_2
    & =
    - \int_0^\infty e^{\mu_\kappa+\l_\kappa} \l
    \left(e^{-2\lambda_\kappa} r \lambda\right)' dr \\
    & =
    - \int_0^{\infty} e^{\mu_\kappa-\l_\kappa}
    \left(1-2 r \lambda_\kappa'\right)\l^2 \,dr
    -  \int_0^{\infty} e^{\mu_\kappa-\l_\kappa} r \frac{1}{2}\frac{d}{dr}\l^2 \,dr \\
    & =
    - \frac{1}{2} \int_0^{\infty} e^{\mu_\kappa-\l_\kappa}
    \left(1 - 3 r \l_\kappa' - r\mu_\kappa'\right) \l^2\,dr;
  \end{align*}
  since $\lambda(r)\sim 1/r$ for large $r$ there is no boundary term
  at infinity in the integration by parts.
  Finally, using \eqref{E:NKAPPA1} it follows that
  \[
  I_3 =
  - 6 \pi \int_0^\infty r^2 e^{\l_\kappa} n_\kappa \lambda^2\,dr
  =  -\frac{3}{2} \int_0^\infty r e^{\mu_\kappa-\l_\kappa} (\mu_\kappa'+\l_\kappa')
  \l^2\,dr.
  \]
  Summing $I_j$, $j=1,2,3$, we obtain the asserted formula for
  $D^2\mathscr E_\kappa$, cf.\ \eqref{E:ENERGYEE1}.
\end{proof}

In the context of stellar perturbations the natural notion of
dynamically accessible perturbations
parallels Definition~\ref{D:DYNACCDEF} for the Einstein-Vlasov system.

\begin{definition}[Linearly dynamically accessible perturbations]
  \label{D:LDAPEE}
  A function $\rho$ is a {\em linearly dynamically accessible perturbation}
  if $\rho \in $ $\overline{R\left( A_{\kappa }\right) }$,
  the closure of $R\left( A_{\kappa }\right)$ in $X_1$.
\end{definition}

This definition is motivated by~\eqref{eqn-LEE-rho}, which states that for
solutions of the linearized system,
$\dot\rho \in \overline{R\left( A_{\kappa }\right)}$. In other words,
the space of linearly dynamically accessible perturbations is identified
with the tangent space at $\rho_\kappa$ of
the leaf of all density configurations of a fixed total mass
$\M(\rho)=\M(\rho_\kappa)$.
Since $\overline{R\left( A_{\kappa }\right)}=(\ker A_\kappa')^\perp$ and
$A_\kappa' =e^{-\frac{3}{2}\lambda _{\kappa }} n_{\kappa }^{\frac{1}{2}} \pa_r$
we have  $\overline{R\left( A_{\kappa }\right)}=(\ker \pa_r)^\perp$ and therefore
\begin{equation}\label{E:CHANDRA}
    \mathcal C_\kappa: = \overline{R\left( A_{\kappa }\right) }
    =\left\{ \rho \in X_{1} \mid
    \delta\M(\rho)=\int_{B_\kappa}\rho \ dx=0\right\} .
\end{equation}
A weaker notion of linear stability is spectral stability.
\begin{definition}
  The steady state $(\rho_\kappa,\l_\kappa,\mu_\kappa)$ is
  {\em spectrally stable} if the spectrum of the linearized  operator
  $A_{\kappa }'\mathcal{\tilde{L}}_{\kappa }A_{\kappa}$ is non-negative.
\end{definition}

Combining Lemmas~\ref{L:LINEAREE} and~\ref{L:VARIATIONEE},
we can provide a
variational criterion for spectral stability of steady states of
the Einstein-Euler system, which in the physics literature goes back
to Chandrasekhar~\cite{Chandrasekhar1964}.

\begin{proposition}[Chandrasekhar stability criterion]
  \label{P:CHANDRASEKHAR}
  A steady state $(\rho_\kappa,\l_\kappa,\mu_\kappa)$ of the Einstein-Euler system
  given by Proposition~\ref{ssfamilies} is spectrally stable if and only if
  \begin{align*}
    \left( \tilde{\mathcal L}_\kappa \rho, \rho \right)_{L^2_r}  \ge 0 \
    \text{ for all}\ \rho \in \mathcal C_\kappa,
  \end{align*}
  where the constraint set $\mathcal C_\kappa=\overline{R(A_\kappa)}$
  is the set of all linearly dynamically accessible perturbations.
\end{proposition}

\begin{proof}
  We use the second-order formulation~\eqref{eqn-2nd order-LEE} and
  note that the spectral stability is equivalent to the non-negativity of the quadratic form
  \[
  \left\langle A_{\kappa }^{\prime}\mathcal{\tilde{L}}_{\kappa }A_{\kappa}v,v\right\rangle
  =\left\langle \mathcal{\tilde{L}}_{\kappa }A_{\kappa}v,A_{\kappa }v\right\rangle .
  \]
  The claim now follows from~\eqref{E:CHANDRA}.
\end{proof}

\begin{remark}
  By Lemma~\ref{L:VARIATIONEE} spectral stability is  equivalent to the
  positive semi-definiteness of the second variation of the energy
  along the manifold of linearly dynamically accessible
  perturbations.
\end{remark}

\subsection{The reduced Einstein-Euler operator}\label{SS:REDUCEDEE}

In analogy to Section~\ref{SS:REDUCED} our goal is to show that the
operator $\mathcal{\tilde{L}}_{\kappa}$ is bounded from above and
below by a Schr\"odinger-type operator.
Our main estimate is contained in Theorem~\ref{T:EULERREDUCED}.
It is of independent interest, but it is also a crucial ingredient
in our proof of Theorem~\ref{mainee}.

We first recall the definition \eqref{E:MUBARDEFNEW_EE}
of the modified potential $\bar\mu$ and the fact that
the relation~\eqref{E:BASICVLASOV} also holds in the present situation.

\begin{definition}[Reduced operator for the Einstein-Euler system]\label{D:REDUCEDEE}
  By analogy to $S_\kappa$ the {\em reduced operator}
  $\Sigma_\kappa \colon \dot H^1_r\to (\dot H^1_r)'$ is defined by
  \be\label{E:SIGMAKAPPADEF}
  \Sigma_\kappa \psi : = -\Delta_\kappa \psi - \frac{e^{\l_\kappa}}{\Psi_\kappa} \psi,
  \ee
  with $\dot H^1_r$ and $\Delta_\kappa$ defined as in
  Definition~\ref{redEVdef}.
\end{definition}

\begin{remark}
  \begin{itemize}
  \item[(a)]
    We recall from Section~\ref{micromacro} that the macroscopic quantities
    related to a steady state of the Einstein-Vlasov system constitute
    a steady state of the Einstein-Euler system.
    Indeed, in this situation $S_\kappa=\Sigma_\kappa$, since by \eqref{Nmu}
    and \eqref{E:TOV2},
    \be\label{s=sigma}
    e^{\mu_\kappa+\l_\kappa} \frac{\rho_\kappa'}{\mu_\kappa'}
    = - e^{\mu_\kappa+\l_\kappa} \frac{\rho_\kappa'}{\Psi_\kappa n_\kappa'}
    = - \frac{e^{\l_\kappa}}{\Psi_\kappa}.
    \ee
  \item[(b)]
    Just as $S_\kappa$ the operator $\Sigma_\kappa$ is  self-dual.
  \end{itemize}
\end{remark}


\begin{theorem}\label{T:EULERREDUCED}
  Let $\Sigma_\kappa$ be the reduced operator given in
  Definition~\ref{D:REDUCEDEE}.
  For every  $\mu \in \dot H^1_r$ and
  $\rho = \rho_\mu := - \frac{e^{-\mu_\kappa}}{\Psi_\kappa}\mu$ it holds that
  \begin{equation} \label{E:LOWERBOUND11}
    \left( \Sigma_{\kappa }\mu ,\,\mu \right)_{L^2_r}\geq
    \left(\mathcal{\tilde{L}}_{\kappa}\rho ,\rho \right)_{L^2_r}.
  \end{equation}
  For every $\rho \in X_{1}$ and $\bar\mu=\bar\mu_\rho$ as defined in
  \eqref{E:MUBARDEFNEW_EE} we have
  \begin{equation}\label{E:UPPERBOUND11}
    \left(\mathcal{\tilde{L}}_{\kappa}\rho ,\rho \right)_{L^2_r} \geq
    \left( \Sigma_\kappa\bar{\mu},\,\bar{\mu}\right)_{L^2_r}.
  \end{equation}
\end{theorem}

\begin{proof}
  The proof is analogous to the corresponding result,
  Theorem~\ref{T:VLASOVREDUCED},
  in the Einstein-Vlasov case. We first observe that the identities
  \eqref{lap_kappa_mubar1} and \eqref{lap_kappa_mubar2} remain valid
  in the present situation. Using these together with
  the definitions of $\mathcal{\tilde{L}}_\kappa$ and $\Sigma_\kappa$ implies
  that
  \begin{eqnarray*}
    (\Sigma_\kappa \bar\mu,\bar\mu)_{L^2_r}
    &=&
    \int_0^\infty e^{\mu_\kappa-\l_\kappa} \left(2r\mu_\kappa'+1\right)\, \lambda^2 dr
    -\int \frac{e^{\l_\kappa}}{\Psi_\kappa} \bar\mu^2\,dx \\
    &=&
    \left(\mathcal{\tilde{L}}_\kappa\rho ,\rho \right)_{L^2_r}
    -\int e^{2\mu_\kappa + \lambda_\kappa}\Psi_\kappa \rho^2 dx
    + 2 \int e^{\mu_\kappa+\l_\kappa} \bar\mu \rho \,dx
    -\int \frac{e^{\l_\kappa}}{\Psi_\kappa} \bar\mu^2\,dx \\
    &=&
    \left(\mathcal{\tilde{L}}_\kappa\rho ,\rho \right)_{L^2_r}
    - \int e^{\l_\kappa} \left(e^{\mu_\kappa} (\Psi_\kappa)^{1/2}\rho
    - (\Psi_\kappa)^{-1/2}\bar\mu\right)^2 \,dx \\
    &\leq&
    \left(\mathcal{\tilde{L}}_\kappa\rho ,\rho \right)_{L^2_r},
  \end{eqnarray*}
  and \eqref{E:UPPERBOUND11} is proven.
  To obtain the estimate \eqref{E:LOWERBOUND11}
  we observe that \eqref{upper1} remains valid for  $\mu\in \dot H^1_r$ and
  $\lambda=\lambda_\rho$ with $\rho\in X_1$, and hence
  \begin{eqnarray*}
    (-\Delta_\kappa \mu,\mu)_{L^2_r}
    &\geq&
    2  \int_0^\infty
    r e^{-2\l_\kappa} \left(e^{\mu_\kappa+\l_\kappa} \mu\right)' \lambda \,dr
    - \int_0^\infty e^{\mu_\kappa-\l_\kappa}\left(2r\mu_\kappa'+1\right)\,
    \lambda^2 \,dr\\
    &=&
    - 2  \int e^{\mu_\kappa+\l_\kappa} \mu \rho \,dx
    - \int_0^\infty e^{\mu_\kappa-\l_\kappa}\left(2r\mu_\kappa'+1\right)\,
    \lambda^2 \,dr;
  \end{eqnarray*}
  in the last equality we used \eqref{mulambdarho}, which again remains valid
  in the present situation.
  This implies that
  \begin{eqnarray*}
    \left(\Sigma_\kappa\mu, \mu\right)_{L^2_r}
    &\geq&
    - 2 \int e^{\mu_\kappa+\l_\kappa} \mu \rho \,dx
    + \left(\mathcal{\tilde{L}}_\kappa\rho ,\rho \right)_{L^2_r}
    - \int e^{2 \mu_\kappa+\l_\kappa}\Psi_\kappa \rho^2 dx
    - \int \frac{e^{\l_\kappa}}{\Psi_\kappa} \mu^2 dx\\
    &=&
    \left(\mathcal{\tilde{L}}_\kappa\rho ,\rho \right)_{L^2_r},
  \end{eqnarray*}
  provided $\rho = - \frac{e^{-\mu_\kappa}}{\Psi_\kappa}\mu$,
  and \eqref{E:LOWERBOUND11}
  is proven.
\end{proof}

\subsection{Main results for the Einstein-Euler system}\label{SS:EEMAIN}


\begin{lemma}\label{L:CORRECTION}
  The identity
  \[
  e^{2\mu_\kappa(R_\kappa)} = 1 - \frac{2M_\kappa}{R_\kappa}
  \]
  holds. In particular
  \[
  \sign (\frac{d}{d\kappa}(\mu_\kappa(R_\kappa)))
  = - \sign \frac{d}{d\kappa}\left(\frac{M_\kappa}{R_\kappa}\right).
  \]
\end{lemma}

\begin{proof}
  Recall that for any $r\ge0$ $e^{-2\l_\kappa(r)} = 1 - \frac{2m(r)}{r}$,
  where $m(r)=\int_0^r 4\pi \rho_\kappa(s) s^2\,ds$. In particular,
  $M_\kappa=\int_0^\infty 4\pi s^2\rho_\kappa\,ds = m(R_\kappa)$.
  Since $e^{\mu_\kappa(r)+\l_\kappa(r)} = 1$ for all $r\ge R_\kappa$, it finally follows that
  \[
  e^{2\mu_\kappa(R_\kappa)} = e^{-2\l_\kappa(R_\kappa)}
  = 1 - \frac{2m(R_\kappa)}{R_\kappa} = 1 - \frac{2M_\kappa}{R_\kappa}.
  \]
\end{proof}


\begin{theorem}
  \label{T:TPP}
  Let $(\rho_\kappa,\l_\kappa,\mu_\kappa)$ be a steady state of the
  Euler-Einstein system given by Proposition~\ref{ssfamilies}, and let
  \be\label{E:IKAPPADEF}
  i_\kappa =
  \begin{cases}
    1&\text{if }
    \frac{d}{d\kappa} M_\kappa \frac{d}{d\kappa}\left(\frac{M_\kappa}{R_\kappa}\right) >0, \\
    0&\text{if }
    \frac{d}{d\kappa} M_\kappa\frac{d}{d\kappa}\left(\frac{M_\kappa}{R_\kappa}\right) <0.
  \end{cases}
  \ee
  If
  $\frac{d}{d\kappa} M_\kappa \frac{d}{d\kappa}\left(\frac{M_\kappa}{R_\kappa}\right)\neq0$,
  then the number of growing modes associated with the linearized operator
  $J_\kappa L_\kappa$ equals $n^-(\Sigma_\kappa)-i_\kappa$.
\end{theorem}

\begin{proof}
  By Lemma~\ref{L:VARIATIONEE} we have
  \begin{align*}
    \langle \mathcal{\tilde{L}}_{\kappa }\rho ,\rho \rangle_{L^2}
    & =\langle D^2\mathscr E[\rho_\kappa]\rho,\,\rho\rangle
    = \langle \left(D^2\mathscr M[\rho_\kappa]
    - c_\kappa D^2\mathscr N_0[\rho_\kappa]\right)\rho,\,\rho\rangle,
  \end{align*}
  where $\mathscr N_0(\rho)$ is defined in~\eqref{E:NZERODEF}.
  Since $c_\kappa>0$ the spectral stability is equivalent to checking
  the non-negativity of
  $\langle \left(\frac{D^2\mathscr M[\rho_\kappa]}{c_\kappa}
  -  D^2\mathscr N_0[\rho_\kappa]\right)\rho,\,\rho\rangle$ for all
  $\rho\in \mathcal C_\kappa$ (i.e. $\rho\in X_1$, $\int_{D_\kappa}\rho=0$)
  by Proposition~\ref{P:CHANDRASEKHAR}.
  Define the operator
  \[
  \tilde{\mathscr E}_\kappa  = \frac1{c_\kappa} \mathscr E_\kappa
  \]
  By the above analysis it is clear that $\rho_\kappa$ is a critical point of
  $\tilde{\mathscr E}_\kappa$, i.e.,
  \be\label{E:FIRSTVAR}
  \frac1{c_\kappa}D\mathscr M[\rho_\kappa] - D\mathscr N_0[\rho_\kappa]=0.
  \ee
  It is not hard to see that the maps
  $]0,\infty[ \, \ni \kappa\mapsto
    \frac1{c_\kappa}D\mathscr M[\rho_\kappa] - D\mathscr N_0[\rho_\kappa]
    \in \mathcal L(X_1, \mathbb R)$
  and
  $]0,\infty[ \, \ni \kappa\mapsto \rho_\kappa \in X_1$
  are $C^1$-maps in the sense of Fr\'echet differentiability.
  The claim follows from the smooth dependence of the solutions
  of the steady state ODE~\eqref{yeq} on the initial data.
  Differentiating~\eqref{E:FIRSTVAR} with respect to $\kappa$ we arrive at
  \begin{align}\label{E:KEYCALC1}
    D^2\tilde{\mathscr E}_\kappa[\rho_\kappa] \frac{d\rho_\kappa}{d\kappa}
    =- \frac{d}{d\kappa}\left(\frac1{c_\kappa}\right) D\mathscr M[\rho_\kappa].
  \end{align}
  Therefore we have the relationship
  \begin{align}\label{E:ORTHOGONAL2}
    \langle D^2\tilde{\mathscr E}_\kappa[\rho_\kappa]
    \frac{d\rho_\kappa}{d\kappa}, \rho\rangle
    =- \frac{d}{d\kappa}\left(\frac1{c_\kappa}\right)
    \int_0^{R_\kappa} 4\pi r^2 \rho\,dr =0
  \end{align}
  for any $\rho\in X_1$. In other words, $\rho$ is orthogonal
  to $\frac{d\rho_\kappa}{d\kappa}$ with respect to the
  (possibly sign indefinite) distance induced by the quadratic form
  $\langle D^2\tilde{\mathscr E}_\kappa[\rho_\kappa] \cdot, \cdot\rangle $.
  Furthermore
  \begin{align*}
    \langle D^2\tilde{\mathscr E}_\kappa[\rho_\kappa]
    \frac{d\rho_\kappa}{d\kappa}, \frac{d\rho_\kappa}{d\kappa}\rangle
    & = - \frac{d}{d\kappa}\left(\frac1{c_\kappa}\right)
    \frac{d}{d\kappa} M_\kappa, \\
    & =  \frac1{c_0} e^{-\mu_\kappa(R_\kappa)}
    \frac{d}{d\kappa} M_\kappa \frac{d}{d\kappa}
    \left(\mu_\kappa(R_\kappa)\right),
  \end{align*}
  where $M_\kappa = \mathscr M(\rho_\kappa)$ is the ADM mass of the
  steady star, and we recall from Section~\ref{SS:NUMBERDENSITY}
  that $c_\kappa=c_0 e^{\mu_\kappa(R_\kappa)}$
  for some $\kappa$-independent constant $c_0$.
  Recalling the reduced operator $\Sigma_\kappa$ given by
  Definition~\ref{D:REDUCEDEE}, we then have
  \be\label{E:INDEXRELATION}
  n^-(\tilde{\mathcal L}_\kappa)
  =n^-(D^2\tilde{\mathscr E}_\kappa[\rho_\kappa])=n^-(\Sigma_\kappa).
  \ee
  First assume
  $\frac{d}{d\kappa} M_\kappa\frac{d}{d\kappa} \left(\mu_\kappa(R_\kappa)\right)<0$.
  Then by the above calculation $\frac{d\rho_\kappa}{d\kappa}$ is a
  negative direction
  for the quadratic form
  $\langle D^2\tilde{\mathscr E}_\kappa[\rho_\kappa]\cdot, \cdot\rangle$,
  i.e., for $\langle \tilde{\mathcal L}_\kappa \cdot,\cdot\rangle_{L^2}$.
  Since $\mathcal C_\kappa$ is orthogonal to $\frac{d\rho_\kappa}{d\kappa}$
  by~\eqref{E:ORTHOGONAL2}
  it follows that the number of growing dynamically accessible  modes
  is $n^-(\Sigma_\kappa)-1$.
  If
  $\frac{d}{d\kappa} M_\kappa\frac{d}{d\kappa} \left(\mu_\kappa(R_\kappa)\right)>0$,
  then by a similar argument
  \[
  n^-(\tilde{\mathcal L}_\kappa) = n^-(D^2\tilde{\mathscr E}_\kappa)
  =n^-(D^2\tilde{\mathscr E}_\kappa\big\vert_{\mathcal C_\kappa}),
  \]
  since $\mathcal C_\kappa$ is orthogonal to
  $\frac{d\rho_\kappa}{d\kappa}$ by~\eqref{E:ORTHOGONAL2}. Therefore, by~\eqref{E:INDEXRELATION} it follows that
  $n^-(\Sigma_\kappa) = n^-(L_\kappa)$ and the theorem is proved, since
  by Lemma~\ref{L:CORRECTION},
  $\text{sign}(\frac{d}{d\kappa}(\mu_\kappa(R_\kappa))
  =-\text{sign}(\frac{d}{d\kappa}\left(\frac{M_\kappa}{R_\kappa}\right))$.
\end{proof}

A simple consequence of the above theorem is the following
characterization of spectral stability.

\begin{corollary}[A near-characterization of spectral stability]
 The steady state $\rho_\kappa$ is spectrally stable if $n^-(\Sigma_\kappa)=1$ and
  $\frac{d}{d\kappa} M_\kappa\frac{d}{d\kappa}\left(\frac{M_\kappa}{R_\kappa}\right)<0$.
  The transition of stability to
    instability can only happen at the values $\kappa $ satisfying
    $n^{-}\left(\Sigma_{\kappa }\right) =1$ and
    $\frac{d}{d\kappa }M_\kappa\frac{d}{d\kappa}\left(\frac{M_\kappa}{R_\kappa}\right)=0$.
\end{corollary}

\begin{remark} \label{R:GROWINGMODES}
  \begin{itemize}
  \item[(a)]
    By~\eqref{E:KEYCALC1} it holds that at a critical point
    $\kappa^\ast$ of  $\kappa\mapsto \frac{M_\kappa}{R_{\kappa}}$,
    \[
    \mathcal{L}_{\kappa^\ast}\frac{d\rho_{\kappa}}{d\kappa}|_{\kappa=\kappa^\ast}=0.
    \]
    Therefore, $\mathcal{L}_{\kappa^\ast}$ (and equivalently the reduced operator
    $\Sigma_{\kappa^\ast}$) has a zero eigenvalue at $\kappa^\ast$.
    This is a strong indication
    that the number of negative eigenvalues of $\mathcal{L}_{\kappa}$
    increases
    or decreases by $1$ as $\kappa$ crosses $\kappa^\ast$.
     A detailed analysis of this subtle point will be given in~\cite{HaLin}.
    \item[(b)]
    Since $M_\kappa$ is the ADM mass and $R_\kappa$
    the area-radius of the $2$-sphere $t=0$ and $r=R_\kappa$
    it follows that $i_\kappa$ has an invariant geometric meaning.
  \end{itemize}
\end{remark}

We now state the main result giving a detailed description of
the linearized dynamics for the Einstein-Euler system.

\begin{theorem} \label{thm:main EE}
  The operator $J^{\kappa }L^{\kappa }$ generates a $%
  C^{0}$ group $e^{tJ^{\kappa }L^{\kappa }}$ of bounded linear operators on $X$
  and there exists a decomposition%
  \[
    X=E^{u}\oplus E^{c}\oplus E^{s},\quad
  \]
  with the following properties:
  \begin{itemize}
  \item[(i)]
    $E^{u}\left( E^{s}\right) $ consist only of eigenvectors corresponding to
    negative (positive) eigenvalues of $J^{\kappa }L^{\kappa }$ and
    when $\frac{d}{d\kappa} M_\kappa \frac{d}{d\kappa}\left(\frac{M_\kappa}{R_\kappa}\right)\neq 0$,
    \begin{equation}
      \dim E^{u}=\dim E^{s}=n^{-}\left( \Sigma_{\kappa }\right) -i_{\kappa },
      \label{unstable-dimension-formula-EE}
    \end{equation}
    where $i_{\kappa }$ is defined by~\eqref{E:IKAPPADEF}.
  \item[(ii)]
    The quadratic form $\left( L^{\kappa }\cdot ,\cdot \right)_X $
    vanishes on $E^{u,s}$, but is non-degenerate on $E^{u}\oplus E^{s} $, and
    \begin{equation*}
      E^{c}=\left\{
      \begin{pmatrix} \rho \\ v \end{pmatrix}
      \in X \mid \left( L^{\kappa }
      \begin{pmatrix} \rho \\ v \end{pmatrix},
      \begin{pmatrix} \rho_1 \\ v_1 \end{pmatrix}
      \right)_X =0\ \mbox{for all}\
      \begin{pmatrix} \rho_1 \\ v_1 \end{pmatrix}
      \in E^{s}\oplus E^{u}\right\}.
    \end{equation*}
  \item[(iii)]
    $E^{c}, E^{u}, E^{s}$ are invariant under $e^{tJ^{\kappa }L^{\kappa}}$.
  \item[(iv)]
    Let $\lambda _{u}=\min \{\lambda \mid \lambda \in \sigma (J^{\kappa}L^{\kappa }|_{E^{u}})\}>0$.
    Then there exist $M>0$ such that
    \begin{equation}
      \begin{split}
        & \left\vert e^{tJ^{\kappa }L^{\kappa }}|_{E^{s}}\right\vert
        \leq M e^{-\lambda _{u}t},\ t\geq 0,\\
        & \left\vert e^{tJ^{\kappa }L^{\kappa }}|_{E^{u}}\right\vert
        \leq M e^{\lambda _{u}t},\ t\leq 0,
      \end{split}
      \label{estimate-stable-unstable-EE}
    \end{equation}
    \begin{equation}
      \left\vert e^{tJ^{\kappa }L^{\kappa }}|_{E^{c}}\right\vert
      \leq M (1+|t|)^{k_{0}},
      \ t\in \mathbb{R},  \label{estimate-center-EE}
    \end{equation}
    where
    \begin{equation}
      k_{0}\leq 1+2i_{\kappa }.  \label{bound-k-0-EE}
    \end{equation}
  \end{itemize}
\end{theorem}

The proof of Theorem~\ref{thm:main EE} is similar to the proof of
Theorem~\ref{thm:main VE}.
We need several auxiliary lemmas.
By an argument analogous to the proof of Lemma~\ref{L:KCOMPACT},
one can show that the linear operator
$X_1\ni \rho\mapsto e^{-\mu_\kappa}\Psi_\kappa^{-1}\bar{\mu}_\rho\in X_1$
is compact and thus the only possible accumulation point of the
spectrum of $\mathcal L_\kappa$ and therefore $\tilde{\mathcal L}_\kappa$ is $1$.
Analogously to the proof of Proposition~\ref{prop-L-decompo},
we obtain the following result.

\begin{lemma} \label{lemma-L-kappa-tilde}
  There exists a decomposition of $X_{1}$ into the
  direct sum of three closed subspaces
  \[
  X_{1}=X_{1,-}\oplus \ker \tilde{\mathcal{L}}_{\kappa }\oplus X_{1,+}
  \]
  with the following properties:
  \begin{itemize}
  \item[(H1)]
    $\left( \mathcal{\tilde{L}}_{\kappa }\rho ,\rho \right)_{L^2} <0$
    for all $\rho \in X_{1,-}\backslash \{0\}$; and
    \[
    n^{-}\left( \tilde{\mathcal{L}}_{\kappa }\right) := \dim
    X_{1,-}=n^{-}\left( \Sigma_{\kappa }\right) ,\ \ \
    \dim \ker \tilde{\mathcal{L}}_{\kappa }=\dim \ker \Sigma_{\kappa }.
    \]
  \item[(H2)]
    There exists a $\delta >0$ such that
    \[
    \left( \mathcal{\tilde{L}}_{\kappa }\rho ,\rho \right)_{L^2} \geq
    \delta \left\Vert \rho \right\Vert _{X_{1}}^{2}\ ,\text{ for any }\rho \in X_{1,+}.
    \]
  \end{itemize}
\end{lemma}

By the definition~\eqref{defn-J-K-kappa} of the operator $L^{\kappa }$,
the above lemma readily implies the following.

\begin{proposition}
\label{prop-L-decomp-EE}
There exists a decomposition of $X$ into the direct
sum of three closed subspaces
\begin{equation*}
X=X_{-}\oplus \ker L^{\kappa }\oplus X_{+}
\end{equation*}
with the following properties:
\begin{itemize}
\item[(H1)]
  $\left( L^{\kappa }\begin{pmatrix} \rho \\ v \end{pmatrix},
  \begin{pmatrix} \rho \\ v \end{pmatrix}
  \right)_X <0$ for all
  $\begin{pmatrix} \rho \\ v \end{pmatrix}\in X_{-}\backslash \{0\}$; and
\begin{equation*}
n^{-}\left( L^{\kappa }\right) := \dim X_{-}=n^{-}\left( \Sigma_{\kappa}\right),\ \ \
\dim \ker L^{\kappa }=\dim \ker \Sigma_{\kappa }.
\end{equation*}
\item[(H2)]
  There exists $\delta >0$ such that
\begin{equation*}
  \left(L^{\kappa }\begin{pmatrix} \rho \\ v \end{pmatrix},
  \begin{pmatrix} \rho \\ v \end{pmatrix}
  \right)_X \geq
\delta \left( \left\Vert \rho \right\Vert_{X_{1}}^{2}
+\left\Vert v\right\Vert _{X_{2}}^{2}\right) \ \text{ for all }
\begin{pmatrix} \rho \\ v \end{pmatrix} \in X_{+}.
\end{equation*}
\end{itemize}
\end{proposition}

\begin{lemma} \label{prop-self-adjoint-composition}
  The operator
  $A_{\kappa}^{\ast}\mathcal{L}_{\kappa}A_{\kappa}\colon X_{2}\rightarrow X_{2}$
  is self-adjoint.
\end{lemma}

\begin{proof}
  The operator $A_{\kappa}\colon D\left(  A_{\kappa}\right)  \subset X_{2}\rightarrow X_{1}$
  is clearly densely defined and closed. The operator
  $A_{\kappa}^{\ast}\colon D\left(  A_{\kappa}^{\ast}\right)  \subset X_{1}\rightarrow X_{2}$ is also
  densely defined. By Lemma \ref{lemma-self-adjoint-L-kappa},  the operator
  $\mathcal{L}_{\kappa}\colon X_{1}\rightarrow X_{1}$ is self-adjoint. Moreover, by
  Lemma~\ref{lemma-L-kappa-tilde}, for the quadratic form
  $\left(  \mathcal{L}_{\kappa}\cdot,\cdot\right)_{X_1}
  =\left\langle \mathcal{\tilde{L}}_{\kappa}\cdot,\cdot\right\rangle $,
  there exists a decomposition
  $X_{1}=X_{1,-}\oplus \ker\mathcal{L}_{\kappa}\oplus X_{1,+}$
  such that the properties (H1) and (H2) from
  Lemma~\ref{lemma-L-kappa-tilde} hold.
  Therefore, by Proposition \ref{prop-self-adjoint} in the Appendix
  the operator
  $A_{\kappa}^{\ast}\mathcal{L}_{\kappa}A_{\kappa}$ is self-adjoint.
\end{proof}

Now we are ready to prove Theorem \protect\ref{thm:main EE}.

\begin{proof}[Proof of Theorem \protect\ref{thm:main EE}]
To prove (i), we use the second order
formulation \eqref{eqn-2nd order-LEE}.
We note that
\begin{equation*}
  \left\langle A_{\kappa }^{\prime}\mathcal{\tilde{L}}_{\kappa }A_{\kappa}v,
  v\right\rangle =\left\langle \mathcal{\tilde{L}}_{\kappa }A_{\kappa}v,
  A_{\kappa }v\right\rangle .
\end{equation*}

When
$\frac{d}{d\kappa }M_\kappa\frac{d}{d\kappa }
\left( \mu _{\kappa}(R_{\kappa })\right) \neq 0$,
 by~\eqref{E:CHANDRA}  and the
proof of Theorem~\ref{T:TPP}, we have
\begin{equation*}
n^{-}\left( A_{\kappa }^{\prime}\mathcal{\tilde{L}}_{\kappa }A_{\kappa
}\right) =n^{-}\left( \mathcal{\tilde{L}}_{\kappa }|_{\overline{R\left(
A_{\kappa }\right) }}\right) =n^{-}\left( \Sigma_{\kappa }\right) -i_{\kappa}<\infty .
\end{equation*}
Thus there exist exactly $n^{-}\left( \Sigma_{\kappa }\right) -i_{\kappa }$
negative eigenvalues of $A_{\kappa }'\mathcal{\tilde{L}}_{\kappa
}A_{\kappa }\ $in $X_{2}$, each of which corresponds to a pair of stable and
unstable eigenvalues for (\ref{eqn-2nd order-LEE}) and consequently for
(\ref{eqn-LEE-hamiltonian}).
This finishes the proof of part (i).
The conclusions in (ii)--(iv)
follow from \cite[Thm.~2.2]{LinZeng2017}.
\end{proof}


\subsection{A negative energy direction for $\kappa$ sufficiently large}
\label{SS:NEGATIVEEE}

In order to prove that the negative Morse index
of $\tilde{\mathcal L}_\kappa$ is at least $1$ we need to
introduce an additional
assumption on the equation of state.

\noindent
{\bf Assumption (P2).}
There exists a constant $C>0$ such that
\be \label{Pass3}
\left|P'(\rho) -\frac{1}{3}\right| \leq C \rho^{-1/2}\ \mbox{for all}\ \rho>0.
\ee
\smallskip
We note that the examples $P_\Phi$ and $P_\mathrm{NS}$ which were discussed in
Section~\ref{micromacro} satisfy this assumption, cf.\
Propositions~\ref{veqstate} and \ref{nsok}.

To analyze the sign of the quadratic form
$\langle \tilde{\mathcal L}_\kappa \rho,\rho\rangle_{L^2}$ it
will be convenient to rewrite in a slightly different form.
Motivated by the Lagrangian coordinates point of view we consider
perturbations  of the form $\rho = -(r^2 n_\kappa \xi)'/r^2$ and
the associated metric perturbation $\l$. The second variation then
can be written in the form
\begin{align}
  & A_\kappa(\xi,\xi):=
  \left( \tilde{\mathcal L}_\kappa \rho,\rho\right)_{L^2} \notag \\
  & =  - 4\pi  \int_0^\infty
  e^{2\mu_\kappa+\l_\kappa}\frac{2r\mu_\kappa'+1}{r}(\l_\kappa'+\mu_\kappa')
  n_\kappa \xi^2 r^2\,dr
  +4\pi \int_0^\infty e^{2\mu_\kappa+\l_\kappa}  \Psi_\kappa \,
  \frac{1}{r^2}\left[\pa_r\left(r^2n_\kappa \xi\right)\right]^2\,dr,
  \label{E:QUADRATICFORMEE2}
\end{align}
where $\rho = -A_\kappa \xi$,
and therefore $\rho\in \mathcal C_\kappa$ is linearly dynamically accessible
in the sense of Definition~\ref{D:LDAPEE}.

\begin{theorem}\label{T:NEGATIVEA2}
  Let the equation of state $p=P(\rho)$ satisfy the assumptions
  $({\mathrm P} 1)$ and $({\mathrm P} 2)$,
  and let $(\rho_{\kappa},\lambda_{\kappa},\mu_{\kappa})_{\kappa >0}$ be a
  family of steady states given by Proposition~\ref{ssfamilies}.
  There exists $\kappa_0>0$ sufficiently large such that for any
  $\kappa>\kappa_0$ there exists a nontrivial perturbation
  $\rho\in \mathcal C_\kappa$
  satisfying $(\tilde{\mathcal L}_\kappa \rho,\rho)_{L^2}<0$.
  In particular, the negative Morse index $n^-(\tilde{\mathcal L}_\kappa)$
  of the operator $\tilde{\mathcal L}_\kappa$ is at least 1.
\end{theorem}

\begin{proof}
Some of the main ideas of this proof
are analogous to the ones used in the proof
of Theorem~\ref{T:NEGATIVEA},
but due to the local character of the operator $L_\kappa$, the proof is actually
simpler than the one of Theorem~\ref{T:NEGATIVEA}.
Our starting point is the formula~\eqref{E:QUADRATICFORMEE2}.
We localize the perturbation $\xi$ to the interval $[r_\kappa^1, r_\kappa^2]$
by setting $\xi(r) = r^a\chi_\kappa(r)$ for an $a\in\mathbb R$
to be specified later; the smooth cut-off function $\chi_\kappa$ is as
in~\eqref{chipbounds}.
We split the quadratic form~\eqref{E:QUADRATICFORMEE2} into two parts:
$A_\kappa=A_{\kappa,1}+A_{\kappa,2}$, where
\beas
A_{\kappa,1}
&=&
4\pi  \int_{r_\kappa^1}^{r_\kappa^2} e^{2\mu_\kappa+\l_\kappa}
\left[- \frac{2r\mu'_\kappa+1}{r}(\l'_\kappa +\mu'_\kappa)n_\kappa r^{2+2a}
+ \Psi_\kappa \frac{1}{r^2}\left(\left(r^{2+a}n_\kappa \right)'\right)^2\right]\,
\chi_\kappa^2 dr, \\
A_{\kappa,2}
&=&
4\pi \left( \int_{r_\kappa^1}^{2 r_\kappa^1}+  \int_{r_\kappa^2/2}^{r_\kappa^2}\right)
e^{2\mu_\kappa+\l_\kappa} \Psi_\kappa \frac{1}{r^2}
\left[2 \left(r^{2+a}n_\kappa \right)' r^{2+a} n_\kappa \chi_\kappa
+ r^{4+2 a} n_\kappa^2 \chi'_\kappa\right]\,\chi'_\kappa\, dr.
\eeas
In order to understand the asymptotic behavior of these expressions as
$\kappa \to \infty$ we need to supplement the results from
Propositions~\ref{P:EVTOBKZ}
and \ref{muass} with corresponding information on $n_\kappa$, $\Psi_\kappa$,
and $n'_\kappa$. But by \eqref{E:TOV2} and \eqref{E:TOV3} it holds that
\[
n_\kappa = e^{\mu_\kappa}\left(\rho_\kappa + p_\kappa\right),\
n'_\kappa = -\mu'_\kappa n_\kappa \left(\frac{d P}{d\rho}(\rho_\kappa)\right)^{-1},\
\Psi_\kappa = \frac{1}{n_\kappa} \frac{d P}{d\rho}(\rho_\kappa).
\]
We also recall \eqref{laplusmu} so that
\bea
A_{\kappa,1}
&=&
4\pi  \int_{r_\kappa^1}^{r_\kappa^2} e^{3\mu_\kappa+\l_\kappa}\chi_\kappa^2 r^{2+2a}
\left(\rho_\kappa + p_\kappa\right)
\Biggl[
  - 4\pi e^{2\l_\kappa}\left(\rho_\kappa + p_\kappa\right)\left(2r\mu'_\kappa+1\right)
\nonumber \\
&&
\qquad \qquad {}+ (2+a)^2 \frac{1}{r^2}\frac{d P}{d\rho}(\rho_\kappa)
- 2 (2+a) \frac{1}{r}\mu'_\kappa
+ \left(\mu'_\kappa\right)^2\left(\frac{d P}{d\rho}(\rho_\kappa)\right)^{-1}
\Biggr]\, dr. \quad \label{AK1} \nonumber
\eea
By Proposition~\ref{P:EVTOBKZ},
\[
\rho_\kappa (r) \geq \frac{3}{56\pi} r^{-2} + \mathrm{O}(1)
\geq \frac{3}{56\pi} R_\kappa^{-2} + \mathrm{O}(1) \geq e^{3\kappa},\
r\in [r_\kappa^1,r_\kappa^2]
\]
for $\kappa$ sufficiently large, and hence by (P2),
\[
\left|\frac{d P}{d\rho}(\rho_\kappa)-\frac{1}{3}\right|+
\left|\left(\frac{d P}{d\rho}(\rho_\kappa)\right)^{-1}-3\right| \leq C e^{-\kappa}
\]
on $[r_\kappa^1, r_\kappa^2]$ and for $\kappa$ sufficiently large. Together with
the estimates from Proposition~\ref{P:EVTOBKZ} this implies that for
any $\epsilon>0$ the estimates
\[
\left|\frac{d P}{d\rho}(\rho_\kappa)-\frac{1}{3}\right|,\
\left|\left(\frac{d P}{d\rho}(\rho_\kappa)\right)^{-1}-3\right|,\
\left|e^{2 \l_\kappa}-\frac{7}{4}\right|,\
\left|r^2 \left(\rho_\kappa + p_\kappa\right) -\frac{4}{56\pi}\right|,\
\left|2r\mu'_\kappa -1\right|
< \epsilon
\]
hold on $[r_\kappa^1, r_\kappa^2]$ for $\kappa$ sufficiently large.
If we substitute the exact asymptotic values of these expressions
and the estimate from Proposition~\ref{muass} into \eqref{AK1}
the integrand becomes
proportional to $r^{-1}$, provided we choose $a=-1/4$.
With this choice and with the exact asymptotic values
a short computation shows for the term $[\ldots]$ in \eqref{AK1} that
\[
[\ldots]_{|\epsilon=0} = -\frac{47}{48} r^{-2}.
\]
Choosing $\epsilon >0$ sufficiently small and $\kappa$ sufficiently large
it therefore follows that
\beas
A_{\kappa,1}
&\leq&
4\pi  \int_{r_\kappa^1}^{r_\kappa^2} e^{3\mu_\kappa+\l_\kappa}\chi_\kappa^2 r^{2+2a}
\left(\rho_\kappa + p_\kappa\right) \left[-\frac{1}{2} r^{-2}\right]\,dr\\
&\leq&
- C C_\kappa^{3/2} \int_{2 r_\kappa^1}^{r_\kappa^2/2} r^{-1} dr
\leq - C C_\kappa^{3/2} \ln \kappa,
\eeas
where $C>0$ is independent of $\kappa$ and $C_\kappa$ is the constant
introduced in Proposition~\ref{muass}.

Turning now to $A_{\kappa,2}$ we must show that this term
is smaller in modulus for $\kappa$ large. We split this term further:
\[
|A_{\kappa,2}|
\leq 4\pi
\left( \int_{r_\kappa^1}^{2r_\kappa^1}+  \int_{r_\kappa^2/2}^{r_\kappa^2}\right)
f_{\kappa,21} dr
+4\pi
\left( \int_{r_\kappa^1}^{2r_\kappa^1}+  \int_{r_\kappa^2/2}^{r_\kappa^2}\right)
f_{\kappa,22} dr,
\]
where
\beas
f_{\kappa,21}
&=&
e^{2\mu_\kappa+\l_\kappa} \Psi_\kappa \frac{1}{r^2}
2 \left((2+a) r^{a+1} n_\kappa + r^{a+2} |n'_\kappa|\right) r^{2+a} n_\kappa
\chi_\kappa |\chi'_\kappa|\\
&\leq&
C e^{3\mu_\kappa} r^{2 a-1} |\chi'_\kappa| \leq C C_\kappa^{3/2}|\chi'_\kappa|,
\eeas
and
\[
f_{\kappa,22}
=
e^{2\mu_\kappa+\l_\kappa} \Psi_\kappa \frac{1}{r^2} r^{4+2 a} n_\kappa^2
(\chi'_\kappa)^2
\leq
C e^{3\mu_\kappa} r |\chi'_\kappa|^2 \leq C C_\kappa^{3/2} r |\chi'_\kappa|^2.
\]
We recall that on the interval $[r_\kappa^1,2r_\kappa^1]$ the
estimate $|\chi'_\kappa|\leq 2/r_\kappa^1$ holds, and
$|\chi'_\kappa|\leq 4/r_\kappa^2$ holds $[r_\kappa^2/2, r_\kappa^2]$. Thus
\[
|A_{\kappa,2}| \leq C C_\kappa^{3/2},
\]
and together with the estimate for $A_{\kappa,1}$ this implies that
$A_{\kappa} < 0$ for $\kappa$ sufficiently large.
\end{proof}

\begin{corollary}[Unstable star models]\label{C:UNSTABLESTARS}
  By Theorem~\ref{T:NEGATIVEA2} and Proposition~\ref{P:CHANDRASEKHAR}
  steady states of the Einstein-Euler system are dynamically unstable
  for sufficiently large central redshifts $\kappa\gg1$. In particular
  the stable and unstable spaces $E^{u,s}$
  are non-empty and the exponential trichotomy with precise estimates
  described in Theorem~\ref{thm:main EE} holds.
\end{corollary}

\subsection{The micro-macro stability principle}\label{SS:MICROMACRO}

The following theorem is a general-relativistic analogue of Antonov's
First Law which can be expressed as follows:
In the context of Newtonian physics a spherically symmetric stellar system
with phase space density of the form $f=\phi(E)$ with $\phi$ strictly
decreasing on its support is stable if the barotropic star with the
corresponding equation of state and the same macroscopic density is stable,
cf.\ \cite[Sect.~5.2]{BiTr}.

\begin{theorem}[Micro-macro stability comparison]\label{T:MICROMACRO}
  Let $\Phi$ be a microscopic equation of state satisfying the
  assumption $(\Phi1)$ and let $P=P_\Phi$ be the macroscopic equation
  of state generated by $\Phi$ according to Proposition~\ref{veqstate}.
  Now fix some $\kappa>0$. If the corresponding steady state
  $(f_\kappa, \l_\kappa,\mu_\kappa)$ of the Einstein-Vlasov system
  is spectrally unstable, then the same is true for
  the corresponding steady state $(\rho_\kappa,\l_\kappa,\mu_\kappa)$
  of the Einstein-Euler system.
\end{theorem}

\begin{proof}
  If $f_\kappa$ is spectrally unstable, then by the linear stability criterion
  for the Einstein-Vlasov system
  there exists a function
  $f\in  R\left( \mathcal{B}_{\kappa }\right)$
  such that $\mathcal A_\kappa(f,f)<0$, where
  $\mathcal A_\kappa(f,f)$ is defined
  by~\eqref{E:QUADRATICFORM}, \eqref{E:LAMBDAOFF}.
  We show that $f\in R\left( \mathcal{B}_{\kappa }\right) $ implies that
  $\int \rho_{f}dx=0$,
  where $\rho_{f}=\int f\, \p \, dv$.
  Indeed, let $f=\mathcal{B}_{\kappa }h$, $h\in D(\mathcal B_\kappa)$.
  Then
  \begin{eqnarray*}
    \int \rho_{f} dx
    &=&
    \iint f\,\p\, dv\,dx
    =\iint \frac{e^{\lambda _{\kappa }}}{|\phi_{\kappa }^{\prime }(E_{\kappa })|}
    \left(\mathcal{B}_{\kappa }h\right) \,
    \left( |\phi_{\kappa }^{\prime }(E_{\kappa})|\,E_{\kappa }
    e^{-\mu _{\kappa }-\lambda_{\kappa }}\right) dv\,dx \\
    &=&
    -\iint \frac{e^{\lambda_{\kappa }}}{|\phi_{\kappa }^{\prime }(E_{\kappa})|}h\,
    \mathcal{B}_{\kappa }\left( |\phi _{\kappa }^{\prime }(E_{\kappa})|\,E_{\kappa }
    e^{-\mu _{\kappa }-\lambda _{\kappa }}\right) dv\,dx=0.
  \end{eqnarray*}
  Here we used the facts that $\mathcal{B}_{\kappa }$ is
  anti-self-adjoint in $X$
  and
  $\mathcal{B}_{\kappa }\left( |\phi _{\kappa }^{\prime }(E_{\kappa })|
  E_{\kappa}e^{-\mu _{\kappa }-\lambda _{\kappa }}\right) =0$; the latter follows
  from the definition of $\mathcal{B}_{\kappa }$ and \eqref{11}.
  On the other hand by the Cauchy-Schwarz inequality,
  \[
  \rho_{f}^{2} \leq \int \frac{f^{2}}{|\phi _{\kappa }^{\prime }(E_{\kappa })|}dv\;
  \int |\phi _{\kappa }^{\prime }(E_{\kappa })| \p^{2}dv.
  \]
  By \eqref{E:IDLOC} and \eqref{s=sigma},
  \[
  \int |\phi_\kappa'(E_\kappa)| \p^2\,dv
  = \frac{e^{-2\mu_\kappa}}{\Psi_\kappa} ,
  \]
  which simply says that in the present situation $S_\kappa=\Sigma_\kappa$.
  This implies that
  \[
  \int \frac{f^{2}}{|\phi _{\kappa }^{\prime }(E_{\kappa })|}dv
  \geq \rho_{f}^{2}\Psi _{\kappa }e^{2\mu _{\kappa }} ,
  \]
  and
  \[
  \iint \frac{e^{\lambda _{\kappa }}}{|\phi _{\kappa }^{\prime }(E_{\kappa })|}
  f^{2}\,dv\,dx
  \geq \int e^{2\mu_{\kappa }+\lambda_{\kappa }}\Psi _{\kappa}\rho_{f}^{2} dx,
  \]
  which in turn implies that
  \be\label{E:MICROMACRO}
  \mathcal A_\kappa(f,f) \ge
  \left( \mathcal{\tilde{L}}_{\kappa }\rho_f ,\rho_f \right)_{L^2}.
  \ee
  It follows that
  $\left(\tilde{\mathcal L}_{\kappa} \rho_f ,\rho_f \right)_{L^2}<0$
  with $f$ chosen as above. By Proposition~\ref{P:CHANDRASEKHAR}
  this implies that the steady state $\rho_\kappa$ of the Einstein-Euler system
  is unstable as well.
\end{proof}

A corollary of the above micro-macro stability estimate is the stability
part of the turning point principle for the Einstein-Vlasov system.

\begin{theorem}[The stability part of the turning point principle for
                the Einstein-Vlasov system]
  Under the assumptions of Theorem~\ref{T:MICROMACRO} the following holds:
  If $n^-(S_\kappa)=1$ and
  $\frac{d}{d\kappa} M_\kappa\frac{d}{d\kappa}
  \left(\mu_\kappa(R_\kappa)\right)<0$, then the steady state
  $f_\kappa$ is linearly stable. Here $S_\kappa$ denotes the reduced operator
  defined in~\eqref{E:SKAPPADEFINITION}.
\end{theorem}

\begin{proof}
  Let $(\rho_\kappa,\l_\kappa,\mu_\kappa)$ be the associated steady state of the
  Einstein-Euler system. Since $S_\kappa=\Sigma_\kappa$ by~\eqref{s=sigma}, we
  have $n^-(\Sigma_\kappa)=1$. Thus by the part (i) of the turning point
  principle for the Einstein-Euler system, Theorem~\ref{T:TPP}, we conclude
  that $(\rho_\kappa,\l_\kappa,\mu_\kappa)$ is spectrally stable.
  By~\eqref{E:MICROMACRO} it follows that $\mathcal A_\kappa(f,f)\ge0$ for
  all linearly dynamically accessible $f\in \overline{R(\mathcal B_\kappa)}$.
  This in turn implies the linear stability of $f_\kappa$.
\end{proof}

\begin{remark}
  The proof of the full turning point principle for the Einstein-Vlasov system
  is considerably harder than in the case of the Einstein-Euler system since
  the space of linearly dynamically accessible perturbations
  $R(\mathcal B_\kappa)$ is of infinite codimension, by contrast to
  $\mathcal C_\kappa$,  see~\eqref{E:CHANDRA}.
\end{remark}


\appendix

\setcounter{section}{0}
\setcounter{equation}{0}

\renewcommand{\thesection}{\Alph{section}}
\renewcommand{\theequation}{\Alph{section}.\arabic{equation}}

\section{Auxiliary results for self-adjoint operators} \label{A}

In this appendix we collect some auxiliary results on symmetric
and self-adjoint operators;
to make the present paper self-contained, we include the proofs.
First we show how estimates between two operators translate to
corresponding information on their negative Morse index.

\begin{lemma} \label{op_est_morse}
  Let $X$ and $Y$ be Hilbert spaces with scalar products $(\cdot,\cdot)_X$
  and  $(\cdot,\cdot)_Y$, let $A\colon X\supset D(A) \to X$ and
  $B \colon Y\supset D(B) \to Y$ be linear, symmetric operators,
  let $P \colon D(A) \to D(B)$ be linear, and assume that
  for all $x\in D(A)$,
  \be \label{ABest}
  (BPx,Px)_Y \leq (Ax,x)_X.
  \ee
  Then $n^-(A) \leq n^-(B)$.
  The analogous conclusion holds if  $B \colon Y \to Y'$ is linear and self-dual instead.
\end{lemma}

\begin{proof}
  Let $V\subset D(A)$ be a subspace such that $(Ax,x)_X <0$ for all
  $x\in V\setminus \{0\}$, and define $W:= P(V)$, which is a subspace
  of $D(B)$. If $x\in V$ such that $Px=0$, then by \eqref{ABest}
  and the choice of $V$ it follows that $x=0$. Hence $P\colon V \to W$
  is one-to-one and onto so that $\dim V = \dim W$.
  Let $y\in W\setminus\{0\}$ so that $y=P x$ for some $x\in V\setminus\{0\}$.
  Then by  \eqref{ABest} and the choice of $V$
  it holds that $(By,y)_Y < 0$. Hence $n^-(B)\geq \dim V$, and the assertion
  follows.
\end{proof}

Next we consider the spectral properties of a self-adjoint operator with
negative Morse index.

\begin{proposition} \label{saspec}
  Let $A \colon X\supset D(A)\rightarrow X$ be a self-adjoint operator on a
  Hilbert space $X$ with scalar product $(\cdot,\cdot)$ and assume that
  $n^-(A) < \infty$.
  Then the negative part of its spectrum, $\sigma(A)\,\cap\, ]-\infty,0[$,
  consists of at most finitely many eigenvalues which have finite
  multiplicities.
  If $n^-(A) >0$, then $A$ has a negative eigenvalue of finite multiplicity.
\end{proposition}

\begin{proof}
  Let $E$ denote the resolution of the identity corresponding to $A$; we follow
  the notation in \cite{Ru}. Then for all $x\in D(A)$,
  \begin{equation} \label{axx}
    (Ax,x) = \int_{-\infty}^\infty t dE_{xx}(t)
    = \int_0^\infty (E(]\tau,\infty[)x,x)\,d\tau
    - \int_{-\infty}^0 (E(]-\infty,\tau[)x,x)\,d\tau,
  \end{equation}
  cf.~\cite[Thm.~13.30, Def.~12.17]{Ru}.
  Next we note that
  \begin{equation}\label{domain}
    R(E(I)) \subset D(A)\ \mbox{for any bounded interval}\ I\subset \R;
  \end{equation}
  as above, $R(E(I))$ denotes the range of the operator $E(I)$. To see the
  latter we note that
  \begin{equation}\label{domaincrit}
    \int_{-\infty}^\infty t^2 dE_{xx}(t) =
    \int_0^\infty (E(\R\setminus[-\sqrt{\tau},\sqrt{\tau}])x,x)\, d\tau =
    \int_0^\infty \left(||x||^2 - (E([-\sqrt{\tau},\sqrt{\tau}])x,x)\right)\,
    d\tau.
  \end{equation}
  Now let $x=E(I) y$ with some $y\in D(A)$.
  For $\tau$ sufficiently large so that $I\subset [-\sqrt{\tau},\sqrt{\tau}]$,
  \[
  (E([-\sqrt{\tau},\sqrt{\tau}])x,x) =
  (E([-\sqrt{\tau},\sqrt{\tau}]) E(I) y,E(I) y) =
  (E(I) y,E(I) y) = ||x||^2.
  \]
  Hence the integral in \eqref{domaincrit} converges and \eqref{domain}
  follows from \cite[Thm.~13.24]{Ru}.

  Now we assume that the essential spectrum of $A$
  intersects the negative reals,
  i.e., there exists $\lambda \in ]-\infty,0[\, \cap\, \sigma_\mathrm{ess}(A)$.
  We choose an $\epsilon>0$ such that
  $I:=]\lambda-\epsilon,\lambda+\epsilon[ \,\subset\, ]-\infty,0[$.
  Then $\dim R(E(I)) = \infty$, see \cite[X Remark 1.11]{Kato}.
  We fix some $k>n^-(A)$ and choose linearly independent vectors
  $x_1,\ldots,x_k \in R(E(I))$.
  By \eqref{domain}, $x_j \in D(A)$, and $x_j=E(I) y_j$
  with $y_j\in X$, $j=1,\ldots,k$. We consider the integrals in \eqref{axx}
  with $x=x_j$. For $\tau > \lambda +\epsilon$ it holds that
  \[
  (E(]-\infty,\tau[) x_j, x_j)=(E(]-\infty,\tau[) E(I) y_j, E(I) y_j)
    = (E(I) y_j, E(I) y_j)=||x_j||^2 >0.
  \]
  For $\tau >0$ we have that $]\tau,\infty[ \cap I =\emptyset$ and therefore
  $(E(]\tau,\infty[) x_j, x_j)=0$, and since $E(]-\infty,\tau[)$ is a
  projection, $E(]-\infty,\tau[) x_j, x_j)\geq 0$ for all $\tau$. Therefore by
  \eqref{axx}, $(A x_j,x_j)<0$ for all $j=1,\ldots k$, but since
  $(x_1,\ldots,x_k)$ are linearly independent and $k>n^-(A)$ this is a
  contradiction.
  Hence  $\sigma(A) \cap ]-\infty,0[$ consists only of isolated
  eigenvalues with finite multiplicities, and there can be at most $n^-(A)$
  of them.

  If we now assume that in addition $n^-(A)>0$, then \cite[Thm.~13.31]{Ru}
  shows that $\sigma(A)$ must intersect the negative reals, and the proof
  is complete.
\end{proof}

Finally we give a special version of an abstract result from
\cite{LinZeng2017b}, which is used in the proof of
Lemma~\ref{prop-self-adjoint-composition}.

\begin{proposition} \label{prop-self-adjoint}
Let $X,Y$ be Hilbert spaces and
assume $A \colon Y\supset D(A)\rightarrow X$ is densely defined and closed. In
addition, assume the following:
\begin{itemize}
\item[(i)]
  The adjoint operator $A^{\ast}\colon X\supset D(A^{\ast})\rightarrow Y$ is also
  densely defined.
\item[(ii)]
  The operator $L\colon X\rightarrow X$ is bounded and self-adjoint
  so that $\left(  Lu,v\right)  $ is a bounded
  symmetric bilinear form
  on $X$, where $(\cdot,\cdot)$ denotes the inner product in $X$.
\item[(iii)]
  There exists a decomposition of $X$ into the direct sum of three closed
  subspaces
  \begin{equation}
    X=X_{-}\oplus\ker L\oplus X_{+},\quad n^{-}(L) := \dim X_{-}
    <\infty,\ \ker L<\infty. \label{decom-X-A}
  \end{equation}
  such that the following holds:
  \begin{itemize}
  \item[(L1)]
   $\left\langle Lu,u\right\rangle <0$ for all $u\in X_{-}\backslash\{0\}$,
  \item[(L2)]
    there exists $\delta>0$ such that
    \[
    \left(  Lu,u\right)  \geq\delta\left\Vert u\right\Vert _{X}^{2}
    \ \text{for all}\ u\in X_{+}.
    \]
  \end{itemize}
\end{itemize}
Then the operator $B=A^{\ast}L A \colon Y\rightarrow Y$ is self-adjoint
with domain
$D\left(  A^{\ast}LA\right)  \subset D\left(A\right)$.
\end{proposition}

\begin{proof}
First, we show that there exists a new decomposition
\[
X=X_{-}^{\prime}\oplus\ker L\oplus X_{+}^{\prime}
\]
such that (L1) and (L2) are satisfied and in addition
\begin{equation}
\left(  X_{+}^{\prime}\right)  ^{\perp}:=\left\{  x\in X\ |\ \left(
x,x_{+}\right)  =0\ \mbox{for all}\ x_{+}\in X_{+}^{\prime}\right\}  \subset
D(A^{\ast})\label{p1}
\end{equation}
and
\begin{equation}
(Lx_{+},x_{-})=0\ \mbox{for all}\ x_{+}\in X_{+}^{\prime},\ x_{-}\in
X_{-}^{\prime}.\label{p2}
\end{equation}
Indeed, Let $(f_{1},\ldots,f_{k})$ be a basis of $X_{+}^{\perp}$, where
$k=\dim\ker L+n^{-}(L)$ is the co-dimension of $X_{+}$. As $D(A^{\ast})$ is
dense in $X$, one may take $g_{j}\in D(A^{\ast})$ sufficiently close to
$f_{j}$, $j=1,\ldots,k$. Let
\[
X_{+}^{\prime}=\{u\in X\mid\left(  g_{j},u\right)  =0,\  j=1,\ldots ,k\}.
\]
Since $X_{+}^{\prime}$ is close to $X_{+}$, it is easy to show that (L2) is
satisfied on $X_{+}^{\prime}$. Define
\[
X_{1}^{\prime}=\left\{  x\in X \mid \left(  Lx,x_{+}\right)  =0
\ \mbox{for all}\ x_{+}\in X_{+}^{\prime}\right\}  .
\]
Clearly, $\ker L\subset X_{1}^{\prime}$ and $\dim X_{1}^{\prime}=k$. Let
$X_{1}^{\prime}=\ker L\oplus X_{-}^{\prime}$. Then (L1) is satisfied on
$X_{-}^{\prime}$, since
\[
n^{-}\left(  L|_{X_{-}^{\prime}}\right)  =n^{-}\left(  L\right)  -n^{-}\left(
L|_{X_{+}^{\prime}}\right)  =n^{-}\left(  L\right)
\]
and $\dim X_{-}^{\prime}=n^{-}(L)$.

For convenience, we denote $\left(  X_{-}^{\prime},X_{+}^{\prime}\right)  $ by
$\left(  X_{-},X_{+}\right)  $ below and let $X_{1}=\ker L\oplus X_{-}$. Let
$P_{+,1}$ be the projection of $X$ to $X_{+},_{1}$. By (\ref{p2}),
\[
P_{1}^{\ast}L P_{+}=P_{+}^{\ast}L P_{1}=0,
\]
which implies that
\[
L=P_{+}^{\ast}L P_{+}+P_{1}^{\ast}L P_{1}:= L_{+}-L_{1}
\]
with symmetric $L_{+,1}\geq0$. Since
$R(P_{1}^{\ast})=X_{+}^{\perp}\subset D(A^{\ast})$,
by the Closed Graph Theorem $A^{\ast}P_{1}^{\ast}$ is bounded.
Therefore $P_{1}A$ has a continuous extension
$(A^{\ast}P_{1}^{\ast})^{\ast}=(P_{1}A)^{\ast\ast}$, i.e., $P_{1}A$ is bounded.
Therefore $P_{+}A$ is closed
and densely defined. Let $S_{+} \colon X\rightarrow X$
be a bounded symmetric linear
operator such that
\[
S_{+}^{\ast}S_{+}=S_{+}^{2}=L_{+},\quad S_{+}\geq0.
\]
For any $x\in X_{+}$,
\[
\left\Vert S_{+}x\right\Vert ^{2}=(L_{+}x,x)=(Lx,x)\geq\delta\left\Vert
x\right\Vert ^{2},
\]
which implies that for all $x\in X_{+}$,
\[
\left\Vert S_{+}x\right\Vert \geq\delta\left\Vert x\right\Vert.
\]
This lower bound of $S_{+}$ implies that $T_{+}:= S_{+}P_{+}A$ is
also closed with the dense domain $D(T_{+})=D(A)$ and thus $T_{+}^{\ast}T_{+}$
is self-adjoint. Then
\begin{align*}
A^{\ast}L A= &  A^{\ast}P_{+}^{\ast}L_{+}P_{+}A-A^{\ast}P_{1}^{\ast}L_{1}
P_{1}A=(A^{\ast}P_{+}^{\ast}S_{+})(S_{+}P_{+}A)-A^{\ast}P_{1}^{\ast}L_{1}
P_{1}A\\
:= &  T_{+}^{\ast}T_{+}-B_{1},
\end{align*}
where $B_{1}=A^{\ast}P_{1}^{\ast}L_{1}P_{1}A$ is bounded and symmetric.
Therefore, $A^{\ast}L A$ is self-adjoint by the Kato-Rellich Theorem.
\end{proof}

\bigskip

\noindent
{\bf Acknowledgments.}
MH acknowledges the support of EPSRC Grant EP/N016777/1 and the
EPSRC Early Career Fellowship EP/S02218X/1.
ZL is supported in part by NSF grants DMS-1715201 and
DMS-2007457.


\begin{thebibliography}{AAAA}

\bibitem{AnKo}
{\sc Andersson, L., Korzy\'nski, M.,}
Variational principle for the Einstein-Vlasov equations.
{\em Available online at https://arxiv.org/abs/1910.12152}

\bibitem{Andr2007}
  {\sc Andr\'{e}asson, H.},
  On static shells and the Buchdahl inequality
  for the spherically symmetric Einstein-Vlasov system.
  {\em Commun.\ Math.\ Phys.}, {\bf 274}, 409--425 (2007)

\bibitem{Andr2008}
  {\sc Andr\'{e}asson, H.},
  Sharp bounds on $2m/r$ of general spherically
  symmetric static objects.
  {\em J.~Differential Equations}, {\bf 245}, 2243--2266 (2008)

\bibitem{andfajthal}
{\sc Andr\'{e}asson, H., Fajman, D., Thaller, M.},
Models for self--gravitating photon shells and geons.
{\em Ann.\ Henri Poincar\'e}
{\bf 18}, 681--705 (2017)

\bibitem{AndKun}
{\sc Andr\'{e}asson, H., Kunze, M.},
Comments on the paper ``Static solutions of the Vlasov-Einstein system''
by G.~Wolansky.
{\em Preprint. Available on arXiv as arXiv:1805.10683}

\bibitem{AnRe1}
{\sc Andr\'{e}asson, H., Rein, G.},
A numerical investigation of the stability of steady states
and critical phenomena for the spherically symmetric
Einstein-Vlasov system.
{\em Class.\ Quantum Grav.}\
{\bf 23}, 3659--3677 (2006)


\bibitem{An1961}
{\sc Antonov, A. V.},
Remarks on the problem of stability in stellar dynamics.
{\em Soviet Astr., A J.},
{\bf 4}, 859--867 (1961)

\bibitem{BaThMe1966}
{\sc Bardeen, J. M., Thorne, K. S., Meltzer, D. W.},
A Catalogue of Methods for Studying the Normal Modes of Radial
Pulsation of General-Relativistic Stellar Models.
{\em Astrophysical Journal},
{\bf 145}, p.505 (1966)


\bibitem{BiTr}
{\sc Binney, J., Tremaine, S.},
{\em Galactic Dynamics}. 2nd edition.
Princeton Series in Astrophysics (2008)

\bibitem{BKZ}
{\sc Bisnovatyi-Kogan, G. S., Zel'dovich, Ya. B.},
Models of clusters of point masses with great central red shift.
{\em Astrofizika}, {\bf 5}, 223--234 (1969)

\bibitem{BKThorne1970}
{\sc Bisnovatyi-Kogan, G.~S., Thorne, K.~S.},
Relativistic gas spheres and clusters of point masses with
arbitrarily large central redshifts: can they be stable?
{\em Astrophysical J.} {\bf 160} (1970)

\bibitem{Buch}
  {\sc Buchdahl, H.~A.},
  General relativistic fluid spheres,
  {\em Phys.\ Rev.}, {\bf 116}, 1027–-1034 (1959)

\bibitem{Ca1970a}
{\sc Calamai, G.},
On a static criterion for the stability of the equilibrium,
{\em Astrophysics and space Sciences},
{\bf 8}, 53--58 (1970)

\bibitem{Ca1970b}
{\sc Calamai, G.},
The stability of the equilibrium of supermassive and superdense stars,
{\em Astrophysics and space Sciences},
{\bf 6}, 240--257 (1970)

\bibitem{Ca1971}
{\sc Calamai, G.},
A static theory for the stability of the equilibrium,
{\em Astrophysics and space Sciences}, {\bf 10}, 165--177 (1971)

\bibitem{Chandrasekhar1964}
{\sc Chandrasekhar, S.},
A General Variational Principle Governing the Radial and the
Non-Radial Oscillations of Gaseous Masses.
{\em Astrophysical Journal}, {\bf 139}, p.664 (1964)

\bibitem{dafermos}
{\sc Dafermos, M.},
A note on the collapse of small data self--gravitating
massless collisionless matter.
{\em J.\ Hyperbol.\ Differ.\ Equ.}\
{\bf 3}, 589--598 (2006)

\bibitem{DoFeBa}
{\sc Doremus, J.~P., Feix, M.~R., Baumann, G.},
Stability of Encounterless Spherical Stellar Systems.
{\em Phys. Rev. Letts},
{\bf 26}, p. 725 (1971)

\bibitem{FaIpTh}
{\sc Fackerell, E.~D., Ipser, J.~R., Thorne, K.~S.},
Relativistic star clusters.
{\em Comments on Astrophysics and Space Physics,} {\bf 1}, p.134 (1969)

\bibitem{Fa1970}
{\sc Fackerell, E. D.},
Relativistic, Spherically Symmetric Star Clusters IV - a Sufficient
Condition for Instability of Isotropic Clusters against Radial Perturbations.
{\em Astrophysical Journal},
{\bf 160}, p.859 (1970)

\bibitem{FaSu1976a}
{\sc Fackerell, E. D., Suffern, K. G.},
The structure and dynamic instability of isothermal relativistic star clusters.
{\em Australian Journal of Physics},
{\bf 29}, 311--328 (1976)

\bibitem{FaSu1976b}
{\sc Fackerell, E. D., Suffern, K. G.},
The dynamic instability of isothermal relativistic star clusters.
{\em Astrophysical Journal},
{\bf 203}, 477--480 (1976)

\bibitem{FaJoSm}
{\sc Fajman D., Joudioux, J.,  Smulevici, J.},
The Stability of the Minkowski space for the Einstein-Vlasov system.
{\em Preprint. Available on Arxiv at: https://arxiv.org/abs/1707.06141}

\bibitem{FoSc}
{\sc Fournodavlos, G., Schlue, V.},
On `hard stars' in general relativity.
Annales Henri Poincar\'e,
{\bf 20} (7), 2135--2172 (2019)

\bibitem{GaGu}
{\sc Mart\'in-Garc\'ia, J.~M., Gundlach, C.},
Critical Phenomena in Gravitational Collapse.
{\em Living Reviews in Relativity}, 2007

\bibitem{Ge1998}
{\sc G\'erard, P.},
Description du d\'efaut de compacit\'e de l'injection de Sobolev.
ESAIM COCV
{\bf 3}, 213--233 (1998)

\bibitem{Gl2000}
{\sc Glendenning, N.~K.},
{\em Compact stars}, 2nd edition, Springer, 2000

\bibitem{GuRe2001}
{\sc Guo, Y., Rein, G.},
Isotropic steady states in galactic dynamics.
{\em Commun.\ Math.\ Phys.}\
{\bf 219}, 607--629 (2001)

\bibitem{GuRe2007}
{\sc Guo, Y., Rein, G.},
A non-variational approach to nonlinear stability in stellar dynamics
applied to the King model.
{\em Commun.\ Math.\ Phys.}\
{\bf 271}, 489--509 (2007)


\bibitem{GuoLin}
{\sc Guo, Y., Lin, Z.}
Unstable and Stable Galaxy Models.
{\em Comm. Math. Phys.},
{\bf 279}, 789--813 (2008)

\bibitem{HaLin}
{\sc Had\v zi\'c, M., Lin, Z.},
Turning Point Principle for Relativistic Stars.
{\em arXiv:2006.09749.}

\bibitem{HaRe2013}
{\sc Had\v zi\'c, M., Rein, G.},
Stability for the spherically symmetric Einstein-Vlasov
system---a coercivity estimate.
{\em Math.\ Proc.\ Camb.\ Phil.\ Soc.}\
{\bf 155}, 529--556 (2013)

\bibitem{HaRe2014}
{\sc Had\v zi\'c, M., Rein, G.},
On the small redshift limit of steady states of the spherically
symmetric Einstein-Vlasov system and their stability.
{\em Math.\ Proc.\ Camb.\ Phil.\ Soc.}\
{\bf 159}, 529--546 (2015)

\bibitem{HaThWaWh}
{\sc Harrison, B.~K., Thorne, K.~S., Wakano, M., Wheeler J.~A.},
{\em Gravitation Theory and gravitational collapse.}
The University of Chicago press, Chicago and London (1965)

\bibitem{HsuLinMakino}
{\sc Hsu, C.-H., Lin, S.-S., Makino, T.},
On spherically symmetric solutions of the relativistic Euler equation.
{\em J.~Differential Equations}\
{\bf 201}, 1--24 (2004)

\bibitem{IT68}
{\sc Ipser, J., Thorne, K.~S.},
Relativistic, spherically symmetric star clusters
I.\ Stability theory for radial perturbations.
{\em Astrophys.\ J.}\
{\bf 154}, 251--270 (1968)

\bibitem{IP1969}
{\sc Ipser, J.},
Relativistic, Spherically Symmetric Star Clusters. II.
Sufficient Conditions for Stability against Radial Perturbations.
{\em Astrophys.\ J.}\
{\bf 156}, 509--527 (1969)

\bibitem{IP1969b}
{\sc Ipser, J.},
Relativistic, Spherically Symmetric Star Clusters. III.
Stability of Compact Isotropic Models.
{\em Astrophys.\ J.}\
{\bf 158}, 17--43 (1969)

\bibitem{IP1980}
{\sc Ipser, J.},
A binding-energy criterion for the dynamical stability of
spherical stellar systems in general relativity.
{\em Astrophysical Journal},
{\bf 238}, 1101--1110 (1980)

\bibitem{KM}
{\sc Kandrup, H., Morrison, P.},
Hamiltonian structure of the Vlasov-Einstein system and the problem
of stability for spherical relativistic star clusters.
{\em Annals of Physics}\
{\bf 225}, 114--166 (1993)

\bibitem{KS}
{\sc Kandrup, H., Sygnet, J. ~F.},
A simple proof of dynamical stability for a class of spherical clusters.
{\em Astrophys.\ J.}\
{\bf 298}, 27--33 (1985)

\bibitem{Kato}
{\sc Kato, T.},
{\em Perturbation Theory for Linear Operators},
Springer-Verlag, Berlin (1980)

\bibitem{KaHoKl1975}
{\sc Katz, J., Horwitz, G., Klapisch, M.},
Thermodynamic stability of relativistic stellar clusters.
{\em Astrophysical Journal},
{\bf 199}, 307--321 (1975)

\bibitem{LaMePe}
{\sc Laval, G., Mercier, P., Pellat, R.},
Necessity of the energy principles for magnetostatic stability.
{\em Nuclear Fusion},
{\bf 5}, 2, p. 165 (1965)

\bibitem{LeMeRa2}
{\sc Lemou, M., Mehats, F., Rapha\"el, P.},
A new variational approach to the stability of gravitational systems.
{\em Commun. Math. Phys.}
{\bf 302}, no 1, 161--224 (2011)

\bibitem{LeMeRa3}
{\sc Lemou, M., Mehats, F., Rapha\"el, P.},
Orbital stability of spherical systems.
{\em Inventiones Math.}\
{\bf 187}, 145--194 (2012)

\bibitem{LinStrauss2008}
{\sc Lin, Z., Strauss, W.},
A sharp stability criterion for the Vlasov-Maxwell system.
{\em Invent. Math.},
{\bf 173}, no. 3, 497--546 (2008)

\bibitem{LinZeng2017}
{\sc Lin, Z., Zeng, C.},
Instability, index theorem, and exponential
trichotomy for Linear Hamiltonian PDEs.
{\em To appear in Memoirs of AMS.}

\bibitem{LinZeng2017b}
{\sc Lin, Z., Zeng, C.,}
Separable Hamiltonian PDEs and
Turning point principle for stability of gaseous stars.
{\em To appear in Comm. Pure. Appl. Math.}

\bibitem{LiTa}
{\sc Lindblad, H., Taylor, M.},
Global stability of Minkowski space for the
Einstein--Vlasov system in the harmonic gauge.
{\em To appear in Archive for Rational Mechanics and Analysis,
Available on ArXiv at: https://arxiv.org/abs/1707.06079}

\bibitem{Ma1998}
{\sc Makino, T.},
On spherically symmetric stellar models in general relativity.
{\em J. Math. Kyoto Univ}.,
{\bf 38} 1, 55--69 (1998)


\bibitem{Ma2016}
{\sc Makino, T.},
On Spherically Symmetric Solutions of the Einstein-Euler Equations.
{\em Kyoto J. Math.}
{\bf 56}, no. 2, 243--282 (2016)

\bibitem{MeltzerThorne1966}
{\sc Meltzer, D.~W., Thorne, K.~S.},
Normal Modes of Radial Pulsation of Stars at the
End Point of Thermonuclear Evolution.
{\em Astrophys.\ J.}\
{\bf 145}, 514--543 (1966)

\bibitem{Mou}
{\sc Mouhot, C.},
Stabilit\'{e} orbitale pour le syst\`{e}me de Vlasov-Poisson gravitationnel,
d'apr\`{e}s Lemou-M\'{e}hats-Rapha\"el, Guo, Lin, Rein et al.
{\em S\'{e}minaire Nicolas Bourbaki Nov.~2011}, arXiv:1201.2275 (2012)


\bibitem{OpVo1939}
{\sc Oppenheimer J. R., Volkoff, G. M.},
On Massive Neutron Cores.
{\em Physical Review},
{\bf 55}, 374--381 (1939)

\bibitem{RaRe}
{\sc Ramming, T., Rein, G.},
Spherically symmetric equilibria for self--grav\-itating
kinetic or fluid models in the non-relativistic and
relativistic case---A simple proof for finite extension.
{\em SIAM Journal on Mathematical Analysis},  {\bf 45}, 900--914 (2013)


\bibitem{RaShTe1989}
{\sc Rasio, F. A., Shapiro, S. L., Teukolsky, S. A.},
Solving the Vlasov equation in general relativity.
{\em Astrophys. J.}
{\bf 344},
146--157 (1989)

\bibitem{Rein94}
{\sc Rein, G.},
Static solutions of the spherically symmetric Vlasov-Einstein
system.
{\em Math.\ Proc.\ Camb.\ Phil.\ Soc.}\
{\bf 115}, 559--570 (1994)

\bibitem{Rein95}
{\sc Rein, G.},
{\em\,\,The Vlasov-Einstein System with Surface Symmetry},
Habilitations\-schrift, M\"unchen 1995

\bibitem{Rein07}
{\sc Rein, G.},
Collisionless kinetic equations from astrophysics---The Vlasov-Poisson
system.
In {\em Handbook of Differential Equations, Evolutionary Equations,
vol.~3}, edited by C.~M.~Dafermos and E.~Feireisl,
Elsevier (2007)

\bibitem{RR92a}
{\sc Rein, G., Rendall, A.},
Global existence of solutions of the spherically symmetric
Vlasov-Einstein system with small initial data.
{\em Commun.\ Math.\ Phys.}\ {\bf 150}, 561--583 (1992).
Erratum: {\em Commun.\ Math.\ Phys.}\ {\bf 176}, 475--478 (1996)

\bibitem{RR00}
{\sc Rein, G., Rendall, A.},
Compact support of spherically symmetric equilibria in
non-relativistic and relativistic galactic dynamics.
{\em Math.\ Proc.\ Camb.\ Phil.\ Soc.}\
{\bf 128}, 363--380 (2000)

\bibitem{RStr}
{\sc Rein, G., Straub, C.},
On the transport operators arising from linearizing the Vlasov-Poisson or
Einstein-Vlasov system about isotropic steady states.
{\em Kinetic \& Related Models}, to appear

\bibitem{Ru}
{\sc Rudin, W.},
{\em Functional Analysis.}
International Series in Pure \& Applied Mathematics, 2nd edition, 1990

\bibitem{Sch}
{\sc Schaeffer, J.},
A class of counterexamples to Jeans' Theorem
for the Vlasov-Einstein system.
{\em Commun.\ Math.\ Phys.}\ {\bf 204}, 313--327 (1999)

\bibitem{ScWa2014}
{\sc Schiffrin, J.~S., Wald, R.~M.},
Turning point instabilities for relativistic stars and black holes.
{\em Classical and Quantum Gravity}
{\bf 31} (2014)

\bibitem{ShTe1985}
{\sc Shapiro, S. L., Teukolsky, S. A.},
Relativistic Stellar Dynamics on the Computer--Part Two--Physical Applications.
{\em Astrophys. J.}
{\bf 298}, 58--79 (1985)

\bibitem{So1981}
{\sc Sorkin, R.},
A Criterion for the Onset of Instability at a Turning Point.
{\em Astrophysical Journal},
{\bf 249},  p. 254 (1981)

\bibitem{Straumann}
{\sc Straumann, N.},
{\em General Relativity.}
Graduate Texts in Physics, Springer (2013)

\bibitem{taylor}
{\sc Taylor, M.},
The global nonlinear stability of Minkowski space
for the massless Einstein-Vlasov system.
{\em Annals of PDE}, {\bf 3}:9 (2017)


\bibitem{Th1966}
{\sc Thorne, K. S.},
The General-Relativistic Theory of Stellar Structure and Dynamics.
{\em Proceedings of the International School of Physics "Enrico Fermi,"
  Course XXXV, at Varenna, Italy, July 12-24} , ed. L. Gratton,
Academic Press, New York, 166--280, (1966)

\bibitem{Wo}
{\sc Wolansky, G.},
Static solutions of the Vlasov-Einstein system.
{\em Arch.\ Rational Mech.\ Anal.}\ {\bf 156}, 205--230 (2001)


\bibitem{Yosida}
{\sc Yosida, K.},
{\em Functional Analysis.}
Grundlehren der mathematischen Wissenschaften, Springer (1974)

\bibitem{Ze1963}
{\sc Zel’dovich, Ya. B.:} Hydrodynamical stability of star.
{\em Voprosy Kosmogonii},
{\bf 9},
157--170 (1963)


\bibitem{ZeNo}
{\sc Zel'dovich, Ya.~B. , Novikov, I.~D.},
{\em Relativistic Astrophysics}
{\bf Vol. 1},
Chicago: Chicago University Press (1971)

\bibitem{ZePo}
{\sc Zel'dovich, Ya.~B., Podurets, M.~A.},
The evolution of a system of gravitationally interacting point masses.
{\em Soviet Astronomy---AJ}
{\bf 9}, 742--749 (1965),
translated from {\em Astronomicheskii Zhurnal,}~{\bf 42}

\end{thebibliography}
\end{document}